%% file: SuperresolutionJournal.tex
\documentclass[12pt]{article}

\usepackage[top=2.5cm, bottom=3.5cm, left=2.5cm, right=2.5cm]{geometry}
\usepackage{setspace, dsfont}
\usepackage{amssymb, amsmath, mathrsfs, amsthm, stmaryrd, upgreek, mathtools}
\usepackage{pgfplots, caption, subfigure, graphicx}
\usepackage{ifthen}
\usepackage{url, hyperref, cite}
\usepackage{titlesec}
\usepackage{MyMnSymbol}
\usepackage{atbegshi}
\AtBeginDocument{\AtBeginShipoutNext{\AtBeginShipoutDiscard}}

\newcommand{\versionColor}{true}

\newcommand{\pdfFigs}{true}

\newcommand{\generateFigs}{no}

\ifthenelse{\equal{\generateFigs}{yes}}{
	\renewcommand{\pdfFigs}{false}
	\usepgfplotslibrary{external}
	\tikzexternalize[prefix=figs/]
}{}

\hypersetup{
	bookmarks = true,
	pdffitwindow = true,
	colorlinks = true,
	linkcolor = black,
	urlcolor = black,
	citecolor = black,
	pdftitle = {Super-Resolution from Short-Time Fourier Transform Measurements},
	pdfauthor = {Celine Aubel, David Stotz, Helmut Boelcskei},
	pdfsubject = {Journal version},
	pdfkeywords = {Super-resolution, inverse problems in measure spaces, short-time Fourier transform},
	pdfcreator = {LaTeX}
}

\AtBeginDocument{
  \addtolength\abovedisplayskip{-2.7pt}
  \addtolength\belowdisplayskip{-2.7pt}
}

\input{def}

\begin{document}

\title{A Theory of Super-Resolution from \\ Short-Time Fourier Transform Measurements}

\author{{\normalsize C\'eline Aubel, David Stotz, and Helmut B\"olcskei} \\[0.1cm]
{\normalsize Dept.~IT~\&~EE, ETH Zurich, Switzerland} \\[-0.2cm]
{\normalsize Email: \{aubelc, dstotz, boelcskei\}@nari.ee.ethz.ch}
\thanks{Part of the material in this paper was presented at the IEEE International Conference on Acoustics, Speech, and Signal Processing (ICASSP), Florence, Italy, May 2014 \cite{Aubel2014}.}
}

\date{}

\renewcommand\footnotemark{}
 
\maketitle

\input{abstract}
\input{introduction}

\input{notation}
\input{problemStatement}

\input{previousWork}
\input{reconstructionCompleteMeasurements}

\input{reconstructionIncompleteMeasurements}
\input{recoveryForTorus}
\input{simulations}

\appendix 
\input{proofMainResultReal}

\input{proofMainResultTorus}

\newpage
\renewcommand{\abstractname}{Acknowledgments}
\begin{abstract}
	The authors are indebted to H.~G.~Feichtinger for valuable comments, in particular, for pointing out an error in an earlier version of the manuscript, J.-P.~Kahane for answering questions on results in \cite{Kahane2011}, M.~Lerjen for his technical support with the numerical results, and C.~Chenot for inspiring discussions. We also acknowledge the detailed and insightful comments of the anonymous reviewers.
\end{abstract}

\bibliographystyle{IEEEtran} 
\bibliography{ref}

\end{document}

%% file: def.tex
\newtheorem{thm}{Theorem}
\newtheorem{cor}[thm]{Corollary}   
\newtheorem{prop}{Proposition}
\newtheorem{defn}{Definition}

\newtheorem{lem}[thm]{Lemma}

\DeclareMathOperator{\sinc}{sinc}
\DeclareMathOperator{\real}{Re}
\DeclareMathOperator{\imag}{Im}
\renewcommand{\Re}[1]{\real\left\{#1\right\}}
\renewcommand{\Im}[1]{\imag\left\{#1\right\}}
\DeclareMathOperator{\supp}{supp}
\DeclareMathOperator{\cont}{cont}
\DeclareMathOperator{\dom}{dom}
\newcommand{\abs}[1]{\left|#1\right|}
\DeclareMathOperator*{\argmin}{arg\,min}

\DeclareMathOperator{\erf}{erf}
\DeclareMathOperator{\tr}{Tr}

\renewcommand{\triangleq}{\vcentcolon=}
\newcommand{\reversetriangleq}{=\vcentcolon}

\DeclareMathOperator*{\minimize}{minimize\ }
\DeclareMathOperator*{\maximize}{maximize\ }
\DeclareMathOperator{\st}{\ subject\ to\ }
\DeclareMathOperator{\fourierProb}{(SR)}
\DeclareMathOperator{\dualFourierProb}{(D-SR)}
\DeclareMathOperator{\gaborProb}{(STFT-SR)}
\DeclareMathOperator{\gaborProbg}{STFT-SR}
\newcommand{\gaborProbN}{(\gaborProbg_N)}
\DeclareMathOperator{\dualGaborProb}{(PD-STFT-SR)}
\DeclareMathOperator{\dualGaborProbg}{PD-STFT-SR}
\newcommand{\dualGaborProbN}{(\dualGaborProbg_N)}

\newcommand{\C}{\mathbb{C}}
\newcommand{\R}{\mathbb{R}}
\newcommand{\Z}{\mathbb{Z}}
\newcommand{\N}{\mathbb{N}}
\newcommand{\torus}{\mathbb{T}}

\renewcommand{\leq}{\leqslant}
\renewcommand{\geq}{\geqslant}

\newcommand{\meas}{\mu}
\newcommand{\measProb}{\nu}
\newcommand{\measOpt}{\measProb_0}
\newcommand{\measErr}{h}
\newcommand{\measSpace}[1]{\mathcal{M}(#1)}
\newcommand{\group}{G}
\newcommand{\dualGroup}{{\widehat{G}}}
\newcommand{\freq}{f}
\newcommand{\borelAlgebra}[1]{\mathcal{B}(#1)}
\newcommand{\measSupport}{T}
\newcommand{\measSupportIdx}{\Omega}
\newcommand{\measAmplitude}[1]{a_{#1}}
\newcommand{\measSpikeIdx}{\ell}
\newcommand{\measSpikeIdxBis}{m}
\newcommand{\dmeas}[2]{#1\langle#2\rangle}

\newcommand{\innerProd}[2]{\left\langle #1, #2 \right\rangle}
\DeclareMathOperator*{\tv}{TV}
\newcommand{\normTV}[1]{\left\|#1\right\|_{\tv}}
\newcommand{\normOp}[1]{\left\|#1\right\|}
\newcommand{\normInfty}[1]{\left\|#1\right\|_{\infty}}

\newcommand{\normLtwo}[1]{\left\|#1\right\|_{L^2}}
\newcommand{\normLone}[1]{\left\|#1\right\|_{L^1}}
\newcommand{\normLinfty}[1]{\left\|#1\right\|_{L^\infty}}
\newcommand{\normlinfty}[1]{\left\|#1\right\|_{\ell^\infty}}

\newcommand{\measurements}{y}
\newcommand{\gaborOp}{\mathcal{A}_\window}
\newcommand{\window}{g}
\newcommand{\autocorrelation}{R}
\newcommand{\STFT}{\mathcal{V}_\window}
\newcommand{\schwartzSpace}[1]{\mathcal{S}(#1)}

%% file: abstract.tex

\begin{abstract}
	While spike trains are obviously not band-limited, the theory of super-resolution tells us that perfect recovery of unknown spike locations and weights from low-pass Fourier transform measurements is possible provided that the minimum spacing, $\Delta$, between spikes is not too small. Specifically, for a measurement cutoff frequency of $\freq_c$, Donoho~\cite{Donoho1991} showed that exact recovery is possible if the spikes (on $\R$) lie on a lattice and $\Delta > 1/\freq_c$, but does not specify a corresponding recovery method.
Cand\`es and Fernandez-Granda~\cite{Candes2012, FernandezGranda2015} provide a convex programming method for the recovery of periodic spike trains (i.e., spike trains on the torus $\torus$), which succeeds provably if $\Delta > 2/\freq_c$ and $\freq_c \geq 128$ or if $\Delta > 1.26/\freq_c$ and $\freq_c \geq 10^3$, and does not need the spikes within the fundamental period to lie on a lattice. 
In this paper, we develop a theory of super-resolution from short-time Fourier transform (STFT) measurements. Specifically, we present a recovery method similar in spirit to the one in~\cite{Candes2012} for pure Fourier measurements. For a STFT Gaussian window function of width $\sigma = 1/(4\freq_c)$ this method succeeds provably if $\Delta > 1/\freq_c$, without restrictions on $\freq_c$. Our theory is based on a measure-theoretic formulation of the recovery problem, which leads to considerable generality in the sense of the results being grid-free and applying to spike trains on both $\R$ and $\torus$. The case of spike trains on $\R$ comes with significant technical challenges. For recovery of spike trains on $\torus$ we prove that the correct solution can be approximated---in weak-* topology---by solving a sequence of finite-dimensional convex programming problems.
\end{abstract}

\bigskip


\medskip


\medskip

%% file: introduction.tex

\section{Introduction}
\label{section: introduction}

The recovery of spike trains with unknown spike locations and weights from low-pass Fourier measurements, commonly referred to as super-resolution, has been a topic of long-standing interest~\cite{Logan1965, Logan1977, Beurling1966, Beurling1989-1, Beurling1989-2, Donoho1992, Vetterli2002, Dragotti2007}, with recent focus on $\ell^1$-minimization-based recovery techniques~\cite{Tang2012, Candes2012}. 
It was recognized in~\cite{Donoho1991, Bredies2012, Castro2012, Candes2012, Duval2013} that a measure-theoretic formulation of the super-resolution problem in continuous time not only leads to a clean mathematical framework, but also to results that are ``grid-free''~\cite{Tang2012}, that is, the spike locations are not required to lie on a lattice. This theory is inspired by Beurling's seminal work on the ``balayage of measures'' in Fourier transforms~\cite{Beurling1966, Beurling1989-1} and on interpolation using entire functions of exponential type~\cite{Beurling1989-2}. Specifically, de Castro and Gamboa~\cite{Castro2012} and Cand\`es and Fernandez-Granda~\cite{Candes2012} propose to recover a periodic discrete measure (modeling the spike train), that is, a measure on the torus $\torus$, from low-pass Fourier measurements by solving a total variation minimization problem. Despite its infinite-dimensional nature this optimization problem can be solved explicitly, as described in~\cite{Candes2012, Duval2013, Bredies2012}. Concretely, it is shown in~\cite{Castro2012, Candes2012} that the analysis of the Fenchel dual problem leads to an interpolation problem, which can be solved explicitly provided that the elements in the support set of the discrete measure to be recovered are separated by at least $2/\freq_c$, where $\freq_c$ is the cutoff frequency of the low-pass measurements, and $\freq_c \geq 128$. More recently, Fernandez-Granda~\cite{FernandezGranda2015} improved the minimum distance condition to $\Delta > 1.26/\freq_c$, but had to impose the additional constraint $\freq_c \geq 10^3$. Donoho~\cite{Donoho1991} considers the recovery of a spike train on $\R$ and proves that a separation of $1/\freq_c$ is sufficient for perfect recovery provided that the spikes lie on a lattice; a concrete method for reconstructing the measure is, however, not provided. In a different context, Kahane~\cite{Kahane2011} showed that recovery of spike trains on $\R$ under a lattice constraint can be accomplished by solving an infinite-dimensional minimal extrapolation problem, provided that the minimum separation between spikes is at least $\frac{5}{\freq_c}\sqrt{\log(1/\freq_c)}$.
In \cite{Dragotti2007, Vetterli2002} the recovery of periodic spike trains is considered in the context of sampling of signals with finite rate of innovation. The main result in \cite{Vetterli2002} states that a periodic spike train with $K$ spikes per period can be recovered from $2K+1$ Fourier series coefficients without imposing a minimum separation condition. The corresponding recovery algorithm falls into the category of subspace-based methods such as, e.g., MUltiple SIgnal Classification (MUSIC)~\cite{Schmidt1986} and Estimation of Signal Parameters via Rotational Invariance Techniques (ESPRIT)~\cite{Roy1986}, algorithms widely used for direction-of-arrival estimation in array processing.

\subsection{Contributions} In practical applications signals are often partitioned into short time segments and windowed for acquisition (as done, e.g., in spectrum analyzers) so that one has access to windowed Fourier transform, i.e., short-time Fourier transform (STFT), measurements only. Moreover, the frequency characteristics of the spike train to be recovered often vary over time, i.e., the spike locations can be more packed in certain time intervals, as illustrated in Fig.~\ref{fig: example spike to recover}. Time-localized spectral information, as provided by the STFT, can therefore be expected to lead to improved reconstruction quality for the same measurement band-limitation. 
This motivates the development of a theory of super-resolution from STFT measurements, which is the goal of the present paper. Inspired by~\cite{Donoho1991, Bredies2012, Castro2012, Candes2012, Duval2013, AuYeung2015}, we consider the continuous time case and employ a measure-theoretic formulation of the recovery problem. Our main result shows that exact recovery through total variation minimization is possible, for a Gaussian STFT window function of width $\sigma = 1/(4\freq_c)$, provided that the minimum spacing between spikes, i.e., the elements in the support set of the discrete measure to be recovered, exceeds $1/\freq_c$. While our theory applies to general window functions---from the Schwartz-Bruhat space---that are extendable to entire functions, the recovery condition $\Delta > 1/\freq_c$ is obtained by particularizing to a Gaussian window function of width $\sigma = 1/(4\freq_c)$. Similar recovery thresholds can be obtained for other window functions and for Gaussian window functions of different widths, but this would require adapting the computational parts of our proofs, in particular Appendices~\ref{section: proof main result} and \ref{sec: proof torus theorem exact recovery}.
Our theory applies to spike trains on $\R$ and to periodic spike trains, i.e., spike trains on $\torus$, and we do not need to impose a lattice constraint on the spike locations. The case of general (i.e., without a lattice constraint) spike trains on $\R$ comes with significant technical challenges. For recovery of spike trains on $\torus$ we prove that the correct solution can be approximated---in weak-* topology---by solving a sequence of finite-dimensional convex programming problems. 
We finally present numerical results, which demonstrate an improvement for recovery from STFT measurements over recovery from pure Fourier measurements in the sense of the minimum spacing between spikes being allowed to be smaller for the same cutoff frequency. This improvement comes, however, at the cost of increased computational complexity owing to the redundancy in the STFT, which leads to an increased number of measurements and thereby a larger optimization problem size.

\pgfplotsset{
    dirac/.style={
        mark=triangle*,
        mark options={solid,scale=1},
        ycomb,
        scatter,
        visualization depends on={y/abs(y)-1 \as \sign},
        scatter/@pre marker code/.code={\scope[rotate=90*\sign,yshift=-2pt]}
    }
}

\begin{figure}
	\centering
	\ifthenelse{\equal{\versionColor}{true}}{
		\ifthenelse{\equal{\pdfFigs}{true}}{
			\includegraphics[width = 0.85\textwidth]{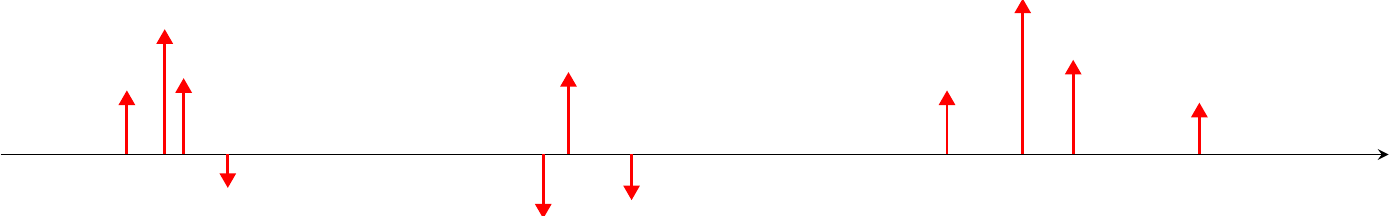}
		}{
            		\ifthenelse{\equal{\generateFigs}{yes}}{\tikzsetnextfilename{exampleSpikeTrainColor}}{}
                    	\begin{tikzpicture}
                    		\tikzstyle{every node}=[font=\footnotesize]
                    		\begin{axis}[
                    			axis x line = middle, axis y line = none,
                    			xmin = -0.05,xmax = 1.05,
                    			ymin = -1, ymax = 3.5, 
                    			grid = none, 
                    			height = 2.8cm, width = 0.85\textwidth, 
                    			scale only axis, 
                    			enlarge x limits = false, enlarge y limits = false,
                    			ytick = \empty, 
                    			xtick = \empty, 
                    			legend columns = 2, legend cell align  = left, legend pos = north west,
                    			]
                    			\addplot+[thick, dirac, red] file {data/exampleMeasure/measure.dat};
                    		\end{axis}
                    	\end{tikzpicture}
		}
	}{
		\ifthenelse{\equal{\pdfFigs}{true}}{
			\includegraphics[width = 0.85\textwidth]{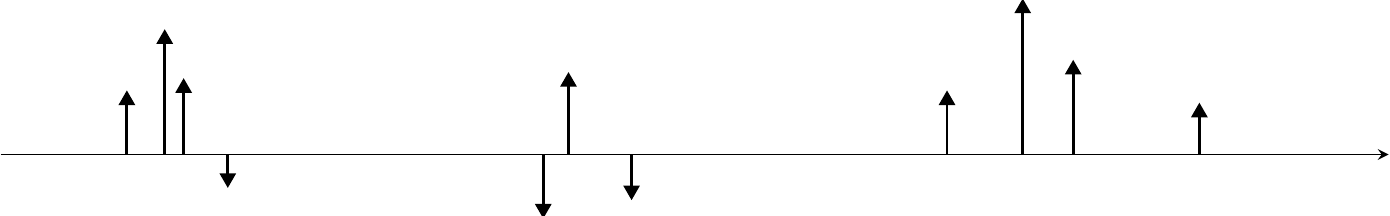}
		}{
            		\ifthenelse{\equal{\generateFigs}{yes}}{\tikzsetnextfilename{exampleSpikeTrainBW}}{}
                        	\begin{tikzpicture}
                        		\tikzstyle{every node}=[font=\footnotesize]
                        		\begin{axis}[
                        			axis x line = middle, axis y line = none,
                        			xmin = -0.05,xmax = 1.05,
                        			ymin = -1, ymax = 3.5, 
                        			grid = none, 
                        			height = 2.8cm, width = 0.85\textwidth, 
                        			scale only axis, 
                        			enlarge x limits = false, enlarge y limits = false,
                        			ytick = \empty, 
                        			xtick = \empty, 
                        			legend columns = 2, legend cell align  = left, legend pos = north west,
                        			]
                        			\addplot+[thick, dirac, black] file {data/exampleMeasure/measure.dat};
                        		\end{axis}
                        	\end{tikzpicture}
		}
	}
	\caption{Example of a spike train. The spikes are more densely packed in certain intervals than in others.}
	\label{fig: example spike to recover}
\end{figure}
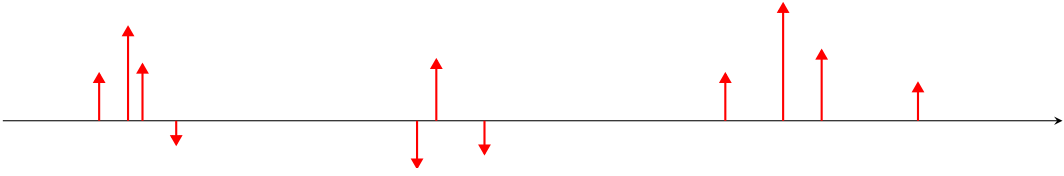

%% file: notation.tex

\subsection{Notation and preparatory material} The complex conjugate of $z \in \C$ is denoted by $\overline{z}$. The first and second derivatives of the function $\varphi$ are designated by $\varphi'$ and $\varphi''$, respectively, for $n \in \N\!\setminus\! \{0, 1, 2\}$, its $n$th derivative is written as $\varphi^{(n)}$. The $\sinc$ function is defined as $\sinc(t) \triangleq \sin(t)/t$, $t \neq 0$, and $\sinc(0) \triangleq 1$. $\delta_{0, \ell}$ is the Kronecker delta, i.e., $\delta_{0, \ell} = 1$ for $\ell = 0$ and $\delta_{0, \ell} = 0$, else.
Uppercase boldface letters stand for matrices. The entry in the $k$th row and $\ell$th column of the matrix $\mathbf{M}$ is $m_{k,\ell}$. The superscript $^H$ denotes Hermitian transposition. 
We define the real inner product of the matrices $\mathbf{X}, \mathbf{Y} \in \C^{M \times N}$ as $\innerProd{\mathbf{X}}{\mathbf{Y}} \triangleq \Re{\tr(\mathbf{Y}^H\mathbf{X})}$. 
For a finite or countable set $\Omega$, $\ell^\infty(\Omega)$ denotes the space of bounded sequences $\alpha = \{\alpha_\ell\}_{\ell \in \Omega}$, with corresponding norm $\normlinfty{\alpha} = \sup_{\ell \in \Omega} \abs{\alpha_\ell}$.
Linear operators are designated by uppercase calligraphic letters. For $u \in \R$, $\mathcal{T}_u$ denotes the translation operator, i.e., $(\mathcal{T}_u\varphi)(t) \triangleq \varphi(t - u)$, for all $t \in \R$. For $\freq \in \R$, $\mathcal{M}_\freq$ is the modulation operator, i.e., $(\mathcal{M}_\freq\varphi)(t) \triangleq \varphi(t)e^{2\pi i\freq t}$, for all $t \in \R$.
Let $X$ and $Y$ be topological vector spaces, and $X^*$ and $Y^*$ their topological duals. The adjoint of the linear operator $\mathcal{L} \colon X \rightarrow Y$ is denoted by $\mathcal{L}^* \colon Y^* \rightarrow X^*$. For an extended-valued function $\varphi \colon X \rightarrow \R \cup \{\infty\}$, we use $\dom \varphi$ to denote its domain, i.e., the subset of $X$ where $\varphi$ takes finite value, and $\cont \varphi$ stands for the subset of $X$ where $\varphi$ takes finite value and is continuous. The function $\varphi \colon X \rightarrow \R \cup \{\infty\}$ is lower semi-continuous if, for every $\alpha \in \R$, the set $\{x \in X \colon \varphi(x) \leq \alpha\}$ is closed. Let $\innerProd{x}{x^*}$ be a dual pairing between $x \in X$ and $x^* \in X^*$. If $\varphi \colon X \rightarrow \R \cup \{\infty\}$ is a convex function, its Fenchel convex conjugate is the function $\varphi^* \colon X^* \rightarrow \R \cup \{\infty\}$ defined by $\varphi^*(x^*) \triangleq \sup_{x \in X} \left\{\innerProd{x}{x^*}  - \varphi(x) \right\}$, for all $x^* \in X^*$. If $\varphi$ is not identically equal to $\infty$ and if $x_0 \in X$, $\partial \varphi(x_0)$ denotes the subdifferential of $\varphi$ at the point $x_0$, i.e., $\partial \varphi(x_0) \triangleq \left\{x^* \in X^* \colon \varphi(x) \geq \varphi(x_0) + \innerProd{x - x_0}{x^*}, \text{ for all } x \in X\right\}$.
The set of all solutions of an optimization problem $(\mathrm{P})$ is denoted by $\mathrm{Sol}\{(\mathrm{P})\}$. 
For a measure space $(X, \Sigma, \meas)$ and a measurable function $\varphi \colon X \rightarrow \C$, the integral of $\varphi$ with respect to the measure $\meas$ is written as $\int_X \varphi(x)\dmeas{\meas}{x}$, where we set $\mathrm{d}x \triangleq \dmeas{\lambda}{x}$ if $\lambda$ is the Lebesgue measure. 
For $p \in [1, \infty)$, $L^p(X, \Sigma, \meas)$ denotes the space of measurable functions $\varphi \colon X \rightarrow \C$ such that $\|\varphi\|_{L^p} \triangleq \left(\int_X \abs{\varphi(x)}^p\dmeas{\meas}{x}\right)^{1/p} < \infty$. The space $L^\infty(X, \Sigma, \meas)$ contains all measurable functions $\varphi \colon X \rightarrow \C$ such that $\|\varphi\|_{L^\infty} \triangleq \inf\{C > 0 \colon \abs{\varphi(x)} \leq C, \text{ for }\meas\text{-almost all }x \in X\} < \infty$. 
For $\varphi \in L^p(X, \Sigma, \meas)$ and $\psi \in L^q(X, \Sigma, \meas)$ with $p, q \in [1, \infty]$ satisfying $1/p+1/q = 1$, we define the complex inner product $(\varphi|\psi) \triangleq \int_X \overline{\varphi(t)}\psi(t)\dmeas{\mu}{t}$ and the real inner product $\innerProd{\varphi}{\psi} \triangleq \Re{\int_X \overline{\varphi(x)}\psi(x)\dmeas{\meas}{x}}$.
For a separable locally compact Abelian topological group $\group$ (e.g., the additive group $\R$ or the torus $\torus \triangleq \R/\Z$, both endowed with the natural topology), we write $L^p(G)$ in the particular case where $\Sigma = \mathcal{B}(\group)$ is the Borel $\sigma$-algebra of $G$ and $\meas$ the Haar measure on $\group$. $\schwartzSpace{\group}$ is the space of Schwartz-Bruhat functions~\cite{Osborne1975}. We will not need the entire formalism of Schwartz-Bruhat functions, as we exclusively consider the cases $\group = \R$ and $\group = \torus$. Specifically, $\schwartzSpace{\R}$ is the space of Schwartz functions, i.e., functions $\varphi \colon \R \rightarrow \C$ that are infinitely often differentiable and satisfy $\sup_{t \in \R} \abs{t^m}\abs{\varphi^{(n)}(t)} < \infty$, for all $n, m \in \N$. $\schwartzSpace{\torus}$ is the space of infinitely often differentiable functions. 
A complex-valued bounded finitely additive measure $\meas$ on $\group$ is a function $\meas \colon \borelAlgebra{\group} \rightarrow \C$ such that for all disjoint finite collections $\{B_\ell\}_{\ell = 1}^L$ of sets in $\borelAlgebra{\group}$, 
\begin{equation}
	\meas\left(\bigcup_{\ell = 1}^{L} B_\ell\right) = \sum_{\ell = 1}^{L} \meas(B_\ell), \label{eq: def complex measure series}
\end{equation}
and for every $B \in \borelAlgebra{\group}$ we have $\abs{\meas(B)} < \infty$.
A complex-valued bounded finitely additive measure is said to be regular if for each $B \in \borelAlgebra{\group}$ and each $\varepsilon > 0$, there exists a closed set $F \subseteq B$ and an open set $G \supseteq B$ such that for every $C \in G \setminus F$, $\abs{\meas(C)} < \varepsilon$.
We denote the space of complex-valued regular bounded finitely additive measures on $\group$ by $\measSpace{\group}$.
For $t \in \group$, $\delta_t \in \measSpace{\group}$ designates the Dirac measure at $t$, defined as follows: for $B \in \borelAlgebra{\group}$, $\delta_t(B) = 1$, if $t \in B$, and $\delta_t(B) = 0$, else. The support $\supp(\meas)$ of $\meas \in \measSpace{\group}$ is the largest closed set $C \subseteq \group$ such that for every open set $B \in \borelAlgebra{\group}$ satisfying $B \cap C \neq \emptyset$, it holds that $\meas(B \cap C) \neq 0$.
We define the total variation (TV) of $\meas \in \measSpace{\group}$ as the measure $\abs{\meas}$ satisfying 
\begin{equation*}
	\forall B \in \borelAlgebra{\group}, \quad |\meas|(B) \triangleq \sup_{\pi \in \Pi(B)} \sum_{A \in \pi} |\meas(A)|,
\end{equation*}
where $\Pi(B)$ denotes the set of all finite disjoint partitions of $B$, i.e., the set of all disjoint finite collections $\{B_\ell\}_{\ell = 1}^L$ of sets in $\borelAlgebra{\R}$ such that $B = \bigcup_{\ell = 1}^L B_\ell$. Throughout the paper, we equip the space $\measSpace{\group}$ with the TV norm $\normTV{\meas} \triangleq |\meas|(\group)$.
By \cite[Thm.~IV.6.2]{Dunford1988}, $\measSpace{\group}$ is the dual of the space $C_b(\group)$ of complex-valued bounded continuous functions $\varphi$. $C_b(\group)$ is equipped with the supremum norm $\normInfty{\varphi} = \sup_{t \in \group} |\varphi(t)|$.
By analogy with the real inner product in $L^2(\group)$, we define the complex and real dual pairing of the measure $\meas \in \measSpace{\group}$ and the function $\varphi \in C_b(\group)$ as $(\varphi|\meas) \triangleq \int_\group \overline{\varphi(t)}\dmeas{\mu}{t}$ and $\innerProd{\varphi}{\meas} \triangleq \Re{\int_\group \overline{\varphi(t)}\dmeas{\meas}{t}}$, respectively.
We endow $\measSpace{\group}$ with the weak-* topology~\cite[Chap.~3]{Brezis2010}, i.e., the coarsest topology on $\measSpace{\group}$ for which every linear functional $\mathcal{L}_\varphi \colon \measSpace{\group} \rightarrow \R$ defined by $\meas \mapsto \mathcal{L}_\varphi(\meas) = \innerProd{\varphi}{\meas}$, with $\varphi \in C_b(\group)$, is continuous.
The Pontryagin dual group of $\group$ is denoted as $\dualGroup$, and the Fourier transform of $\mu \in \measSpace{\group}$ is the function $\hat{\mu} \colon \dualGroup \rightarrow \C$ defined by $\hat{\mu}(f) = \displaystyle\int_{\group} e^{-2\pi ift} \dmeas{\mu}{t}$, for all $f \in \dualGroup$. If $\dualGroup$ is equipped with the metric $\abs{\cdot}$, we denote by $B_r(\dualGroup)$ the ball in $\dualGroup$---with respect to the metric $\abs{\cdot}$---that is centered at $0$ and has radius $r$, i.e., $B_r(\dualGroup) \triangleq \{\freq \in \dualGroup \colon \abs{\freq} \leq r\}$.

%% file: problemStatement.tex

\section{The problem statement}
\label{section: statement of the problem}

We first state the recovery problem considered and then discuss its relation to prior work.
Let $\group = \R$ or $\group = \torus$ (for $\group = \R$, we have $\dualGroup = \R$, and for $\group = \torus$, we get $\dualGroup = \Z$). We model the spike train with weight $\measAmplitude{\measSpikeIdx} \in \C\!\setminus\!\{0\}$ attached to the point $t_\measSpikeIdx \in \group$ by a measure in $\measSpace{\group}$ of the form
\begin{equation}
	\meas = \sum_{\measSpikeIdx \in \measSupportIdx} \measAmplitude{\measSpikeIdx} \delta_{t_\measSpikeIdx},
	\label{eq: original measure}
\end{equation}	
where $\measSupportIdx$ is a finite or countably infinite index set.
The measure is supported on the set $\measSupport \triangleq \{t_\measSpikeIdx\}_{\measSpikeIdx \in \measSupportIdx}$, assumed closed and uniformly discrete, i.e., there exists a $\delta > 0$ such that $\abs{t_\ell - t_{\ell'}} \geq \delta$, for all $\measSpikeIdx, \measSpikeIdx' \in \measSupportIdx$. We have $\normTV{\meas} = \sum_{\measSpikeIdx \in \measSupportIdx} \abs{\measAmplitude{\measSpikeIdx}}$. Moreover, since $\meas \in \measSpace{\group}$, we also have $\normTV{\meas} < \infty$.
Henceforth $\meas$ exclusively designates the measure defined in~\eqref{eq: original measure}.

Suppose we have measurements of $\meas$ in the time-frequency domain of the form
\begin{equation*}
	\measurements(\tau, \freq) = (\STFT\meas)(\tau, \freq), \quad \tau \in \group,\ \freq \in B_{\freq_c}(\dualGroup),
\end{equation*}
where $\freq_c$ is the cutoff frequency (when $\group = \torus$, the cutoff frequency $\freq_c$ becomes an integer, which we denote by $K$) and
\begin{equation}
	(\STFT\meas)(\tau, \freq) \triangleq \int_{\group} \overline{\window(t - \tau)}e^{-2\pi i\freq t}\dmeas{\meas}{t}
	\label{eq: STFT measure}
\end{equation}
denotes the STFT~\cite{Groechenig2000} of $\meas$ with respect to the window function $\window$ taken to be a Schwartz-Bruhat function which, in addition, is assumed to be extendable to an entire function (examples of such functions are the Gaussian function, Hermite functions, and the Fourier transform of smooth functions with compact support).

We are interested in recovering the unknown measure~$\meas$ from the STFT measurements $\measurements$ through the following optimization problem:
\begin{equation*}
	\gaborProb \quad \minimize_{\measProb \in \measSpace{\group}} \normTV{\measProb} \st \measurements = \gaborOp\measProb,
\end{equation*}
 where $\gaborOp \colon \measSpace{\group} \rightarrow L^1(\group \times \dualGroup)$ maps $\measProb \in \measSpace{\group}$ to the function $\rho \in L^1(\group \times \dualGroup)$ given by
\begin{equation}
	\forall (\tau, \freq) \in \group \times \dualGroup, \quad \rho(\tau, \freq) = \begin{cases} (\STFT\measProb)(\tau, f), & f \in B_{\freq_c}(\dualGroup) \\
		0, & \text{otherwise}. \end{cases}
	\label{eq: definition gabor operator}
\end{equation}
The idea of minimizing the TV norm to recover $\meas$ from $\measurements$ originates from the seminal work of Beurling~\cite{Beurling1989-2}, who studied so-called ``minimal extrapolation''. 
The reason for hoping that $\gaborProb$ delivers the correct solution lies in an observation made in~\cite{Beurling1989-2,Bredies2012,Candes2012,Duval2013,Tang2012}, which states that the TV norm $\normTV{\cdot}$ acts as an atomic-norm regularizer as a consequence of the extreme points of the unit ball $\{\measProb \in \measSpace{\group} \colon \normTV{\measProb} \leq 1\}$ being given by the Diracs $\delta_x$, $x \in \group$. Minimizing $\normTV{\cdot}$ therefore forces the solution to be discrete much in the same way as minimizing the $\ell^1$-norm for vectors in $\C^N$ enforces sparsity. For more details we refer to \cite{Bredies2012}, which contains a general and rigorous analysis of inverse problems in spaces of measures.

%% file: previousWork.tex

\section{Relation to previous work}
\label{section: previous work}

Super-resolution theory dates back to the pioneering work by Logan~\cite{Logan1965, Logan1977} and by Donoho~\cite{Donoho1991}. Specifically, \cite{Donoho1991} considers recovery of the complex measure
\begin{equation}	
	\meas = \sum_{\measSpikeIdx \in \measSupportIdx} \measAmplitude{\measSpikeIdx} \delta_{t_\measSpikeIdx}
	\label{eq: previous work measure}
\end{equation}
in $\measSpace{\R}$, where $\measAmplitude{\measSpikeIdx} \in \C\!\setminus\!\{0\}$ and $t_\measSpikeIdx \in \alpha\Z$, for all $\measSpikeIdx \in \measSupportIdx$, with $\alpha > 0$, from pure Fourier measurements
\begin{equation*}
	y(\freq) = \hat{\meas}(\freq) \triangleq \int_\R e^{-2\pi i\freq t} \dmeas{\meas}{t} = \sum_{\measSpikeIdx \in \Z} \measAmplitude{\measSpikeIdx} e^{-2\pi i\freq t_\measSpikeIdx}, \qquad f \in [-\freq_c, \freq_c],
\end{equation*}
where $\freq_c$ is the measurement cutoff frequency. The main result in \cite{Donoho1991} is as follows. If the support $\supp(\meas) = \{t_\measSpikeIdx\}_{\measSpikeIdx \in \measSupportIdx}$ of $\meas$ has Beurling density
\begin{equation*}
	D^+(\supp(\meas)) \triangleq \limsup_{r \rightarrow \infty} \frac{n^+(\supp(\meas), r)}{r} < \freq_c,
\end{equation*}
where for $r > 0$, $n^+(\supp(\meas), r)$ denotes the largest number of points of $\supp(\meas)$ contained in the translates of $[0, r]$, then $\measurements$ uniquely characterizes $\meas$ among all discrete complex measures $\measProb \in \measSpace{\R}$ of support $\supp(\measProb) \subseteq \alpha\Z$ with Beurling density strictly less than $\freq_c$. 

In \cite{Candes2012} Cand\`es and Fernandez-Granda deal with the recovery of periodic spike trains of the form \eqref{eq: previous work measure}.
The resulting measure $\meas = \sum_{\measSpikeIdx = 1}^L \measAmplitude{\measSpikeIdx} \delta_{t_\measSpikeIdx}$, with $t_\measSpikeIdx \in \torus$, for all $\measSpikeIdx \in \{1, 2, \ldots, L\}$, is in $\measSpace{\torus}$ (i.e., $\group = \torus$), and the goal is to recover $\meas$ from the band-limited pure Fourier measurements
\begin{equation}
	y_k = \hat{\meas}(k), \qquad k \in \{-K, \ldots, K\},
	\label{eq: measurements in the case of Fourier measurements}
\end{equation}
where $K$ corresponds to the (integer) cutoff frequency. This is accomplished by solving the following optimization problem:
\begin{equation*}
	\fourierProb \quad \minimize_{\measProb \in \measSpace{\torus}} \normTV{\measProb} \st \mathbf{\measurements} = \mathcal{F}\measProb,
\end{equation*}
where $\mathbf{y} \triangleq \{y_k\}_{k = -K}^K \in \C^{2K+1}$ and $\mathcal{F} \colon \measSpace{\torus} \rightarrow \C^{2K+1}$ maps $\measProb \in \measSpace{\torus}$ to the vector $\mathbf{x} \triangleq \{x_k\}_{k = -K}^K \in \C^{2K+1}$ with $x_k \triangleq \hat{\measProb}(k)$.
The main result in \cite{Candes2012} establishes that $\meas$ is the unique solution of the problem $\fourierProb$ if the minimum wrap-around distance
\begin{equation*}
	\Delta \triangleq \min_{n \in \Z} \min_{\substack{1 \leq \measSpikeIdx, \measSpikeIdxBis \leq L \\ \measSpikeIdx \neq \measSpikeIdxBis}} \abs{t_\measSpikeIdx - t_\measSpikeIdxBis + n}
\end{equation*}
between points in $\measSupport = \{t_\measSpikeIdx\}_{\measSpikeIdx = 1}^L$ obeys $\Delta > 2/K$. 
Moreover, it is shown in \cite{Candes2012} that, under $\Delta > 2/K$, the support $T = \{t_\measSpikeIdx\}_{\measSpikeIdx = 1}^L$ of $\meas$ can be recovered by solving the dual problem
\begin{equation*}
	\dualFourierProb\qquad  \maximize_{\mathbf{c} \in \C^{2K+1}} \innerProd{\mathbf{c}}{\mathbf{\measurements}} \st \normInfty{\mathcal{F}^*\mathbf{c}} \leq 1
\end{equation*}
of $\fourierProb$, where the adjoint operator $\mathcal{F}^*$ is given by 
\begin{equation*}
	\forall \mathbf{c} \in \C^{2K+1}, \quad (\mathcal{F}^*\mathbf{c})(t) = \sum_{k = -K}^K c_k e^{2\pi ikt}.
\end{equation*}
The dual problem $\dualFourierProb$ is equivalent to a finite-dimensional semi-definite program. 
The $t_\measSpikeIdx$ are determined by locating the roots of a polynomial of finite degree on the unit circle, built from a solution of $\dualFourierProb$. More recently, Fernandez-Granda~\cite{FernandezGranda2015} improved the minimum distance condition to $\Delta > 1.26/\freq_c$, but had to impose the additional constraint $\freq_c \geq 10^3$.

In summary, while Donoho considers super-resolution of measures on $\R$ with possibly infinitely many $t_\measSpikeIdx$ to be recovered and finds that $\Delta > 1/\freq_c$ is sufficient for exact recovery, he imposes a lattice constraint on the locations of the $t_\measSpikeIdx$ and does not provide a concrete recovery method. Cand\`es and Fernandez-Granda~\cite{Candes2012, FernandezGranda2015} on the other hand provide a concrete recovery method which succeeds provably for $\Delta > 2/\freq_c$ and $\freq_c \geq 128$ or $\Delta > 1.26/\freq_c$ and $\freq_c \geq 10^3$, but is formulated for measures on $\torus$ only and hence needs the number of unknown locations $t_\measSpikeIdx$ to be finite. The main contribution of the present paper is to show that recovery from STFT measurements through convex programming succeeds for $\Delta > 1/\freq_c$, both for measures on $\R$ and $\torus$, without imposing a lattice constraint and without additional assumptions on $\freq_c$. The rigorous treatment of the case of measures on $\R$ comes with significant technical challenges as the number of atoms of the measure to be recovered can be infinite and no lattice constraint is imposed.

A polynomial root-finding algorithm for the recovery of the spike locations $t_\measSpikeIdx \in \torus$ was presented in the context of sampling of signals with finite rate of innovation~\cite{Vetterli2002}. Specifically, the algorithm proposed in \cite{Vetterli2002} can be cast as a subspace algorithm and does not need a minimum spacing constraint; corresponding performance results for the noisy case were reported in \cite{Blu2008, Barbotin2014}. 

Very recently \cite{Benedetto2016} considered a super-resolution problem similar to $\fourierProb$, aimed at the recovery of a complex Radon measure $\mu$ on $\torus^d$, $d \in \N \!\setminus\! \{0\}$, from its Fourier coefficients $\widehat{\mu}(k)$, for $k \in \Lambda$, where $\Lambda$ is an arbitrary finite subset of $\Z^d$. The main result of \cite{Benedetto2016} is an adaptation (to the torus) and a generalization to higher dimensions of Beurling's theorem on minimal extrapolation. On a conceptual level \cite{Benedetto2016} establishes a connection between Beurling's theory of minimal extrapolation~\cite[Thm.~2]{Beurling1989-2} and the work by Cand\`es and Fernandez-Granda on super-resolution~\cite{Candes2012, FernandezGranda2015}.

%% file: reconstructionCompleteMeasurements.tex

\section{Reconstruction from full STFT measurements}
\label{section: reconstruction from complete measurements}

Before considering the problem of recovering $\meas$ from $\freq_c$-band-limited measurements $\measurements$, we need to convince ourselves that reconstruction is possible from STFT measurements with full frequency information, i.e., for $\freq_c = \infty$. 

It is well known that the STFT of a function $\varphi \in L^2(\R)$ with respect to the (nonzero) window function $\window$, defined as
\begin{equation}
	\forall (\tau, \freq) \in \R^2, \quad (\STFT\varphi)(\tau, \freq) \triangleq (\mathcal{M}_f\mathcal{T}_\tau\window|\varphi) = \int_\R \varphi(t)\overline{\window(t - \tau)}e^{-2\pi i\freq t}\mathrm{d}t,
	\label{eq: definition STFT for functions}
\end{equation}
uniquely determines $\varphi$ in the sense of $\STFT{\varphi} = 0$ implying $\varphi = 0$. The signal $\varphi \in L^2(\R)$ can be reconstructed from $\STFT{\varphi}$ via the following inversion formula~\cite{Groechenig2000}:
\begin{equation}
	\varphi = \frac{1}{\normLtwo{\window}^2}\int_\R\int_\R (\STFT\varphi)(\tau, \freq) \mathcal{M}_f\mathcal{T}_\tau\window\ \mathrm{d}\tau\mathrm{d}\freq.
	\label{eq: STFT inversion formula for L^2}
\end{equation}
Since for $\group = \torus$, $\dualGroup = \Z$, the STFT of a function $\varphi \in L^2(\torus)$ is given by
\begin{equation*}
	\forall (\tau, k) \in \torus \times \Z, \quad (\STFT\varphi)(\tau, k) \triangleq (\mathcal{M}_k\mathcal{T}_\tau\window|\varphi) = \int_0^1 \varphi(t)\overline{g(t - \tau)}e^{-2\pi ikt}\mathrm{d}t
\end{equation*}
with the corresponding reconstruction formula
\begin{equation*}
	\varphi =  \frac{1}{\normLtwo{\window}^2}\sum_{k \in \Z}\int_0^1 (\STFT\varphi)(\tau, k)\mathcal{M}_k\mathcal{T}_\tau\window\ \mathrm{d}\tau.
\end{equation*}

Gr\"ochenig~\cite{Groechenig2000} extended the classical definition of the STFT for functions in $L^2(\group)$ to tempered distributions. As measures in $\measSpace{\group}$ are special cases of tempered distributions, this extended definition applies to the present case as well. The remainder of this section is devoted to particularizing the STFT inversion formula \cite[Cor.~11.2.7]{Groechenig2000} to measures in $\measSpace{\group}$. In the process, we obtain a simple and explicit inversion formula for discrete measures, Proposition~\ref{prop: reconstruction of the measure in R}, whose proof will be presented in detail as it provides intuition for the proof of our main result in the case of incomplete measurements.

The STFT of a measure $\measProb \in \measSpace{\group}$ is obtained by replacing the complex inner product on $L^2(\group)$ in \eqref{eq: definition STFT for functions} by the complex dual pairing between measures $\measProb \in \measSpace{\group}$ and functions $\varphi \in C_b(\group)$, that is, $(\varphi|\measProb) \triangleq \int_\group \overline{\varphi(t)} \dmeas{\measProb}{t}$. 

\begin{defn}[STFT of measures]
	\label{defn: STFT of measures}
	The STFT of the measure $\measProb \in \measSpace{\group}$ with window function $\window \in \schwartzSpace{\group} \!\setminus\! \{0\}$ is defined as
	\begin{equation*}
		(\STFT\measProb)(\tau, \freq) \triangleq (\mathcal{M}_\freq\mathcal{T}_\tau\window|\measProb) = \int_\group \overline{\window(t - \tau)}e^{-2\pi i\freq t}\dmeas{\measProb}{t}.
	\end{equation*}
\end{defn}

The following result provides a formula for the reconstruction of $\measProb \in \measSpace{\group}$ from $\STFT\measProb$.

\begin{prop}[STFT inversion formula for measures, {\cite[Cor.~11.2.7]{Groechenig2000}}]
	\label{prop: general inversion formula}
	Let $\measProb \in \measSpace{\group}$ and $\window \in \schwartzSpace{\group} \!\setminus\! \{0\}$ be such that $\normLtwo{\window} = 1$. The quantities
	\begin{equation*}
		\begin{array}{lll}
		&\displaystyle \int_\R\int_\R (\STFT\measProb)(\tau, \freq) \mathcal{M}_\freq\mathcal{T}_\tau\window\ \mathrm{d}\tau\mathrm{d}\freq, & \text{for }\group = \R, \\
		&\displaystyle \sum_{k \in \Z}\int_0^1 (\STFT\measProb)(\tau, k) \mathcal{M}_k\mathcal{T}_\tau\window\  \mathrm{d}\tau, & \text{for }\group = \torus,
		\end{array}
	\end{equation*}
	define measures in $\measSpace{\R}$ and $\measSpace{\torus}$, respectively, in the sense that
	\begin{equation}
		\forall \varphi \in C_b(\group), \qquad (\varphi|\measProb) = \begin{cases} \displaystyle \int_\R\int_\R (\STFT\measProb)(\tau, \freq)(\varphi|\mathcal{M}_\freq\mathcal{T}_\tau\window) \mathrm{d}\tau\mathrm{d}\freq, & \qquad \group = \R \\
			\displaystyle\sum_{k \in \Z} \int_0^1 (\STFT\measProb)(\tau, \freq)(\varphi | \mathcal{M}_k\mathcal{T}_\tau\window) \mathrm{d}\tau, & \qquad \group = \torus.  \end{cases}
			\label{eq: proof injectivity reconstruction}
	\end{equation}
	Moreover,
	\begin{equation*}
		 \measProb = \begin{cases} \displaystyle\int_\R\int_\R (\STFT\measProb)(\tau, \freq) \mathcal{M}_\freq\mathcal{T}_\tau\window\ \mathrm{d}\tau\mathrm{d}\freq, & \qquad \group = \R \\
		\displaystyle\sum_{k \in \Z}\int_0^1 (\STFT\measProb)(\tau, k) \mathcal{M}_k\mathcal{T}_\tau\window\  \mathrm{d}\tau, & \qquad \group = \torus. \end{cases}
	\end{equation*}
	\end{prop}

An immediate consequence of Proposition \ref{prop: general inversion formula} is the injectivity of the STFT of measures. To see this, simply note that $\STFT\measProb = 0$ in \eqref{eq: proof injectivity reconstruction} implies $(\varphi|\measProb) = 0$ for all $\varphi \in C_b(\group)$, which means that necessarily $\measProb = 0$.
Equivalently, for $\nu_1, \nu_2 \in \measSpace{\group}$, as $\measSpace{\group}$ is a vector space and hence $\measProb_1 - \measProb_2 \in \measSpace{\group}$, it follows that $\STFT{\nu_1} = \STFT{\nu_2}$ with $\window \in \schwartzSpace{\group} \!\setminus\! \{0\}$ necessarily implies $\nu_1 = \nu_2$. We can therefore conclude that every measure in $\measSpace{\group}$ is uniquely determined by its STFT.
While the inversion formula in Proposition \ref{prop: general inversion formula} applies to general measures in $\measSpace{\group}$, we can get a simplified inversion formula for discrete measures, as considered here. Specifically, for discrete measures inversion reduces to determining the support set $\measSupport = \{t_\measSpikeIdx\}_{\measSpikeIdx \in \measSupportIdx}$ and the corresponding complex weights $\{\measAmplitude{\measSpikeIdx}\}_{\measSpikeIdx \in \measSupportIdx}$ only.  

\begin{prop}[STFT inversion formula for discrete measures]
	Let $\window \in \schwartzSpace{\group} \!\setminus\! \{0\}$ be such that $\normLtwo{\window} = 1$. Let $\meas \triangleq \sum_{\measSpikeIdx \in \measSupportIdx} \measAmplitude{\measSpikeIdx} \delta_{t_\measSpikeIdx} \in \measSpace{\group}$, where $\measSupportIdx$ is a finite or countably infinite index set and $\measAmplitude{\measSpikeIdx} \in \C\!\setminus\! \{0\}$, for all $\measSpikeIdx \in \measSupportIdx$.  The measure $\meas$ can be recovered from its STFT $\STFT\meas$ according to
	\begin{align}
		\lim_{ F  \rightarrow \infty} \frac{1}{2 F } \int_{- F }^{ F }\int_\R (\STFT\meas)(\tau, \freq)\window(t - \tau)e^{2\pi i\freq t}\mathrm{d}\tau\mathrm{d}\freq &= \begin{cases} \measAmplitude{\measSpikeIdx}, & t = t_\measSpikeIdx \\ 0, & \text{otherwise} \end{cases}, \qquad \group = \R \label{eq: STFT inversion formula for discrete measures}
	\\
		\lim_{K  \rightarrow \infty} \frac{1}{2K+1} \sum_{k=-K}^{K}\int_0^1 (\STFT\meas)(\tau, k)\window(t - \tau)e^{2\pi i k t}\mathrm{d}\tau &= \begin{cases} \measAmplitude{\measSpikeIdx}, & t = t_\measSpikeIdx \\ 0, & \text{otherwise} \end{cases}, \qquad \group = \torus. \label{eq: STFT inversion formula for discrete measures 2}
	\end{align}
	\label{prop: reconstruction of the measure in R}
\end{prop}

\begin{proof}
	We prove the proposition for $\group = \R$ and note that the proof for $\group = \torus$ is similar. Let $t \in \R$. For every $ F >0$, we have
	\begin{align}
		 \frac{1}{2 F } \int_{- F }^{ F }\int_\R (\STFT\meas)(\tau, \freq)&\window(t - \tau)e^{2\pi i\freq t}\mathrm{d}\tau\mathrm{d}\freq \notag \\
		 	&=  \frac{1}{2 F } \int_{- F }^{ F }\int_\R \sum_{\measSpikeIdx \in \measSupportIdx} \measAmplitude{\measSpikeIdx}\overline{\window(t_\measSpikeIdx - \tau)}\window(t - \tau)e^{2\pi i\freq (t - t_\measSpikeIdx)}\mathrm{d}\tau\mathrm{d}\freq \label{eq: } \\
		 	&= \sum_{\measSpikeIdx \in \measSupportIdx} \measAmplitude{\measSpikeIdx} \left(\int_\R \overline{\window(t_\measSpikeIdx - \tau)}\window(t - \tau)\mathrm{d}\tau\right) \left(\frac{1}{2 F }\int_{- F }^{ F } e^{2\pi i\freq(t -t_\measSpikeIdx)}\mathrm{d}\freq\right) \label{eq: integration term by term inversion STFT}\\
			&= \sum_{\measSpikeIdx \in \measSupportIdx} \measAmplitude{\measSpikeIdx}\autocorrelation(t - t_\measSpikeIdx)\sinc(2\pi F (t - t_\measSpikeIdx)), \label{eq: normal convergence inversion STFT}
	\end{align}
	where $\autocorrelation \in \schwartzSpace{\R}$ is the autocorrelation function of $\window$, that is, 
	\begin{equation}
		\forall t \in \R, \quad\autocorrelation(t) \triangleq \int_\R \overline{\window(\tau)}\window(t + \tau)\mathrm{d}\tau.
		\label{eq: autocorrelation function}
	\end{equation}
	Thanks to Fubini's Theorem, the order of the integral on $[-F,F] \times \R$ and the sum in \eqref{eq: } can be interchanged as, using the Cauchy-Schwarz inequality, we have
	\begin{equation*}
		\sum_{\measSpikeIdx \in \measSupportIdx} \frac{\abs{\measAmplitude{\measSpikeIdx}}}{2F} \int_{- F }^{ F }\int_\R  \abs{\overline{\window(t_\measSpikeIdx - \tau)}\window(t - \tau) e^{2\pi i\freq(t - t_\measSpikeIdx)}} \mathrm{d}\tau\mathrm{d}\freq \leq \normLtwo{\window}^2\sum_{\measSpikeIdx \in \measSupportIdx} \abs{\measAmplitude{\measSpikeIdx}} < \infty,
	\end{equation*}
	since $\window \in \schwartzSpace{\R} \subseteq L^2(\R)$, and, by assumption, $\sum_{\measSpikeIdx \in \measSupportIdx} \abs{\measAmplitude{\measSpikeIdx}} < \infty$. Now we want to let $F \rightarrow \infty$ in~\eqref{eq: normal convergence inversion STFT}. First note that
	\begin{equation*}
		\lim_{ F  \rightarrow \infty} \sinc(2\pi F (t - t_\measSpikeIdx)) = \begin{cases} 1, & t = t_\measSpikeIdx \\ 0, & \text{otherwise}. \end{cases}
	\end{equation*}
	For the limit $F \rightarrow \infty$ of the series in~\eqref{eq: normal convergence inversion STFT} to equal the sum of the limits of the individual terms, we need to show that the sum converges uniformly in $F$. This can be accomplished through the Weierstrass M-test.
	Specifically, we have
	\begin{equation*}
		\forall F > 0, \quad \abs{\measAmplitude{\measSpikeIdx}\autocorrelation(t - t_\measSpikeIdx)\sinc(2\pi F(t - t_\measSpikeIdx))} \leq \abs{\measAmplitude{\measSpikeIdx}}\abs{\autocorrelation(t - t_\measSpikeIdx)}
	\end{equation*}
	and since $\autocorrelation \in \schwartzSpace{\R}$ and hence $\normInfty{\autocorrelation} < \infty$, it holds that
	\begin{equation*}
		\sum_{\measSpikeIdx \in \measSupportIdx} \abs{\measAmplitude{\measSpikeIdx}}\abs{\autocorrelation(t - t_\measSpikeIdx)} \leq \normInfty{\autocorrelation}\sum_{\measSpikeIdx \in \measSupportIdx} \abs{\measAmplitude{\measSpikeIdx}} < \infty,
	\end{equation*}
	as $\sum_{\measSpikeIdx \in \measSupportIdx} \abs{\measAmplitude{\measSpikeIdx}} < \infty$ by assumption. This allows us to conclude that
	\begin{align*}
		\lim_{ F  \rightarrow \infty} \frac{1}{2 F } \int_{- F }^{ F }& \int_\R (\STFT\meas)(\tau, \freq)\window(t - \tau)e^{2\pi i\freq t}\mathrm{d}\tau\mathrm{d}\freq = \lim_{F \rightarrow \infty} \sum_{\measSpikeIdx \in \measSupportIdx} \measAmplitude{\measSpikeIdx}\autocorrelation(t - t_\measSpikeIdx)\sinc(2\pi F(t - t_\measSpikeIdx)) \\
			&= \sum_{\measSpikeIdx \in \measSupportIdx}\measAmplitude{\measSpikeIdx}\autocorrelation(t - t_\measSpikeIdx) \lim_{F \rightarrow \infty} \sinc(2\pi F(t - t_\measSpikeIdx)) = \begin{cases}  \measAmplitude{\measSpikeIdx}, & t = t_{\measSpikeIdx} \\ 0, & t \in \R \setminus \measSupport, \end{cases}
	\end{align*}
	where the last equality follows from $\autocorrelation(0) = 1$.
\end{proof}

%% file: reconstructionIncompleteMeasurements.tex

\section{Reconstruction from incomplete measurements}
\label{section: reconstruction from partial measurements}

We are now ready to consider the main problem statement of this paper, namely the reconstruction of $\meas$ from band-limited (i.e., $\freq_c < \infty$) STFT measurements via $\gaborProb$. Henceforth, we assume $\freq_c < \infty$. 

We first verify that the range of the operator $\gaborOp$ defined in \eqref{eq: definition gabor operator} is indeed $L^1(\group \times \dualGroup)$, and then show that $\gaborOp$ is bounded.
\begin{lem}[Range of $\gaborOp$]
	Let $\window \in \schwartzSpace{\group} \!\setminus\! \{0\}$. 
	$\gaborOp\measProb \in L^1(\group \times \dualGroup)$ for $\measProb \in \measSpace{\group}$.
\end{lem}
\begin{proof}
	We provide the proof for $\group = \R$ only, the proof for $\group = \torus$ again being essentially identical.
	Let $\measProb \in \measSpace{\R}$. We have the following:
	\begin{align}
		\int_\R\int_\R \abs{(\gaborOp\measProb)(\tau, \freq)}\mathrm{d}\tau\mathrm{d}\freq &= \int_{-\freq_c}^{\freq_c}\int_\R \abs{(\STFT\measProb)(\tau, \freq)}\mathrm{d}\tau\mathrm{d}\freq \notag\\
			&= \int_{-\freq_c}^{\freq_c}\left(\int_\R \abs{\int_\R \window(t - \tau)e^{-2\pi i\freq t}\dmeas{\measProb}{t}}\mathrm{d}\tau\right)\mathrm{d}\freq \notag \\
			&\leq \int_{-\freq_c}^{\freq_c} \left[\int_\R \left(\int_\R \abs{\window(t - \tau)e^{-2\pi i\freq t}} \dmeas{\abs{\measProb}\!}{t}\right)\mathrm{d}\tau\right]\mathrm{d}\freq \notag \\
			&= 2\freq_c \int_\R\left(\int_\R \abs{\window(t - \tau)} \dmeas{\abs{\measProb}\!}{t}\right)\mathrm{d}\tau. \label{eq: proof that the operator gaborOp lands in L1}
	\end{align}
	Since $\window \in L^1(\R)$, we have
	\begin{equation*}
		 \int_\R \left(\int_\R \abs{\window(t - \tau)}\mathrm{d}\tau\right)\dmeas{\abs{\measProb}\!}{t} = \normLone{\window}\normTV{\measProb} < \infty.
	\end{equation*}
	Fubini's Theorem then implies that the integral in~\eqref{eq: proof that the operator gaborOp lands in L1} is finite, thereby concluding the proof.
\end{proof}

\begin{lem}
	Let $\window \in \schwartzSpace{\group} \!\setminus\! \{0\}$. 
	The operator $\gaborOp \colon \measSpace{\group} \rightarrow L^1(\group \times \dualGroup)$ is bounded.
\end{lem}

\begin{proof}
	Again, we provide the proof for $\group = \R$ only, the arguments for $\group = \torus$ are essentially identical. For every $\measProb \in \measSpace{\group}$, we have
	\begin{align}
		\normLone{\gaborOp \measProb} &= \int_\R\int_{-\freq_c}^{\freq_c} \abs{(\STFT\measProb)(\tau, \freq)}\mathrm{d}\tau\mathrm{d}\freq = \int_\R\int_{-\freq_c}^{\freq_c} \abs{\int_\R \overline{\window(t - \tau)}e^{-2\pi i\freq t} \dmeas{\measProb}{t}} \notag \\
		&\leq \int_\R\int_{-\freq_c}^{\freq_c} \int_\R \abs{\window(t - \tau)} \dmeas{\abs{\measProb}\!}{t} \mathrm{d}\tau\mathrm{d}\freq = 2\freq_c\int_\R\int_\R \abs{\window(t - \tau)}\dmeas{\abs{\measProb}\!}{t}\mathrm{d}\tau \notag\\
		&= 2\freq_c\int_\R\int_\R \abs{\window(t - \tau)}\mathrm{d}\tau\dmeas{\abs{\measProb}\!}{t} = 2\freq_c\normLone{\window}\normTV{\measProb}, \label{eq: proof that Ag is bounded}
	\end{align}
	where we used Fubini's Theorem in the first equality of \eqref{eq: proof that Ag is bounded}.
\end{proof}

The adjoint of the operator $\gaborOp$ plays a crucial role in the ensuing developments.

\begin{lem}[Adjoint of $\gaborOp$]
	Let ${\window \in \schwartzSpace{\group} \!\setminus\! \{0\}}$. The adjoint operator $\gaborOp^* \colon L^\infty(\group \times \dualGroup) \rightarrow C_b(\group)$ of $\gaborOp$ maps $c \in L^\infty(\group \times \dualGroup)$ to the function
	\begin{equation}
		\begin{array}{ll}
			\displaystyle t \longmapsto \int_{-\freq_c}^{\freq_c} \int_\R c(\tau, \freq)\overline{\window(t - \tau)}e^{2\pi i\freq t}\mathrm{d}\tau\mathrm{d}\freq, & \qquad \group = \R \\
			\displaystyle t \longmapsto \sum_{k = -\freq_c}^{\freq_c}\int_0^1 c(\tau, \freq)\overline{\window(t - \tau)}e^{2\pi ik t}\mathrm{d}\tau, & \qquad \group = \torus.
		\end{array}
		\label{eq: adjoint operator expression}
	\end{equation}
	Moreover, it holds that $\gaborOp^{**} = \gaborOp$.
	\label{lem: lemma for adjoint operator}
\end{lem}

\begin{proof}
	Again, we provide the proof for the case $\group = \R$ only, the arguments for $\group = \torus$ being essentially identical.
	
	First, we note that $L^\infty(\R^2)$ is the topological dual of $L^1(\R^2)$. Therefore, $\gaborOp^*$ maps $L^\infty(\R^2)$ to $(\measSpace{\R})^*$.
	For $\measProb \in \measSpace{\R}$ and $c \in L^\infty(\R^2)$ we have
	\begin{align}
		\innerProd{\gaborOp\measProb}{c} &= \Re{\int_{-\freq_c}^{\freq_c}\int_\R (\gaborOp\measProb)(\tau, \freq)\overline{c(\tau, \freq)}\mathrm{d}\tau\mathrm{d}\freq} \notag \\
			&= \Re{\int_{-\freq_c}^{\freq_c}\int_\R\left(\int_\R \overline{c(\tau, \freq)}\window(t - \tau)e^{-2\pi i\freq t} \dmeas{\measProb}{t}\right)\mathrm{d}\tau\mathrm{d}\freq} \notag \\
			&= \Re{\int_\R\left(\int_{-\freq_c}^{\freq_c}\int_\R \overline{c(\tau, \freq)}\window(t - \tau)e^{-2\pi i\freq t} \mathrm{d}\tau\mathrm{d}\freq\right)\dmeas{\measProb}{t}} \label{eq: proof derivation adjoint} \\
			&= \Re{\int_\R\overline{\left(\int_{-\freq_c}^{\freq_c}\int_\R c(\tau, \freq)\overline{\window(t - \tau)}e^{2\pi i\freq t} \mathrm{d}\tau\mathrm{d}\freq\right)}\dmeas{\measProb}{t}} \notag \\
			&= \innerProd{\measProb}{\gaborOp^*c}, \notag
	\end{align}
	where \eqref{eq: proof derivation adjoint} follows from Fubini's Theorem whose conditions are satisfied since
	\begin{equation}
		\int_\R\int_{-\freq_c}^{\freq_c}\int_\R \abs{c(\tau, f)}\abs{\window(t - \tau)}\mathrm{d}\tau\mathrm{d}f\dmeas{\abs{\measProb}\!}{t} \leq 2\freq_c\normLinfty{c}\normLone{\window}\normTV{\measProb},
		\label{eq: verification conditions fubini adjoint operator plus bidual}
	\end{equation}	
	as a consequence of $c \in L^\infty(\R^2)$, $\window \in \schwartzSpace{\R}$, and $\measProb \in \measSpace{\R}$.
	
	It remains to show that $\gaborOp^*c \in C_b(\R)$ for all $c \in L^\infty(\R^2)$. To this end, we first note that for $c \in L^\infty(\R^2)$, we have 
	\begin{align*}
		\forall t \in \R, \quad (\gaborOp^*c)(t) &= \int_{-\freq_c}^{\freq_c}\int_\R c(\tau, \freq) \overline{\window(t - \tau)} e^{2\pi ift}\mathrm{d}\tau\mathrm{d}\freq \\
			&= \int_{-\freq_c}^{\freq_c}\int_\R c(t - u, \freq) \overline{\window(u)}e^{2\pi i\freq t} \mathrm{d}u \mathrm{d}\freq.
	\end{align*}
	As $\window \in \schwartzSpace{\R}$, the mapping $t \mapsto c(t - u, \freq)\overline{\window(u)}e^{2\pi i\freq t}$ is continuous on $\R$ for all $(u, \freq) \in \R^2$. Moreover, we have
	\begin{equation*}
		\forall t \in \R,\ \forall (u, \freq) \in \R^2, \qquad \abs{c(t - u, \freq) \window(u) e^{2\pi i\freq t}} \leq \normLinfty{c} \abs{\window(u)}
	\end{equation*}
	and $(u, \freq) \mapsto \normLinfty{c} \abs{\window(u)}$ is Lebesgue-integrable on $\R \times [-\freq_c, \freq_c]$. It therefore follows, by application of the Lebesgue Dominated Convergence Theorem, that $\gaborOp^*c$ is continuous on $\R$. Furthermore, we have
	\begin{equation*}
		\forall t \in \R, \quad \abs{(\gaborOp^*c)(t)} \leq \int_{-\freq_c}^{\freq_c}\int_\R \abs{c(\tau, \freq)} \abs{\window(t - \tau)}\mathrm{d}\tau\mathrm{d}\freq  \leq 2\freq_c \normLinfty{c} \normLone{\window} < \infty,
	\end{equation*}
	as $c \in L^\infty(\R^2)$ and $\window \in \schwartzSpace{\R}$. This shows that $\gaborOp^*c$ is bounded and therefore $\gaborOp^*c \in C_b(\R)$, for all $c \in L^\infty(\R^2)$.
	
	We next show that $\gaborOp^{**} = \gaborOp$. As $\gaborOp^* \colon L^\infty(\R^2) \rightarrow C_b(\R)$ and $\measSpace{\R}$ is the topological dual of $C_b(\R)$, we have $\gaborOp^{**} \colon \measSpace{\R} \rightarrow (L^\infty(\R^2))^*$. Let $\measProb \in \measSpace{\R}$ and note that
	\begin{align}
		\innerProd{\gaborOp^*c}{\measProb} &= \Re{\int_\R \overline{\left(\int_{-\freq_c}^{\freq_c}\int_\R c(\tau, \freq) \overline{\window(t - \tau)}e^{2\pi i\freq t}\mathrm{d}\tau\mathrm{d}\freq\right)} \dmeas{\measProb}{t}} \notag \\
			&= \Re{\int_{-\freq_c}^{\freq_c}\int_\R \overline{c(\tau, \freq)}\left(\int_\R\window(t - \tau)e^{-2\pi i\freq t}\dmeas{\measProb}{t}\right) \mathrm{d}\tau\mathrm{d}\freq} \label{eq: fubini equality bidual gaborOp} \\
			&= \Re{\int_\R\int_\R \overline{c(\tau, \freq)}(\gaborOp\measProb)(\tau, \freq) \mathrm{d}\tau\mathrm{d}\freq} \notag \\
			&= \innerProd{c}{\gaborOp^{**}\measProb}, \notag
	\end{align}
	where \eqref{eq: fubini equality bidual gaborOp} follows by Fubini's Theorem whose conditions are satisfied thanks to \eqref{eq: verification conditions fubini adjoint operator plus bidual}. This establishes that $\gaborOp^{**}\measProb = \gaborOp\measProb$, for all $\measProb \in \measSpace{\R}$. Since $\gaborOp\measProb \in L^1(\R^2)$, for all $\measProb \in \measSpace{\R}$, it follows that $\gaborOp^{**} \colon \measSpace{\R} \rightarrow L^1(\R^2)$, which completes the proof.
\end{proof}

Since the space $\measSpace{\group}$ is infinite-dimensional, proving the existence of a solution of $\gaborProb$ is a delicate problem. It turns out, however, that relying  on the convexity of the TV norm $\normTV{\cdot}$ and on the compactness---with respect to the weak-* topology---of the unit ball $\{\measProb \in \measSpace{\group} \colon \normTV{\measProb} \leq 1\}$, the following theorem ensures the existence of a solution of $\gaborProb$.

\begin{thm}[Existence of a solution of $\gaborProb$]
	Let $\meas \triangleq \sum_{\measSpikeIdx \in \measSupportIdx} \measAmplitude{\measSpikeIdx} \delta_{t_\measSpikeIdx} \in \measSpace{\group}$, where $\measSupportIdx$ is a finite or countably infinite index set and $\measAmplitude{\measSpikeIdx} \in \C\!\setminus\! \{0\}$, for all $\measSpikeIdx \in \measSupportIdx$. 
 	Let $\window \in \schwartzSpace{\group} \!\setminus\! \{0\}$. 
	If $\measurements = \gaborOp\meas$, there exists at least one solution of the problem $\gaborProb$.
	\label{thm: existence}
\end{thm}

\begin{proof}
	As $\measurements  = \gaborOp\meas$  the set $N \triangleq \{\measProb \in \measSpace{\group} \colon \measurements = \gaborOp\measProb\}$ is non-empty, and hence, 
	\begin{equation*}
		m \triangleq \inf\left\{\normTV{\measProb} \colon \measProb \in N\right\} < \infty.
	\end{equation*}
	For all $n \in \N\!\setminus\!\{0\}$, we define the intersection of $N$ and the closed ball of radius $m + 1/n$ in $\measSpace{\group}$ according to
	\begin{equation*}
		K_n \triangleq \left\{\nu \in N \colon \normTV{\nu} \leq m+ \frac{1}{n}\right\}.
	\end{equation*}
	 Since $N$ is the preimage of $\{\measurements\}$ under the linear operator $\gaborOp$, $N$ is convex. Further, as $\gaborOp$ is linear and bounded, $\gaborOp$ is continuous with respect to the topology induced by the norm $\normTV{\cdot}$, which implies that $N$ is closed with respect to the norm $\normTV{\cdot}$.
	Therefore, for every $n \in \N \setminus\! \{0\}$, $K_n$ is convex, closed with respect to the norm $\normTV{\cdot}$, and bounded. By application of~\cite[Cor.~3.22]{Brezis2010}, 
	every $K_n$, $n \in \N \!\setminus\! \{0\}$, is compact with respect to the weak-* topology. 
	Since the sets $K_n$, $n \in \N \!\setminus\! \{0\}$, are non-empty and $K_{n'} \subseteq K_n$ for all $n, n' \in \N \!\setminus\! \{0\}$ such that $n \leq n'$, by definition, we get  
	\begin{equation*}
		\bigcap_{n = 1}^\infty K_n \neq \emptyset.
	\end{equation*}
	Hence, we can find at least one $\nu_0$ in $\bigcap_{n = 1}^\infty K_n$ which, by definition, satisfies $\normTV{\nu_0} \leq m + 1/n$ for all $n \in \N \!\setminus\! \{0\}$. Letting $n \rightarrow \infty$, we obtain $\normTV{\nu_0} \leq m$. Since $m$ also satisfies $m \leq \normTV{\nu_0}$, it follows that $\normTV{\nu_0} = m$. This shows that $\measProb_0$ is a solution of $\gaborProb$, thereby completing the proof.
\end{proof}	

Theorem~\ref{thm: existence} shows that $\gaborProb$ has at least one solution.
Next, with the help of Fenchel duality theory~\cite[Chap.~4]{Borwein2005}, we derive necessary and sufficient conditions for $\meas$ to be contained in the set of solutions of $\gaborProb$. Specifically, we apply the following theorem.

\begin{thm}[Fenchel duality, {\cite[Thms.~4.4.2 and 4.4.3]{Borwein2005}}]
		\label{thm: fenchel}
		Let $X$ and $Y$ be Banach spaces, $\varphi \colon X \rightarrow \R \cup \{\infty\}$ and $\psi \colon Y \rightarrow \R \cup \{\infty\}$ convex functions, and $\mathcal{L} \colon X \rightarrow Y$ a bounded linear operator. Then, the primal and dual values \vspace{-0.3cm}
		\begin{align}
			p &= \inf_{x \in X} \left\{\varphi(x) + \psi(\mathcal{L}x)\right\} \label{eq: dual 0} \\
			d &= \sup_{y^* \in Y^*} \left\{-\varphi^*(\mathcal{L}^*y^*) - \psi^*(-y^*)\right\} \label{eq: dual}
		\end{align}
		satisfy the weak duality inequality $p \geq d$. Furthermore, if $\varphi$, $\psi$, and $\mathcal{L}$ satisfy the condition
		\begin{equation}
			\left(\mathcal{L} \dom \varphi\right) \cap \cont \psi \neq \emptyset,
			\label{eq: condition strong duality fenchel thm}
		\end{equation}
		then strong duality holds, i.e., $p = d$, and the supremum in the dual problem \eqref{eq: dual} is attained if $d$ is finite\footnote{Here, $p, d \in \R \cup \{\pm\infty\}$.}.
\end{thm}

Application of Theorem~\ref{thm: fenchel} yields the following result.

\begin{thm}[Fenchel predual]
	\label{thm: fenchel dual gabor prob}
	Let $\meas \triangleq \sum_{\measSpikeIdx \in \measSupportIdx} \measAmplitude{\measSpikeIdx} \delta_{t_\measSpikeIdx} \in \measSpace{\group}$, where $\measSupportIdx$ is a finite or countably infinite index set and $\measAmplitude{\measSpikeIdx} \in \C\!\setminus\! \{0\}$, for all $\measSpikeIdx \in \measSupportIdx$. 
 	Let $\window \in \schwartzSpace{\group} \!\setminus\! \{0\}$ and $\measurements = \gaborOp\meas$. 
	The Fenchel predual problem of $\gaborProb$ is
	\begin{equation*}
		\dualGaborProb\qquad  \maximize_{c \in L^\infty(\group \times \dualGroup)} \innerProd{c}{\measurements} \st \normInfty{\gaborOp^*c} \leq 1
	\end{equation*}
	with the adjoint operator $\gaborOp^*$ as in \eqref{eq: adjoint operator expression}.
	In addition, the following equality holds 
	\begin{align}
		&\min\Big\{\normTV{\measProb} \colon \gaborOp\measProb = \measurements, \measProb \in \measSpace{\group}\Big\} \notag \\
		&\quad = \sup\Big\{\innerProd{c}{\measurements} \colon \normInfty{\gaborOp^*c} \leq 1, c \in L^\infty(\group \times \dualGroup)\Big\}.
		\label{eq: strong duality}
	\end{align}
	Moreover, if $\dualGaborProb$ has a solution $c_0 \in L^\infty(\group \times \dualGroup)$, then
	\begin{equation}
		\bigcup_{\nu_0 \in \mathrm{Sol}\{\gaborProb\}} \supp(\nu_0) \subseteq \{t \in \group \colon \abs{(\gaborOp^*c_0)(t)} = 1\}.
		\label{eq: implication dual problem}
	\end{equation}
\end{thm}

\begin{proof}
	Similarly to what was done in~\cite[Prop.~2]{Bredies2012}, we prove that $\gaborProb$ is the dual of $\dualGaborProb$ by applying Theorem~\ref{thm: fenchel} with $X = L^\infty(\group \times \dualGroup)$, $Y = C_b(\group)$, $\mathcal{L} = \gaborOp^*$ and with the functions
	\begin{equation}
		\forall c \in L^\infty(\group \times \dualGroup), \quad \varphi(c) \triangleq -\innerProd{c}{\measurements},
		\label{eq: definition varphi}
	\end{equation}
	and
	\begin{equation}
		\forall \eta \in C_b(\group), \quad \psi(\eta) \triangleq \chi_B(\eta) \triangleq \begin{cases} 0, & \eta \in B \\
			\infty, & \text{otherwise,}
			\end{cases}
		\label{eq: definition psi}
	\end{equation}
	where $\chi_B$ is the indicator function of the closed unit ball $B \triangleq \{\eta \in C_b(\group) \colon \normInfty{\eta} \leq 1\}$.  
	Thanks to Lemma~\ref{lem: lemma for adjoint operator}, we have $\mathcal{L}^* = \gaborOp^{**} = \gaborOp$. 
	As $\varphi$ is linear it is convex on $L^\infty(\group \times \dualGroup)$. The function $\psi$ is convex on $C_b(\group)$, because the closed unit ball $B$ is convex and the indicator function of a convex set is convex.
	The Fenchel convex conjugate of $\varphi$ is 
	\begin{align}
		\forall \xi \in L^1(\group \times \dualGroup), \quad \varphi^*(\xi) &= \sup\left\{\innerProd{c}{\xi} - \varphi(c) \colon c \in L^\infty(\group \times \dualGroup)\right\} \notag \\
			&= \sup\left\{\innerProd{c}{\xi + \measurements} \colon c \in L^\infty(\group \times \dualGroup)\right\} \notag \\
			&= \chi_{\{\measurements\}}(-\xi), \label{eq: definition varphi star}
	\end{align}
	where the last equality follows from the fact that the set $\left\{\innerProd{c}{\xi + \measurements} \colon c \in L^\infty(\group \times \dualGroup)\right\}$ is not bounded when $\xi \neq -\measurements$. 
	As for the convex conjugate of $\psi$, it is known that the convex conjugate of the indicator function of $B$ is the support function of $B$,
	that is,
	\begin{align}
		\forall \measProb \in \measSpace{\group}, \quad \psi^*(\measProb) &= \sigma_B(\measProb) \triangleq \sup\left\{\innerProd{\eta}{\measProb} \colon \eta \in B \right\} \notag \\
			&= \sup\left\{\innerProd{\eta}{\measProb} \colon \eta \in C_b(\group), \normInfty{\eta} \leq 1\right\} \notag \\
			&= \normTV{\measProb}, \label{eq: norm TV expression Riesz}
	\end{align}
	where \eqref{eq: norm TV expression Riesz} follows from \cite[Thm.~IV.6.2]{Dunford1988}. 
	It then follows by application of Theorem~\ref{thm: fenchel} that the dual problem to $\dualGaborProb$ is
	\begin{equation*}
		\gaborProb \quad \minimize_{\measProb \in \measSpace{\group}} \normTV{\measProb} \st \measurements = \gaborOp\measProb.
	\end{equation*}
	Next, we prove \eqref{eq: strong duality}. To this end, we verify \eqref{eq: condition strong duality fenchel thm}, which implies \eqref{eq: strong duality} by strong duality, i.e., $p = d$, as follows.
	We start by noting that with $X = L^\infty(\group \times \dualGroup)$, $Y = C_b(\group)$, $\mathcal{L} = \gaborOp^*$, and $\varphi$, $\psi$ as defined in \eqref{eq: definition varphi}, \eqref{eq: definition psi}, we get from \eqref{eq: dual 0} that
	\begin{align*}
		p &= \inf_{c \in L^\infty(\group \times \dualGroup)} \Big\{\varphi(c) + \psi(\gaborOp^*c)\Big\} = \inf_{c \in L^\infty(\group \times \dualGroup)} \Big\{\!-\innerProd{c}{\measurements} + \chi_B(\gaborOp^*c)\Big\} \\
			&= \sup_{c \in L^\infty(\group \times \dualGroup)} \Big\{\!\innerProd{c}{\measurements} - \chi_B(\gaborOp^*c)\Big\} = \sup\left\{\innerProd{c}{\measurements} \colon \normInfty{\gaborOp^*c} \leq 1, c \in L^\infty(\group \times \dualGroup)\right\},
	\end{align*}
	where in the last step we used the fact that $\chi_B(\gaborOp^*c)$ is finite, actually equal to zero, only for $c \in L^\infty(\group \times \dualGroup)$ such that $\normInfty{\gaborOp^*c} \leq 1$.
	With $\varphi^*$, $\psi^*$ as in \eqref{eq: definition varphi star}, \eqref{eq: norm TV expression Riesz}, and $(C_b(\group))^* = \measSpace{\group}$, it follows from \eqref{eq: dual} that
	\begin{align}
		d &= \sup_{\measProb \in \measSpace{\group}} \Big\{-\varphi^*(\gaborOp\measProb) - \psi^*(-\measProb)\Big\} = \sup_{\measProb \in \measSpace{\group}} \Big\{-\varphi^*(-\gaborOp\measProb) - \psi^*(\measProb)\Big\} \notag\\
			&= \sup_{\measProb \in \measSpace{\group}} \Big\{-\chi_{\{\measurements\}}(\gaborOp\measProb) - \normTV{\measProb}\Big\} = \inf_{\measProb \in \measSpace{\group}} \Big\{\chi_{\{\measurements\}}(\gaborOp\measProb) + \normTV{\measProb}\Big\} \notag \\
			&= \inf\Big\{\normTV{\measProb} \colon \gaborOp\measProb = \measurements, \measProb \in \measSpace{\group}\Big\} \label{eq: proof dual equality d 1 }\\
			&= \min\Big\{\normTV{\measProb} \colon \gaborOp\measProb = \measurements, \measProb \in \measSpace{\group}\Big\}. \label{eq: proof dual equality d 2 }
	\end{align}
	Again, in \eqref{eq: proof dual equality d 1 } we used the fact that $\chi_{\{\measurements\}}(\gaborOp\measProb)$ is finite, actually equal to zero, only for $\measProb \in \measSpace{\group}$ such that $\gaborOp\measProb = \measurements$. The infimum in \eqref{eq: proof dual equality d 1 } becomes a minimum in \eqref{eq: proof dual equality d 2 } as a consequence of the existence of a solution of $\gaborProb$, as shown in Theorem \ref{thm: existence}; this solution will be denoted by $\measProb_0$ below.
	It remains to establish \eqref{eq: condition strong duality fenchel thm}. By definition of $\varphi$, we have $\dom \varphi = L^\infty(\group \times \dualGroup)$, and $0 \in \mathcal{L}\dom(\varphi)$. Furthermore, since $\cont \psi = B$, we have $0 \in \cont \psi$ and therefore $0 \in (\mathcal{L}\dom \varphi) \cap (\cont \psi)$, which shows that \eqref{eq: condition strong duality fenchel thm} is satisfied.
	
	Next, we verify \eqref{eq: implication dual problem}. If $\dualGaborProb$ has at least one solution, say $c_0 \in L^\infty(\group \times \dualGroup)$, then necessarily $\normTV{\measOpt} = \innerProd{c_0}{\measurements}$ as a consequence of strong duality. Since $\measOpt$ is a solution of $\gaborProb$, we have $\measurements = \gaborOp\measOpt$, and consequently, 
	\begin{equation*}
		\normTV{\measOpt} = \innerProd{c_0}{\measurements} = \innerProd{c_0}{\gaborOp\measOpt} = \innerProd{\gaborOp^*c_0}{\measOpt} = \Re{\int_\group \overline{(\gaborOp^*c_0)(t)} \dmeas{\measOpt}{t}}.
	\end{equation*}
	Applying~\cite[Thm.~6.12]{Rudin1987} to $\measProb_0$, we can conclude that there exists a measurable function $\eta \colon [0, 2\pi) \rightarrow \R$ such that 
	\begin{equation*}
		\normTV{\measOpt} = \Re{\int_\group \overline{(\gaborOp^*c_0)(t)} \dmeas{\measOpt}{t}} =  \Re{\int_\group \overline{(\gaborOp^*c_0)(t)}e^{i\eta(t)} \dmeas{|\measOpt|}{t}} = \Re{\int_\group \overline{h(t)} \dmeas{|\measOpt|}{t}},
	\end{equation*}
	where we set $h \triangleq (\gaborOp^*c_0) e^{-i\eta}$. Writing $h = \abs{h}\!e^{i\theta}$, where $\theta \colon [0, 2\pi) \rightarrow \R$, we then have
	\begin{align}
		0 = \normTV{\measOpt} - \Re{\int_\group \overline{h(t)} \dmeas{|\measOpt|}{t}} &= \int_\group \dmeas{|\measOpt|}{t}  - \Re{\int_\group \abs{h(t)} e^{-i\theta(t)} \dmeas{|\measOpt|}{t}} \notag \\
			&= \int_\group \Big(1 - \abs{h(t)}\cos(\theta(t))\Big) \dmeas{|\measOpt|}{t}. \label{eq: zero almost everywhere}
	\end{align}
	As $\gaborOp^* \colon L^\infty(\group \times \dualGroup) \rightarrow C_b(\group)$, the function $\abs{h} = \abs{\gaborOp^*c_0}$ is continuous. Moreover, since $c_0$ is a solution of $\dualGaborProb$, it follows that $\normInfty{h} = \normInfty{\gaborOp^*c_0} \leq 1$, and hence,
	\begin{equation*}
		\forall t \in \group, \quad 1 - \abs{h(t)}\cos(\theta(t)) \geq 0.
	\end{equation*}
	Therefore, \eqref{eq: zero almost everywhere} implies that $\abs{h(t)}\cos(\theta(t)) = 1$ for $|\measOpt|$-almost all $t \in \group$, which, as a consequence of $\cos \leq 1$ and $\normInfty{h} \leq 1$, yields $\abs{h(t)} = 1$, for $\abs{\nu_0}$-almost all $t \in \group$. We complete the proof by establishing that for all $t \in \supp(\measProb_0)$, $\abs{h(t)} = \abs{(\gaborOp^*c_0)(t)} = 1$. Since $\abs{h(t)} = 1$ for $\abs{\measProb_0}$-almost all $t \in \group$, the set $S = \{t \in \group \colon \abs{h(t)} \neq 1\}$ satisfies $|\measOpt|(S) = 0$. We need to show that $S \cap \supp(\measProb_0) = \emptyset$. To this end, suppose, by way of contradiction, that there exists a $t_0 \in \supp(\measOpt)$ such that $\abs{h(t_0)} \neq 1$. Then, $S \cap \supp(\measOpt) \neq \emptyset$ as $t_0 \in S \cap \supp(\measOpt)$. Since $\abs{h}$ is continuous  on $\group$ and $\R\!\setminus\!\{1\}$ is open in $\R$, the set $S$ is open. Hence, by definition of $\supp(\measProb_0)$, we have $|\measOpt|(S) \geq |\measOpt|(S \cap \supp(\measOpt)) > 0$, which contradicts $|\measOpt|(S) = 0$ and thereby completes the proof.
\end{proof}

We emphasize that $\dualGaborProb$ is the \textit{pre}dual problem of $\gaborProb$, i.e., $\gaborProb$ is the dual problem of~$\dualGaborProb$. The dual problem of $\gaborProb$ is, however, not $\dualGaborProb$ as the spaces $L^1(\R^2)$ and $C_b(\R)$ are not reflexive. We thus remark that the second statement in Theorem~\ref{thm: fenchel} concerning strong duality cannot be applied by taking $\gaborProb$ as the primal problem, that is, with $X = \measSpace{\group}$, $Y = L^1(\group \times \dualGroup)$, $\mathcal{L} = \gaborOp$, $\varphi(\measProb) = \normTV{\measProb}$ for all $\measProb \in \measSpace{\group}$, and $\psi(c) = \chi_{\{y\}}(c)$ for all $c \in L^1(\group \times \dualGroup)$.
Indeed, the function $\chi_{\{y\}}$ is not continuous at $y$, implying that Condition~\eqref{eq: condition strong duality fenchel thm} is not met. 

We can now apply Theorem~\ref{thm: fenchel dual gabor prob} to conclude the following: Take $\group = \R$ and $\window(t) = \frac{1}{\sqrt{\sigma}}\exp\left(-\frac{\pi t^2}{2\sigma^2}\right)$, for $t \in \R$. Assuming that $\dualGaborProb$ has a solution\footnote{A condition for the existence of a solution of $\dualGaborProb$ will be given in Corollary~\ref{cor: existence solution dual}.}, which we denote by $c_0 \in L^\infty(\R^2)$, the support $\measSupport = \{t_\measSpikeIdx\}_{\measSpikeIdx \in \measSupportIdx}$ of the measure $\meas$ to be recovered must satisfy $\abs{(\gaborOp^*c_0)(t_\measSpikeIdx)} = 1$, for $\measSpikeIdx \in \measSupportIdx$, if $\meas$ is to be contained in the set of solutions of $\gaborProb$. Furthermore, if $\dualGaborProb$ has a solution $c_0 \in L^\infty(\R^2)$ such that $\abs{\gaborOp^*c_0}$ is not identically equal to $1$ on $\R$, then every solution of $\gaborProb$ is a discrete measure despite the fact that $\gaborProb$ is an optimization problem over the space of \textit{all} measures in $\measSpace{\group}$. This can be seen by employing an argument similar to the one in \cite[p.~361]{Beurling1989-2} as follows: Since the window function $\window$ is Gaussian and
\begin{equation*}
	\forall t \in \R, \quad (\gaborOp^*c_0)(t) = \int_\R c_0(\tau, \freq)\window(t - \tau)e^{2\pi i\freq t}\mathrm{d}\tau\mathrm{d}\freq,
\end{equation*}
both $\window$ and $\gaborOp^*c_0$ can be extended to entire functions which, for simplicity, we also denote by $\window$ and $\gaborOp^*c_0$, that is,
\begin{align}
	&\forall z \in \C, \quad \window(z) = \frac{1}{\sqrt{\sigma}} \exp\left(-\frac{\pi z^2}{2\sigma^2}\right) \quad \text{and} \notag\\
	&(\gaborOp^*c_0)(z) = \int_{-\freq_c}^{\freq_c}\int_\R c_0(\tau, \freq)\window(z - \tau)e^{2\pi i\freq z}\mathrm{d}\tau\mathrm{d}\freq. \label{eq: interpolating function}
\end{align}
It is known that $\window$ is holomorphic on $\C$ \cite[Chap.~10]{Rudin1987}, but it remains to show that so is $\gaborOp^*c_0$. To this end, let $T$ be an arbitrary closed triangle in $\C$. We have
\begin{align}
	\int_{\partial T} (\gaborOp^*c_0)(z)\mathrm{d}z &= \int_{\partial T}\left(\int_{-\freq_c}^{\freq_c}\int_\R c_0(\tau, \freq) \window(z - \tau)e^{2\pi i\freq z}\mathrm{d}\tau\mathrm{d}\freq\right)\mathrm{d}z \notag \\
		&= \int_{\freq_c}^{\freq_c}\int_\R c_0(\tau, \freq) \left(\int_{\partial T} \window(z - \tau)e^{2\pi i\freq z}\mathrm{d}z\right) \mathrm{d}\tau\mathrm{d}\freq  \label{eq: analytic continuation fubini} \\
		&= 0,  \label{eq: analytic continuation cauchy}
\end{align}
where we used Fubini's Theorem in \eqref{eq: analytic continuation fubini} and Cauchy's Theorem~\cite[Thm.~10.13]{Rudin1987} in \eqref{eq: analytic continuation cauchy}. To verify the conditions for Fubini's Theorem, we note that $T$ is bounded, i.e., there exists $M > 0$ such that $\abs{z} \leq M$ for all $z \in T$, which implies that for all $z \triangleq r + iq \in T$, $\tau \in \R$, and $\freq \in [-\freq_c, \freq_c]$, we have 
\begin{align*}
	\abs{\window(z - \tau)e^{2\pi i\freq z}} &= \frac{1}{\sqrt{\sigma}}\exp\!\left(-\frac{\pi(r - \tau)^2}{2\sigma^2}\right)\exp\!\left(\frac{\pi q^2}{2\sigma^2}\right)\exp(-2\pi\freq q) \\
		&\leq \frac{1}{\sqrt{\sigma}}\exp\!\left(-\frac{\pi(r - \tau)^2}{2\sigma^2}\right)\exp\!\left(\frac{\pi M^2}{2\sigma^2}\right)\exp(2\freq_c M),
\end{align*}
where we used $\abs{q} \leq M$ and $\freq \in [-\freq_c, \freq_c]$ in the last inequality.
This yields
\begin{align*}
	\int_{-\freq_c}^{\freq_c}\int_\R& \abs{c_0(\tau, \freq)} \abs{\window(z - \tau)e^{2\pi i\freq z}} \mathrm{d}\tau\mathrm{d}\freq \\
		&\leq \frac{2\freq_c\normLinfty{c_0}}{\sqrt{\sigma}}\left(\int_\R\exp\!\left(-\frac{\pi(r - \tau)^2}{2\sigma^2}\right)\mathrm{d}\tau\right)\exp\!\left(\frac{\pi M^2}{2\sigma^2}\right)\exp(2\freq_c M) \\
		&= 2\freq_c\normLinfty{c_0}\sqrt{2\sigma}\exp\!\left(\frac{\pi M^2}{2\sigma^2}\right)\exp(2\freq_c M),
\end{align*}
and hence
\begin{align*}
	\int_{\partial T}\int_{-\freq_c}^{\freq_c}\int_\R& \abs{c_0(\tau, \freq)} \abs{\window(z - \tau)e^{2\pi i\freq z}} \mathrm{d}\tau\mathrm{d}\freq\mathrm{d}z \\
		&\leq 2L\freq_c\normLinfty{c_0}\sqrt{2\sigma}\exp\!\left(\frac{\pi M^2}{2\sigma^2}\right)\exp(2\freq_c M) < \infty,
\end{align*}
where $L$ is the length of $\partial T$. By Morera's Theorem~\cite[Thm.~10.17]{Rudin1987} we can therefore conclude from \eqref{eq: analytic continuation cauchy} that $\gaborOp^*c_0$ is holomorphic on $\C$.
We can then define the function
\begin{equation*}
	\forall z \in \C, \quad h(z) \triangleq 1 - (\gaborOp^*c_0)(z)\overline{(\gaborOp^*c_0)(\overline{z})}.
\end{equation*}
Since $\abs{\gaborOp^*c_0}$ is not identically equal to $1$ on $\R$, by assumption, and $h$ is holomorphic on $\C$ by construction, $h$ is not identically equal to $0$. Therefore, thanks to~\cite[Thm.~10.18]{Rudin1987} the set $\{z \in \C \colon h(z) = 0\}$, and a fortiori the set $\{t \in \R \colon \abs{(\gaborOp^*c_0)(t)} = 1\}$, are at most countable and have no limit points. But since \eqref{eq: implication dual problem} holds, this implies that every solution $\nu_0$ of $\gaborProb$ must have discrete support, and therefore, $\nu_0$ is necessarily a discrete measure. In the case $\group = \torus$, if $\window$ is chosen to be a periodized Gaussian function
\begin{equation*}
	\forall t \in \torus, \quad \window(t) = \frac{1}{\sqrt{\sigma}}\sum_{n \in \Z} \exp\left(-\frac{\pi(t - n)^2}{2\sigma^2}\right),
\end{equation*}
a similar line of reasoning can be applied to establish that every solution of $\gaborProb$ is discrete provided that $\dualGaborProb$ has at least one solution.

We have established the existence of a solution of $\gaborProb$ and shown that $\gaborProb$ is the dual problem of $\dualGaborProb$. Now we wish to derive conditions under which the measure $\mu$ is the unique solution of $\gaborProb$. Similarly to~\cite[Sec.~2.4]{Duval2013}, we start by providing a necessary and sufficient condition for $\meas$ to be a solution of $\gaborProb$, which leads to the following theorem.

\begin{thm}[Optimality conditions]
	\label{thm: equivalence condition mu sol}
	Let $\meas \triangleq \sum_{\measSpikeIdx \in \measSupportIdx} \measAmplitude{\measSpikeIdx} \delta_{t_\measSpikeIdx} \in \measSpace{\group}$, where $\measSupportIdx$ is a finite or countably infinite index set and $\measAmplitude{\measSpikeIdx} \in \C\!\setminus\! \{0\}$, for all $\measSpikeIdx \in \measSupportIdx$. 
 	Let $\window \in \schwartzSpace{\group} \!\setminus\! \{0\}$. 
	The measure $\meas$ to be recovered is contained in the set of solutions of $\gaborProb$ if and only if there exists $c_0 \in L^\infty(\group \times \dualGroup)$ such that
	\begin{equation}
		\normInfty{\gaborOp^*c_0} \leq 1 \quad \text{and} \quad \forall \measSpikeIdx \in \measSupportIdx, \quad (\gaborOp^*c_0)(t_\measSpikeIdx) = \frac{\measAmplitude{\measSpikeIdx}}{\abs{\measAmplitude{\measSpikeIdx}}}.
		\label{eq: equivalence condition mu sol}
	\end{equation}
\end{thm}

\begin{proof}
	Since $\measurements = \gaborOp\meas$, it follows by application of \cite[Thm.~4.3.6]{Borwein2005} that $\meas$ is a solution of $\gaborProb$ if and only if there exists $c_0 \in L^\infty(\group \times \dualGroup)$ such that $\gaborOp^*c_0 \in \partial\!\normTV{\cdot}(\meas)$. Since the subdifferential of the norm $\normTV{\cdot}$ is given by~\cite[Prop.~7]{Duval2013} 
	\begin{equation*}
		\forall \measProb \in \measSpace{\group}, \quad \partial\!\normTV{\cdot}(\measProb) = \left\{\eta \in C_b(\group) \colon \innerProd{\eta}{\measProb} = \normTV{\measProb}, \normInfty{\eta} \leq 1\right\},
	\end{equation*}
	it follows that $\meas$ is a solution of $\gaborProb$ if and only if there exists $c_0 \in L^\infty(\group \times \dualGroup)$ such that $\innerProd{\gaborOp^*c_0}{\meas} = \normTV{\meas}$ and $\normInfty{\gaborOp^*c_0} \leq 1$. Since $\meas = \sum_{\measSpikeIdx \in \measSupportIdx} \measAmplitude{\measSpikeIdx} \delta_{t_\measSpikeIdx}$, the condition $\innerProd{\gaborOp^*c_0}{\meas} = \normTV{\meas}$ is equivalent to
	\begin{equation}
		\Re{\sum_{\measSpikeIdx \in \measSupportIdx} \measAmplitude{\measSpikeIdx} \overline{(\gaborOp^*c_0)(t_\measSpikeIdx)}} = \sum_{\measSpikeIdx \in \measSupportIdx} \abs{\measAmplitude{\measSpikeIdx}}.
		\label{eq: equivalent condition cos}
	\end{equation}
	For every $\measSpikeIdx \in \measSupportIdx$, we can write
	\begin{equation}
		\measAmplitude{\measSpikeIdx}\overline{(\gaborOp^*c_0)(t_\measSpikeIdx)} = \abs{\measAmplitude{\measSpikeIdx}}\abs{(\gaborOp^*c_0)(t_\measSpikeIdx)} e^{i\theta_\measSpikeIdx},
		\label{eq: equation to obtain aell divided by its module}
	\end{equation}
	 where $\theta_\measSpikeIdx \in [0, 2\pi)$. Thus, \eqref{eq: equivalent condition cos} is equivalent to
	\begin{equation}
		\sum_{\measSpikeIdx \in \measSupportIdx} \abs{\measAmplitude{\measSpikeIdx}} \Big(1 - \cos(\theta_\measSpikeIdx)\abs{(\gaborOp^*c_0)(t_\measSpikeIdx)}\Big) = 0.
		\label{eq: condition can only be true}
	\end{equation}
	Since we also have $\normInfty{\gaborOp^*c_0} \leq 1$, \eqref{eq: condition can only be true} can only be true if $\abs{(\gaborOp^*c_0)(t_\measSpikeIdx)} = 1$ and $\theta_\measSpikeIdx = 0$, both for all $\measSpikeIdx \in \measSupportIdx$. Therefore, from \eqref{eq: equation to obtain aell divided by its module}, we can infer that
	\begin{equation*}
		\forall \measSpikeIdx \in \measSupportIdx, \quad (\gaborOp^*c_0)(t_\measSpikeIdx) = \frac{\measAmplitude{\measSpikeIdx}}{\abs{\measAmplitude{\measSpikeIdx}}},
	\end{equation*}
	which concludes the proof.
\end{proof}

Thanks to Theorem~\ref{thm: equivalence condition mu sol}, we can now show the following, which yields a sufficient condition for $\dualGaborProb$ to have at least one solution, an assumption made previously to show that every solution of $\gaborProb$ is discrete.

\begin{cor}
	Let $\meas \triangleq \sum_{\measSpikeIdx \in \measSupportIdx} \measAmplitude{\measSpikeIdx} \delta_{t_\measSpikeIdx} \in \measSpace{\group}$, where $\measSupportIdx$ is a finite or countably infinite index set and $\measAmplitude{\measSpikeIdx} \in \C\!\setminus\! \{0\}$, for all $\measSpikeIdx \in \measSupportIdx$. 
 	Let $\window \in \schwartzSpace{\group} \!\setminus\! \{0\}$. 
	If $\meas$ is contained in the set of solutions of $\gaborProb$, then $\dualGaborProb$ has at least one solution.
	\label{cor: existence solution dual}
\end{cor}

\begin{proof}
	If the measure $\meas$ to be recovered is a solution of $\gaborProb$, then by Theorem~\ref{thm: equivalence condition mu sol}, there must exist $c_0 \in L^\infty(\group \times \dualGroup)$ such that \eqref{eq: equivalence condition mu sol} holds.
	Since $\normInfty{\gaborOp^*c_0} \leq 1$, $c_0$ is feasible for $\dualGaborProb$. Moreover, this $c_0$ satisfies
	\begin{align*}
		\innerProd{c_0}{\measurements} &= \innerProd{c_0}{\gaborOp\meas} = \innerProd{\gaborOp^*c_0}{\meas} = \Re{\int_\group \overline{(\gaborOp^*c_0)(t)} \dmeas{\meas}{t}} \\
			&= \Re{\sum_{\measSpikeIdx \in \measSupportIdx} \measAmplitude{\measSpikeIdx} \overline{(\gaborOp^*c_0)(t_\measSpikeIdx)}} = \sum_{\measSpikeIdx \in \measSupportIdx} \abs{\measAmplitude{\measSpikeIdx}} = \normTV{\meas},
	\end{align*}
	which shows that the supremum in~\eqref{eq: strong duality} is attained for~$c_0$. Therefore, $c_0$ is a solution of $\dualGaborProb$.
\end{proof}

Theorem~\ref{thm: equivalence condition mu sol} provides conditions on $\meas$ to be a solution of~$\gaborProb$. However, we need more, namely, conditions on $\meas$ to be the \emph{unique} solution of $\gaborProb$. Such conditions are given in the following theorem, which is a straightforward adaptation of \cite[App.~A]{Candes2012}.

\begin{thm}[Uniqueness]
	Let $\meas \triangleq \sum_{\measSpikeIdx \in \measSupportIdx} \measAmplitude{\measSpikeIdx} \delta_{t_\measSpikeIdx} \in \measSpace{\group}$, where $\measSupportIdx$ is a finite or countably infinite index set and $\measAmplitude{\measSpikeIdx} \in \C\!\setminus\! \{0\}$, for all $\measSpikeIdx \in \measSupportIdx$. 
 	Let $\window \in \schwartzSpace{\group} \!\setminus\! \{0\}$. 
	If for every sequence $\varepsilon = \{\varepsilon_\measSpikeIdx\}_{\measSpikeIdx \in \measSupportIdx}$ of unit magnitude complex numbers, there exists a function $c_0 \in L^\infty(\group \times \dualGroup)$ obeying
	\begin{align}
		\forall \measSpikeIdx \in \measSupportIdx, \quad &(\gaborOp^*c_0)(t_\measSpikeIdx) = \varepsilon_\measSpikeIdx \label{eq: interpolation}\\
		\forall t \in \group\!\setminus\!\measSupport, \quad &\abs{(\gaborOp^*c_0)(t)} < 1, \label{eq: constraint interpolation}
	\end{align}
	where $\measSupport \triangleq \{t_\measSpikeIdx\}_{\measSpikeIdx \in \measSupportIdx}$, then $\meas$ is the unique solution of $\gaborProb$.
	\label{thm: uniqueness dual certificate}
\end{thm}

\begin{proof}
	Assume that there exists a $c_0 \in L^\infty(\group \times \dualGroup)$ such that \eqref{eq: interpolation} and \eqref{eq: constraint interpolation} are satisfied for all sequences $\{\varepsilon_\measSpikeIdx\}_{\measSpikeIdx \in \measSupportIdx}$ of unit magnitude complex numbers. Let $\measOpt \in \measSpace{\group}$ be a solution of $\gaborProb$ and set $\measErr \triangleq \meas - \measOpt$. Suppose, by way of contradiction, that $\measErr \neq 0$. The Lebesgue decomposition~\cite[Thm.~6.10]{Rudin1987} of $\measErr$ relative to the positive and $\sigma$-finite measure $\Psi_\measSupport = \sum_{\measSpikeIdx \in \measSupportIdx} \delta_{t_\measSpikeIdx}$ is given by 
	\begin{equation*}
		\measErr = \measErr_\measSupport + \measErr_{\measSupport^c},
	\end{equation*}
	where $\measErr_\measSupport$ is a discrete measure of the form $\measErr_\measSupport \triangleq \sum_{\measSpikeIdx \in \measSupportIdx} \measErr_\measSpikeIdx \delta_{t_\measSpikeIdx}$ (see Part (b) of \cite[Thm.~6.10]{Rudin1987}) and $\measErr_{\measSupport^c}$ is a measure supported on $\measSupport^c \triangleq \group \!\setminus\! \measSupport$ (see Part (a) of \cite[Thm.~6.10]{Rudin1987}). Using the polar decomposition~\cite[Thm.~6.12]{Rudin1987} of $\measErr_\measSupport$, we find that there exists a measurable function $\eta$ such that $\abs{\eta(t)} = 1$, for all $t \in \measSupport$, and
	\begin{equation}
		\int_\group \varphi(t)\dmeas{\measErr_\measSupport}{t} = \int_\group \varphi(t)\eta(t)\!\dmeas{\abs{\measErr_\measSupport}\!}{t}
		\label{eq: polar decomposition proof uniqueness}
	\end{equation}
	for all $\varphi \in C_b(\group)$.
	By assumption, there exists a $c_0 \in L^\infty(\group \times \dualGroup)$ such that
	\begin{align}
		\forall \measSpikeIdx \in \measSupportIdx, \quad & (\gaborOp^*c_0)(t_\measSpikeIdx) = \eta(t_\measSpikeIdx), \label{eq: conditions theorem 1} \\
		\forall t \in \group \!\setminus\! \measSupport, \quad & \abs{(\gaborOp^*c_0)(t)} < 1. \label{eq: conditions theorem 2}
	\end{align}
	We thus have
	\begin{align}
		\innerProd{\gaborOp^*c_0}{\measErr_{\measSupport^c}} &=  \innerProd{\gaborOp^*c_0}{\measErr} - \innerProd{\gaborOp^*c_0}{\measErr_\measSupport} \notag \\
			&= \innerProd{c_0}{\gaborOp\measErr} -  \Re{ \int_\group \overline{(\gaborOp^*c_0)(t)}\dmeas{\measErr_\measSupport}{t}} \notag \\
			&= \innerProd{c_0}{\gaborOp\meas - \gaborOp\measProb_0} -  \Re{ \int_\group \overline{(\gaborOp^*c_0)(t)}\eta(t)\!\dmeas{\abs{\measErr_\measSupport}\!}{t}}  \label{eq: use of polar decomposition in the proof of uniqueness} \\
			&= -  \Re{\sum_{\measSpikeIdx \in \measSupportIdx} \abs{\measErr_\measSpikeIdx} \overline{(\gaborOp^*c_0)(t_\measSpikeIdx)} \eta(t_\measSpikeIdx)} = - \sum_{\measSpikeIdx \in \measSupportIdx} \abs{\measErr_\measSpikeIdx}\abs{\eta(t_\measSpikeIdx)}^2  \label{eq: proof uniqueness}\\
			&=  - \sum_{\measSpikeIdx \in \measSupportIdx} \abs{\measErr_\measSpikeIdx} = -\normTV{\measErr_\measSupport}, \label{eq: proof uniqueness last inequality of the chain}
	\end{align}
	where \eqref{eq: use of polar decomposition in the proof of uniqueness} makes use of \eqref{eq: polar decomposition proof uniqueness}, \eqref{eq: proof uniqueness} follows from~\eqref{eq: conditions theorem 1} and $\measurements = \gaborOp\meas = \gaborOp\measProb_0$, which holds because $\measProb_0$ was assumed to be a solution of $\gaborProb$, and \eqref{eq: proof uniqueness last inequality of the chain} is by $\abs{\eta(t)} = 1$, for all $t \in \measSupport$.
	It therefore follows that
	\begin{align}
		 \normTV{\measErr_\measSupport} &= \abs{\innerProd{\gaborOp^*c_0}{\measErr_{\measSupport^c}}} = \abs{\Re{\int_\group \overline{(\gaborOp^*c_0)(t)}\dmeas{\measErr_{\measSupport^c}}{t}}} \leq \abs{\int_\group \overline{(\gaborOp^*c_0)(t)}\dmeas{\measErr_{\measSupport^c}}{t}} \notag \\
			&\leq \int_\group \abs{(\gaborOp^*c_0)(t)}\dmeas{\abs{\measErr_{\measSupport^c}}\!}{t} \label{eq: triangle inequality for measures} \\
			&< \abs{\measErr_{\measSupport^c}}(\group) = \normTV{\measErr_{\measSupport^c}}, \label{eq: strict inequality}
	\end{align}
	where \eqref{eq: triangle inequality for measures} is by the triangle inequality and the strict inequality in~\eqref{eq: strict inequality} is due to \eqref{eq: conditions theorem 2} and $\measErr \neq 0$ and hence $\measErr_{\measSupport^c} \neq 0$ (as $\measErr_{\measSupport^c} = 0$ would imply $\measErr_\measSupport = 0$ and thus $\measErr = 0$, because of \eqref{eq: triangle inequality for measures}). 
	We then obtain
	\begin{align}
		\normTV{\measOpt} = \normTV{\meas - \measErr} &= \normTV{\meas - \measErr_\measSupport} + \normTV{\measErr_{\measSupport^c}} \label{eq: disjoint proof uniqueness} \\
			&\geq \normTV{\meas} - \normTV{\measErr_\measSupport} + \normTV{\measErr_{\measSupport^c}} \label{eq: large inequality proof uniqueness} \\
			&> \normTV{\meas}, \label{eq: strict inequality proof uniqueness}
	\end{align}
	where \eqref{eq: disjoint proof uniqueness} is a consequence of $\meas$ being supported on $\measSupport$ and hence the supports of $\meas - \measErr_\measSupport$ and $\measErr_{\measSupport^c}$ being disjoint,
	\eqref{eq: large inequality proof uniqueness} follows from the reverse triangle inequality and \eqref{eq: strict inequality proof uniqueness} is a consequence of \eqref{eq: strict inequality}.
	This contradicts the assumption that $\measProb_0$ is a solution of $\gaborProb$, which allows us to conclude that $\measErr = 0$ and hence $\measOpt = \meas$, thus completing the proof.
\end{proof}

We have now developed a full theory of super-resolution from STFT measurements for window functions $\window$ from the Schwartz-Bruhat space that are extendable to entire functions. It remains, however, to connect the uniqueness conditions \eqref{eq: interpolation} and \eqref{eq: constraint interpolation} in Theorem~\ref{thm: uniqueness dual certificate}, which amount to constrained interpolation problems, to the minimum spacing condition $\Delta > 1/\freq_c$, as announced earlier in the paper. This will be accomplished through a specific choice for $\window$, namely we take it to be a Gaussian for $\group = \R$ and a periodized Gaussian for $\group = \torus$. Note that the choice of the window width is only for reasons of specificity of the results. Similar results can be derived for other Schwartz-Bruhat window functions or for Gaussian window functions of smaller widths: we only need to adapt the computational part of our proof, but the way of reasoning remains the same.

\begin{thm}[Conditions for exact recovery in the case $\group = \R$]
\label{thm: exact recovery}
	Let $\group = \R$ and 
	\begin{equation*}
		\forall t \in \R, \quad \window(t) \triangleq \frac{1}{\sqrt{\sigma}}\exp\left(-\frac{\pi t^2}{2\sigma^2}\right),
	\end{equation*}
	where $\sigma \triangleq \frac{1}{4\freq_c}$. 
	Let $\meas \triangleq \sum_{\measSpikeIdx \in \measSupportIdx} \measAmplitude{\measSpikeIdx} \delta_{t_\measSpikeIdx} \in \measSpace{\R}$, where $\measSupportIdx$ is a finite or countably infinite index set and $\measAmplitude{\measSpikeIdx} \in \C\!\setminus\! \{0\}$, for all $\measSpikeIdx \in \measSupportIdx$. 
	If the minimum distance
	\begin{equation*}
		\Delta \triangleq \min_{\substack{\measSpikeIdx, \measSpikeIdxBis \in \measSupportIdx \\ \measSpikeIdx \neq \measSpikeIdxBis}} \abs{t_\measSpikeIdx - t_\measSpikeIdxBis}
	\end{equation*}
	satisfies $\Delta > 1/\freq_c\,$,
	then the conditions of Theorem~\ref{thm: uniqueness dual certificate} are met, and hence $\meas$ is the unique solution of~$\gaborProb$. 
\end{thm} 

The proof architecture of Theorem~\ref{thm: exact recovery} is inspired by~\cite[Sec.~2, pp.~15--27]{Candes2012}. 
There are, however, important differences arising from the fact that we deal with STFT measurements as opposed to pure Fourier measurements.
Specifically, in the case of pure Fourier measurements, the interpolation function $\mathcal{A}^*c_0$ must be a Paley-Wiener function~\cite[Thm.~19.3]{Rudin1987}, while here $\gaborOp^*c_0$ is clearly not band-limited and can therefore have better time-localization. We believe that this allows the minimum spacing $\Delta$ to be smaller than in the case of pure Fourier measurements. 

For $\group = \torus$, we have $\freq_c = K \in \N$. Moreover, the set $\measSupportIdx$ has to be finite, as $\torus$ is compact and $\measSupportIdx$ indexes $\measSupport$ which is closed and discrete. Recovering the measure $\meas$ therefore reduces to the recovery of the finite set of points $\{t_\measSpikeIdx\}_{\measSpikeIdx = 1}^L \subseteq [0, 1)$, $L \triangleq \abs{\measSupportIdx}$, and the associated weights $\{\measAmplitude{\measSpikeIdx}\}_{\measSpikeIdx \in \measSupportIdx}$.

\begin{thm}[Conditions for exact recovery in the case $\group = \torus$]
\label{thm: exact recovery torus}
	Let $\group = \torus$ and 
	\begin{equation*}
		\forall t \in \torus, \quad \window(t) \triangleq \frac{1}{\sqrt{\sigma}} \sum_{n \in \Z} \exp\left(-\frac{\pi(t -n)^2}{2\sigma^2}\right),
	\end{equation*}
	where $\sigma \triangleq \frac{1}{4(K+1/2)}$. 
	Let $\meas \triangleq \sum_{\measSpikeIdx = 1}^L \measAmplitude{\measSpikeIdx} \delta_{t_\measSpikeIdx} \in \measSpace{\torus}$ with $\measAmplitude{\measSpikeIdx} \in \C \!\setminus\! \{0\}$, for all $\measSpikeIdx \in \{1, 2, \ldots, L\}$.  
	If the minimum wrap-around distance 
	\begin{equation*}
		\Delta \triangleq \min_{n \in \Z} \min_{\substack{1 \leq \measSpikeIdx, \measSpikeIdxBis \leq L \\ \measSpikeIdx \neq \measSpikeIdxBis}} \abs{t_\measSpikeIdx - t_\measSpikeIdxBis + n}
	\end{equation*}
	satisfies $\Delta > 1/(K+1/2)$,
	then the conditions of Theorem~\ref{thm: uniqueness dual certificate} are met, and hence $\meas$ is the unique solution of~$\gaborProb$.
\end{thm}

%% file: recoveryForTorus.tex

\section{A recovery algorithm for $\group = \torus$}
\label{section: recovery of mu for the torus}

We next provide an explicit recovery algorithm for the case $\group = \torus$. 
Specifically, we show that if $\meas$ is the unique solution of $\gaborProb$ (which is the case provided that $\Delta > 1/(K + 1/2)$ and $\sigma = 1/(4(K+1/2))$), then an approximation of the correct solution $\meas$ can be recovered by solving the convex programming problem $\gaborProbN$, $N \in \N$, defined below. The predual of $\gaborProbN$, unlike that of $\gaborProb$, is equivalent to a finite-dimensional problem, which can be solved numerically. The justification for this procedure is given in Proposition \ref{prop: convergence solution}, which shows that the sequence of solutions $\nu_N$ of $\gaborProbN$ converges in the weak-* sense to $\meas$ as $N \rightarrow \infty$. 

We first note that the $1$-periodic window function $\window$ can be expanded into a Fourier series according to
\begin{equation}
	\forall t \in \R, \quad \window(t) = \sum_{n \in \Z} \frac{1}{\sqrt{\sigma}} \exp\left(-\frac{\pi(t - n)^2}{2\sigma^2}\right) = \sum_{n \in \Z} \window_n e^{2\pi int},
	\label{eq: window torus}
\end{equation}
with $\window_n = \sqrt{2\sigma} \exp\left(-2\pi\sigma^2n^2\right)$,  $n \in \Z$.
The corresponding STFT measurements of $\meas$ are given by
\begin{align*}
	\forall \tau \in \torus,\ \forall k \in \{-K, \ldots, K\}, \quad \measurements(\tau, k) &= (\STFT\meas)(\tau, k) = \displaystyle \int_0^1 \overline{\window(t - \tau)}e^{-2\pi ikt} \dmeas{\meas}{t} \\
		&=  \!\sum_{\measSpikeIdx  = 1}^L \!\measAmplitude{\measSpikeIdx}\overline{\window(t_\measSpikeIdx - \tau)} e^{-2\pi ikt_\measSpikeIdx} \!=\!\sum_{n \in \Z} \measurements_{k, n}e^{2\pi in\tau}\!,
\end{align*}
where $\measurements_{k, n} \triangleq \sum_{\measSpikeIdx  = 1}^L \measAmplitude{\measSpikeIdx} \window_n e^{-2\pi i(n+k)t_\measSpikeIdx}$ is the $n$th Fourier series coefficient of $\tau \mapsto \measurements(\tau, k)$. Using Parseval's theorem, the objective function in $\dualGaborProb$ can be rewritten as
\begin{equation*}
	\innerProd{c}{\measurements} = \Re{\sum_{k = -K}^K \sum_{n \in \Z} c_{k,n}\overline{\measurements_{k,n}}},
\end{equation*}
where $c_{k,n}$ denotes the $n$th Fourier series coefficient of the function $\tau \mapsto c(\tau, k)$, for $k \in \{-K, \ldots, K\}$, and $c \in L^\infty(\torus \times \Z)$ is the optimization variable of $\dualGaborProb$.
For $c \in L^\infty(\torus \times \Z)$, the function $\gaborOp^*c$ is given by
\begin{align}
	\forall t \in \torus, \quad (\gaborOp^*c)(t) &= \sum_{k = -K}^K \int_0^1 c(\tau, k)\window(t - \tau)e^{2\pi i kt}\mathrm{d}\tau \notag \\
		&= \sum_{k = -K}^K \int_0^1 c(\tau, k)\left(\sum_{n \in \Z} \window_n e^{2\pi i n(t - \tau)}\right) e^{2\pi i kt} \mathrm{d}\tau \notag \\
		&= \sum_{k = -K}^K\sum_{n \in \Z}  \window_n \left(\int_0^1 c(\tau, k)e^{-2\pi i n\tau}\mathrm{d}\tau\right) e^{2\pi i(k + n)t} \notag \\
		&= \sum_{k = -K}^K\sum_{n \in \Z} \window_n c_{k, n} e^{2\pi i(k + n)t}. \label{eq: infinite dimensional contraint in the predual problem}
\end{align}
Since infinitely many coefficients $c_{k, n}$ are involved in the expression \eqref{eq: infinite dimensional contraint in the predual problem} of $\gaborOp^*c$, the feasible set $\{c \in L^\infty(\torus \times \Z) \colon \normInfty{\gaborOp^* c} \leq 1\}$ of $\dualGaborProb$ is infinite-dimensional.
We now approximate this infinite-dimensional problem by retaining only $2N+1$ coefficients in the Fourier series expansion of $\window$, that is, we replace $\window$ by 
\begin{equation*}
	\forall t \in \torus, \quad \widetilde{\window}_N(t) =  \sum_{n = -N}^{N} \window_n e^{2\pi int},
\end{equation*}
where $N$ is chosen large enough for the coefficients $\window_n = \sqrt{2\sigma} \exp\left(-2\pi\sigma^2n^2\right)$, $\abs{n} \geq N$, to be ``small''. 
The problem $\gaborProbN$ is now defined as 
\begin{equation*}
	\gaborProbN \qquad \minimize_{\measProb \in \measSpace{\torus}} \normTV{\measProb} \st \measurements = \mathcal{A}_{\widetilde{\window}_N}\measProb.
\end{equation*}
As Theorems \ref{thm: existence} and \ref{thm: fenchel dual gabor prob} hold for every $\window \in \schwartzSpace{\torus}$ and we have $\widetilde{\window}_N \in \schwartzSpace{\torus}$, they can be applied to conclude that $\gaborProbN$ always has a solution. The Fenchel predual problem of $\gaborProbN$ is given by
\begin{equation*}
	\dualGaborProbN \qquad \maximize_{c \in L^\infty(\torus \times \Z)} \innerProd{c}{\measurements} \st \normInfty{\mathcal{A}_{\widetilde{\window}_N}^*c} \leq 1,
\end{equation*}
where the measurements now become $\measurements(\tau, k) = (\mathcal{A}_{\window_{N_m}}\measProb_{N_m})(\tau, k)$, for $\tau \in \torus$ and $k \in \Z$.
Moreover, strong duality holds, implying that
\begin{align}
	\min\Big\{\normTV{\measProb} \colon \mathcal{A}_{\widetilde{\window}_N}\measProb = \measurements, \measProb \in \measSpace{\torus}\Big\} = \sup\Big\{\innerProd{c}{\measurements} \colon \normInfty{\mathcal{A}_{\widetilde{g}_N}^*c} \leq 1, c \in L^\infty(\torus \times \Z)\Big\}.
	\label{eq: strong duality torus}
\end{align}
Now, we make the problem $\dualGaborProbN$ explicit.
The objective of $\dualGaborProbN$ can be written as
\begin{equation*}
	\innerProd{c}{y} = \innerProd{\mathbf{C}}{\mathbf{Y}} = \Re{\sum_{k = -K}^K\sum_{n = -N}^N c_{k, n}\overline{y_{k, n}}},
\end{equation*}
where $\mathbf{C} \triangleq (c_{k, n})_{\abs{k} \leq K, \abs{n} \leq N}$ and $\mathbf{Y} \triangleq (\measurements_{k, n})_{\abs{k} \leq K,\abs{n} \leq N}$.
The function $\mathcal{A}_{\widetilde{\window}_N}^*c$ is a trigonometric polynomial (hereafter referred to as dual polynomial), entirely characterized by the finite set of coefficients $\{c_{k, n}\}_{\abs{k}\leq K, \abs{n} \leq N}$, and expressed as
\begin{equation}
	\forall t \in \torus, \quad (\mathcal{A}_{\widetilde{\window}_N}^*c)(t) = \sum_{k = -K}^K\sum_{n = -N}^N \window_nc_{k,n}e^{2\pi i(k+n)t} = \sum_{m = -(K+N)}^{K+N} x_me^{2\pi imt}
	\label{eq: recast SDP trigonometric polynomial}
\end{equation}
where
\begin{equation*}
	\begin{array}{lll}
	 x_m = \displaystyle\sum_{n = n_\mathrm{min}}^{n_\mathrm{max}} \window_mc_{m -n, n}, & \qquad \text{with}\quad  & n_\mathrm{min} \triangleq \max\{-N, m-K\} \\
	 & & n_\mathrm{max} \triangleq \min\{N, m+K\}.
	\end{array}
\end{equation*}

The problem $\dualGaborProbN$ thus takes on the form
\begin{equation*}
	\maximize_{\mathbf{C} \in \C^{(2K+1) \times (2N+1)}} \innerProd{\mathbf{C}}{\mathbf{Y}} \st \left\{ \begin{array}{l} \displaystyle \normInfty{\sum_{m = -(K+N)}^{K+N} x_me^{2\pi imt}} \leq 1 \\
	x_m = \sum\limits_{n = n_\mathrm{min}}^{n_\mathrm{max}} \window_nc_{m - n, n}.
	\end{array} \right.
\end{equation*}

Now, as in \cite[Sec.~4, pp.~31--36]{Candes2012}, we make use of the following theorem to recast $\dualGaborProbN$ as a semi-definite program.

\begin{thm}[\!\!{\cite[Cor.~4.25]{Dumitrescu2007}}]
	\label{thm: dumitrescu}
	Let $M \in \N$ and let $P$ be a trigonometric polynomial
	\begin{equation*}
		\forall t \in \torus, \quad P(t) = \sum_{m = -M}^{M} p_{m} e^{2\pi im t}.
	\end{equation*}
	Then, $\normInfty{P} \leq 1$ if and only if there exists a Hermitian matrix $\mathbf{Q} \in \C^{(2M+1) \times (2M+1)}$ such that 
	\begin{equation*}
		\begin{pmatrix} \mathbf{Q} & \mathbf{p} \\ \mathbf{p}^H & 1 \end{pmatrix} \succeq \mathbf{0} \qquad \text{and} \qquad \sum_{m = 1}^{2M+1} q_{m, m+\ell} = \begin{cases} 1, & \ell = 0 \\ 0, & \ell \in \{1, \ldots, 2M\},\end{cases}
	\end{equation*}
	where $\mathbf{p} \in \C^{2M+1}$ is the column vector whose $k$th element is $p_{k-M-1}$ for $k \in \{1, \ldots, 2M+1\}$.
\end{thm}

From \eqref{eq: recast SDP trigonometric polynomial} and Theorem \ref{thm: dumitrescu} it then follows that $\dualGaborProbN$ is equivalent to
\begin{equation}
	\maximize_{\substack{\mathbf{C} \in \C^{(2K+1) \times (2N+1)}\\ \mathbf{Q} \in \C^{M'\times M'}}} \innerProd{\mathbf{Y}}{\mathbf{C}} \st \left\{ \begin{array}{l} x_m = \sum\limits_{n = n_\mathrm{min}}^{n_\mathrm{max}} \window_nc_{m - n, n} \\
		\begin{pmatrix} \mathbf{Q} & \mathbf{x} \\
		\mathbf{x}^H & 1 \end{pmatrix} \succeq 0 \\ \sum\limits_{k = 0}^{M'-\ell} q_{k, k + \ell} = \delta_{0, \ell},
	\end{array} \right.
	\label{eq: equivalent problem predual problem}
\end{equation}
where $M' \triangleq 2(K+N)+1$, $\mathbf{x}$ is the column vector whose $k$th element is $x_{k-K-N-1}$ for $k \in \{1, \ldots, M'\}$.
We can now solve $\dualGaborProbN$ and recover the corresponding measure by polynomial root finding as done in \cite[Sec.~4, pp.~31--36]{Candes2012}.

The following result shows that the sequence of solutions of $\gaborProbN$ converges in the weak-* sense to $\meas$ as $N \rightarrow \infty$, provided that $\meas$ is the unique solution of $\gaborProb$.

\begin{prop}
	Let $\meas \triangleq \sum_{\measSpikeIdx = 1}^L \measAmplitude{\measSpikeIdx} \delta_{t_\measSpikeIdx} \in \measSpace{\torus}$ with $\measAmplitude{\measSpikeIdx} \in \C \setminus \{0\}$, for all $\measSpikeIdx \in \{1, 2, \ldots, L\}$.  
	If $\meas$ is the unique solution of $\gaborProb$, then every sequence $\{\measProb_N\}_{N \in \N}$ such that for every $N \in \N$, $\measProb_N$ is in the set of solutions of $\gaborProbN$
	converges in the weak-* sense to $\meas$.
	\label{prop: convergence solution}
\end{prop}

\begin{proof}
	Let $N \in \N$. By \eqref{eq: strong duality torus} we have
	\begin{equation*}
		\normTV{\measProb_N} = \sup\left\{\innerProd{c}{\measurements} \colon \normInfty{\mathcal{A}_{\widetilde{\window}_N}^*c} \leq 1, c \in L^\infty(\torus \times \Z)\right\}.
	\end{equation*}
	For $c \in L^\infty(\torus \times \Z)$, thanks to \eqref{eq: recast SDP trigonometric polynomial}, $(\mathcal{A}_{\widetilde{\window}_N}^*c)(t)$ depends only on the Fourier series coefficients $c_{k, n}$, $n \in \{-N, \ldots, N\}$, of the functions $\tau \mapsto c(\tau, k)$, for $k \in \{-K, \ldots, K\}$. The problem $\dualGaborProbN$ is therefore equivalent to
	\begin{equation*}
		\maximize_{c \in \mathcal{P}_N(\torus \times \Z)} \innerProd{c}{\measurements} \st \normInfty{\gaborOp^*c} \leq 1,
	\end{equation*}
	where
	\begin{equation*}
		\mathcal{P}_N(\torus \times \Z) \triangleq \left\{(\tau, k) \mapsto \begin{cases} \displaystyle \sum_{n = -N}^N c_{k,n} e^{2\pi int}, & \abs{k} \leq K \\ 0, & \text{otherwise}\end{cases} \colon \{c_{k,n}\}_{\substack{-K \leq k \leq K \\ -N \leq n \leq N}} \in \C^{(2K+1) \times (2N+1)}\right\}.
	\end{equation*}
	But since
	\begin{align*}
		\left\{\innerProd{c}{\measurements} \colon \normInfty{\gaborOp^*c} \leq 1, c \in \mathcal{P}_N(\torus \times \Z)\right\} &\subseteq \left\{\innerProd{c}{\measurements} \colon \normInfty{\gaborOp^*c} \leq 1, c \in \mathcal{P}_{N+1}(\torus \times \Z)\right\} \\
			&\hspace{-2.5cm}\subseteq \left\{\innerProd{c}{\measurements} \colon \normInfty{\gaborOp^*c} \leq 1, c \in L^\infty(\torus \times \Z)\right\},
	\end{align*}
	we have that
	\begin{align}
		\normTV{\measProb_{N}} &= \sup\left\{\innerProd{c}{\measurements} \colon \normInfty{\gaborOp^*c} \leq 1, c \in \mathcal{P}_N(\torus \times \Z)\right\} \notag \\
			&\leq \sup\left\{\innerProd{c}{\measurements} \colon \normInfty{\gaborOp^*c} \leq 1, c \in \mathcal{P}_{N+1}(\torus \times \Z)\right\} = \normTV{\measProb_{N+1}} \label{eq: proof convergence problem simulations 1}
	\end{align}
	as well as
	\begin{align}
		\normTV{\measProb_N} &= \sup\left\{\innerProd{c}{\measurements} \colon \normInfty{\gaborOp^*c} \leq 1, c \in \mathcal{P}_N(\torus \times \Z)\right\} \notag \\
			&\leq \sup\left\{\innerProd{c}{\measurements} \colon \normInfty{\gaborOp^*c} \leq 1, c \in L^\infty(\torus \times \Z)\right\} \notag \\
			&= \min\left\{\normTV{\measProb} \colon  \gaborOp\measProb = \measurements, \measProb \in \measSpace{\torus}\right\} = \normTV{\meas}. \label{eq: proof convergence problem simulations 2}
	\end{align}
	From \eqref{eq: proof convergence problem simulations 2} it follows that the sequence $\{\measProb_N\}_{N \in \N}$ is bounded. Therefore, by application of~\cite[Cor.~3.30]{Brezis2010}, there exists a subsequence $\{\measProb_{N_m}\}_{m \in \N}$ that converges in weak-* topology to a measure $\measProb_\infty \in \measSpace{\torus}$. From \eqref{eq: proof convergence problem simulations 1} and \eqref{eq: proof convergence problem simulations 2} if follows that $\{\normTV{\measProb_N}\}_{N \in \N}$ is convergent. Thanks to \cite[Thm.~3.13 (iii)]{Brezis2010} and \eqref{eq: proof convergence problem simulations 2}, we then get
	\begin{equation*}
		\normTV{\measProb_\infty} \leq \lim_{m \rightarrow \infty} \normTV{\measProb_{N_m}} \leq \normTV{\meas}.
	\end{equation*}
	Next, we show that
	\begin{equation*}
		\lim_{m \rightarrow \infty} (\gaborOp\measProb_{N_m})(\tau, k) = y(\tau, k).
	\end{equation*}
	Let $\tau \in \torus$ and $k \in \{-K, \ldots, K\}$. For $m \in \N$, we have
	\begin{align}
		\abs{\measurements(\tau, k) - (\gaborOp\measProb_{N_m})(\tau, k)} &= \abs{(\mathcal{A}_{\window_{N_m}}\measProb_{N_m})(\tau, k) -  (\gaborOp\measProb_{N_m})(\tau, k)} \notag \\
			&= \abs{\int_\torus \window_{N_m}(t - \tau) e^{-2\pi ikt}\dmeas{\measProb_{N_m}}{t} - \int_\torus \window(t - \tau) e^{-2\pi ikt} \dmeas{\measProb_{N_m}}{t}} \notag \\
			&= \abs{\int_\torus \big(\window_{N_m}(t - \tau) - \window(t - \tau)\big) e^{-2\pi ikt} \dmeas{\measProb_{N_m}}{t}} \notag \\
			&\leq \int_\torus \abs{\window_{N_m}(t - \tau) - \window(t - \tau)} \dmeas{\abs{\measProb_{N_m}}}{t} \notag \\
			&\leq \normInfty{\window_{N_m} - \window} \normTV{\measProb_{N_m}}. \label{eq: proof convergence problem simulations 3}
	\end{align}
	From
	\begin{equation*}
		\normInfty{\window_{N_m} - \window} \leq 2\sum_{n = N_m+1}^{\infty} \sqrt{2\sigma}\exp(-2\pi\sigma^2n^2),
	\end{equation*}
	it follows that $\lim_{m \rightarrow \infty} \normInfty{\window_{N_m} - \window} = 0$. Since $\{\normTV{\measProb_{N_m}}\}_{m \in \N}$ is also convergent, \eqref{eq: proof convergence problem simulations 3} converges to $0$ as $m \rightarrow \infty$. It therefore follows that
	\begin{equation*}
		\lim_{m \rightarrow \infty} (\gaborOp\measProb_{N_m})(\tau, k) = y(\tau, k).
	\end{equation*}
	But since $\{\measProb_{N_m}\}_{m \in \N}$ converges to $\measProb_\infty$ in the weak-* topology and
	\begin{equation*}
		(\gaborOp\measProb_{N_m})(\tau, k)  = \int_\torus \window(t - \tau)e^{-2\pi ikt}\dmeas{\measProb_{N_m}}{t},
	\end{equation*}
	we have, by definition of weak-* convergence,
	\begin{equation*}
		\lim_{m \rightarrow \infty} (\gaborOp\measProb_{N_m})(\tau, k) = (\gaborOp\measProb_\infty)(\tau, k).
	\end{equation*}
	We therefore get $\measurements(\tau, k) = (\gaborOp\measProb_\infty)(\tau, k)$, for all $\tau \in \torus$ and $k \in \{-K, \ldots, K\}$, which shows that $\measProb_\infty$ is feasible for the problem $\gaborProb$. As $\meas$ is the unique solution of $\gaborProb$, we have $\normTV{\meas} \leq \normTV{\measProb_\infty}$, which combined with $\normTV{\measProb_\infty} \leq \normTV{\meas}$ leads to $\normTV{\measProb_\infty} = \normTV{\meas}$, and hence shows that $\measProb_\infty$ is a solution of $\gaborProb$. However, by assumption $\meas$ is the unique solution of $\gaborProb$. We can therefore conclude that not only $\normTV{\measProb_\infty} = \normTV{\meas}$ but also $\measProb_\infty = \meas$.
	Since there was nothing specific about the accumulation point $\measProb_\infty$, we can apply the same line of arguments to every accumulation point of the sequence $\{\measProb_N\}_{N \in \N}$, and therefore conclude that every accumulation point of $\{\measProb_N\}_{N \in \N}$ must equal $\meas$. In summary, we have shown that $\meas$ is the unique accumulation point---in the \mbox{weak-*} topology---of the sequence $\{\measProb_N\}_{N \in \N}$.
	But since $\normTV{\measProb_N} \leq \normTV{\meas}$ by \eqref{eq: proof convergence problem simulations 2}, the sequence $\{\measProb_N\}_{N \in \N}$ is contained in the closed centered ball of radius $\normTV{\meas}$, which, by the Banach-Alaoglu Theorem, is compact in the weak-* topology. We next show that this implies weak-* convergence of $\{\measProb_N\}_{N \in \N}$ to $\meas$. Suppose by way of contradiction that $\{\measProb_N\}_{N \in \N}$ does not converge---in the weak-* topology---to $\meas$. Then, there exists an $\varepsilon > 0$ such that for infinitely many $N$ we have $\abs{\innerProd{\measProb_N - \meas}{\varphi}} > \varepsilon$, for all $\varphi \in C(\torus)$. Therefore, we can find a subsequence of $\{\measProb_N\}_{N \in \N}$ which, by compactness---in the weak-* topology---of the closed centered ball of radius $\normTV{\meas}$, has an accumulation point different from $\meas$ in the weak-* topology. This constitutes a contradiction and thereby finishes the proof.
	\end{proof}

Fix $N \in \N$. If $\dualGaborProbN$ has at least one solution $c \in L^\infty(\torus \times \Z)$ such that $\abs{\mathcal{A}_{\widetilde{\window}_N}^*c}$ is not identically equal to $1$, it follows by \eqref{eq: implication dual problem} that, as a consequence of $\mathcal{A}_{\widetilde{\window}_N}^*c$ being a trigonometric polynomial, every solution $\measProb_N$ of $\gaborProbN$ is discrete.  
Unfortunately, the \mbox{weak-*} convergence alone of the sequence $\{\measProb_N\}_{N \in \N}$  of discrete measures to the discrete measure $\meas$, as guaranteed by Proposition \ref{prop: convergence solution}, does not imply, in general, that each element of $\supp(\measProb_N)$ converges to an element of $\supp(\meas)$. 
What we can show, however, is that for small enough $\varepsilon > 0$, one can find an $M_\varepsilon > 0$ such that for all $N \geq M_\varepsilon$, each set $[t_\measSpikeIdx -\varepsilon, t_\measSpikeIdx+\varepsilon]$, $\measSpikeIdx \in \{1, 2, \ldots, L\}$, contains at least one point of $\supp(\measProb_N)$, and in addition, most of the ``energy (in TV norm)'' of $\measProb_N$ is contained in $T_\varepsilon \triangleq \bigcup_{n \in \Z} \bigcup_{\measSpikeIdx = 1}^L\  [t_\measSpikeIdx + n - \varepsilon, t_\measSpikeIdx + n + \varepsilon]$. 
This is formalized as follows.

\begin{prop}
	\label{prop: proposition weak-* convergence support}
	Let $\meas = \sum_{\measSpikeIdx = 1}^L \measAmplitude{\measSpikeIdx}\delta_{t_\measSpikeIdx} \in \measSpace{\torus}$ with $\measAmplitude{\measSpikeIdx} \in \C \!\setminus \{0\}$, for all $\measSpikeIdx \in \{1, 2, \ldots, L\}$. Assume that the wrap-around distance
	\begin{equation*}
		\Delta \triangleq \min_{n \in \Z}\min_{\substack{1 \leq \measSpikeIdx, \measSpikeIdxBis \leq L \\ \measSpikeIdx \neq \measSpikeIdxBis}} \abs{t_\measSpikeIdx - t_\measSpikeIdxBis + n} > 0.
	\end{equation*}
	If $\meas$ is the unique solution of $\gaborProb$, then every sequence $\{\measProb_N\}_{N \in \N}$ such that for every $N \in \N$, $\measProb_N$ is contained in the set of solutions of $\gaborProbN$, satisfies
	\begin{equation}
		\forall \varepsilon \in (0, \Delta/4], \ \exists M_\varepsilon > 0, \ \forall N \geq M_\varepsilon, \ \forall \measSpikeIdx \in \{1, 2, \ldots, L\}, \quad \supp(\measProb_N) \cap [t_\measSpikeIdx - \varepsilon, t_\measSpikeIdx + \varepsilon] \neq \emptyset.
		\label{eq: support not empty near tell}
	\end{equation}
	Moreover, setting
	\begin{align*}
		\forall \varepsilon \in (0, \Delta/4], \qquad &T_\varepsilon \triangleq \bigcup_{n \in \Z} \bigcup_{\measSpikeIdx = 1}^L\  [t_\measSpikeIdx - \varepsilon + n, t_\measSpikeIdx + \varepsilon + n] \\
			&T_\varepsilon^c \triangleq \torus \setminus T_\varepsilon,
	\end{align*}
	and defining $\measProb_{N, T_\varepsilon}, \measProb_{N, T_\varepsilon^c} \in \measSpace{\torus}$ according to  
	\begin{align*}
		\forall B \in \borelAlgebra{\torus}, \qquad & \measProb_{N, T_\varepsilon}(B) \triangleq \measProb_N(T_\varepsilon \cap B) \\
		&\measProb_{N, T_\varepsilon^c}(B) \triangleq \measProb_N(T_\varepsilon^c \cap B),
	\end{align*}
	we have
	\begin{equation}
		\forall \varepsilon \in (0, \Delta/4], \ \exists M_\varepsilon > 0, \ \forall N \geq M_\varepsilon, \quad \left\{ \begin{array}{ll} \normTV{\measProb_{N, T_\varepsilon^c}} &\leq \varepsilon\normTV{\meas} \\
			\normTV{\measProb_{N, T_\varepsilon}} &\geq (1 - \varepsilon) \normTV{\meas}.
			\end{array} \right.
		\label{eq: concentration inequality measure}
	\end{equation}
\end{prop}

Some remarks are in order before we prove Proposition~\ref{prop: proposition weak-* convergence support}. 
An obvious consequence of \eqref{eq: support not empty near tell} is that the number of atoms of $\measProb_N$ is larger than (or equal to) the number $L$ of atoms of $\meas$. Moreover, for all $N \geq M_\varepsilon$ and $t^{(N)} \in \supp(\measProb_N) \cap T_\varepsilon^c$, we have
\begin{align}
	\abs{\measProb_N(\{t^{(N)}\})} &\leq \abs{\measProb_N}\!(\{t^{(N)}\}) \notag \\
		&= \abs{\measProb_{N, T_\varepsilon^c}}\!(\{t^{(N)}\}) \label{eq: inequality to explain proof magnitude small outside T epsilon} \\
		&\leq \normTV{\measProb_{N, T_\varepsilon^c}} \label{eq: inequality to explain proof magnitude small outside T epsilon bis} \\
		&\leq \varepsilon\normTV{\meas}, \notag
\end{align}
where \eqref{eq: inequality to explain proof magnitude small outside T epsilon} follows from $t^{(N)} \in T_\varepsilon^c$ and \eqref{eq: inequality to explain proof magnitude small outside T epsilon bis} is by $\normTV{\measProb_{N, T_\varepsilon^c}} = \sum_{t \in \supp(\measProb_N) \cap T_\varepsilon^c} \abs{\measProb_{N, T_\varepsilon^c}(\{t\})}$.
This means that for sufficiently large $N$, the weights of $\measProb_N$ attached to the spikes located outside $T_\varepsilon$ are smaller than $\varepsilon\normTV{\meas}$, provided that $\varepsilon \in (0, \Delta/4]$. 
Note, however, that each of the sets $[t_\ell-\varepsilon, t_\ell+\varepsilon]$, $\ell \in \{1, 2, \ldots, L\}$, may contain more than one atom of $\nu_N$. Fig.~\ref{fig: proof last proposition} illustrates the statement in Proposition~\ref{prop: proposition weak-* convergence support}. 

\definecolor{lightgray}{gray}{0.5}
\definecolor{darkgray}{gray}{0.5}


\begin{figure}
	\centering
	\ifthenelse{\equal{\versionColor}{true}}{
		\ifthenelse{\equal{\pdfFigs}{true}}{
			\includegraphics[width = 0.9\textwidth]{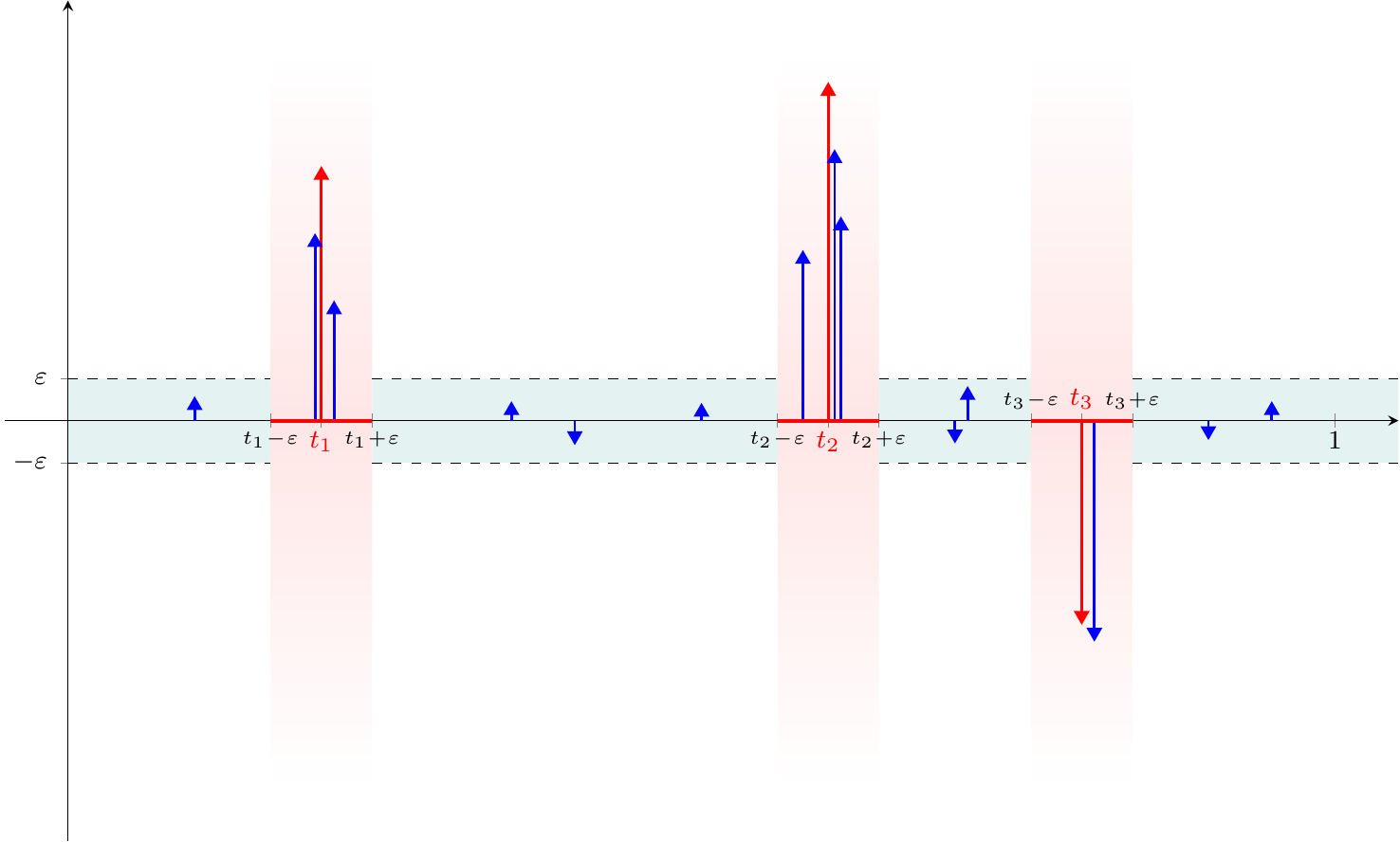}
		}{
			\ifthenelse{\equal{\generateFigs}{yes}}{
				\tikzsetnextfilename{illustrationPropositionColor}
				\begin{tikzpicture}
	                        		\tikzstyle{every node}=[font=\scriptsize]
	                        		\begin{axis}[
	                        			axis x line = middle, axis y line = middle,
	                        			xmin = -0.05,xmax = 1.05,
	                        			ymin = -2.5, ymax = 2.5, 
	                        			grid = none, 
	                        			height = 9cm, width = 0.9\textwidth, 
	                        			scale only axis, 
	                        			enlarge x limits = false, enlarge y limits = false,
	                        			ytick = {-0.25, 0, 0.25},
	                        			yticklabels = {$-\varepsilon$, $0$, $\varepsilon$}, 
	                        			xtick = {0, 0.16, 0.2, 0.24, 0.56, 0.6, 0.64, 1},
	                        			xticklabels = {$0$, {\tiny$t_1\!-\!\varepsilon$}, \textcolor{red}{$t_1$}, {\tiny $t_1\!+\!\varepsilon$}, {\tiny$t_2\!-\!\varepsilon$}, \textcolor{red}{$t_2$}, {\tiny $t_2\!+\!\varepsilon$}, $1$},
	                        			xticklabel shift = {-0.1cm},
	                        			extra x ticks = {0.76, 0.8, 0.84},
	                        			extra x tick labels = {{\tiny $t_3\!-\!\varepsilon$}, \textcolor{red}{$t_3$}, {\tiny $t_3\!+\!\varepsilon$}},
	                        			extra x tick style = {
	                        				xticklabel style={yshift = -0.05cm, anchor = south}
	                        			},
	                        			legend columns = 2, legend cell align  = left, legend pos = north west, 
	                        			]
	                        			\addplot+[thick, dirac, red] coordinates {(0.2, 1.5) (0.6, 2) (0.8, -1.2)};
	                        			\addplot+[thick, dirac, blue] coordinates {(0.195, 1.1) (0.58, 1) (0.81, -1.3)};
	                        			\addplot+[thick, dirac, blue] coordinates {(0.21, 0.7) (0.605, 1.6) (0.61, 1.2)};
	                        			\addplot+[thick, dirac, blue] coordinates {(0.1, 0.13) (0.35, 0.1) (0.4, -0.13) (0.5, 0.09) (0.7, -0.12) (0.71, 0.19) (0.9, -0.1) (0.95, 0.1)};

	                        			\addplot[name path=A1, very thin, dashed, black, no markers, domain = 0:0.16] {0.25};
	                        			\addplot[name path=B1, very thin, dashed, black, no markers, domain = 0:0.16] {-0.25};
	                        			
	                        			\addplot[name path=A2, very thin, dashed, black, no markers, domain = 0.24:0.56] {0.25};
	                        			\addplot[name path=B2, very thin, dashed, black, no markers, domain = 0.24:0.56] {-0.25};
	                        			
	                        			\addplot[name path=A3, very thin, dashed, black, no markers, domain = 0.64:0.76] {0.25};
	                        			\addplot[name path=B3, very thin, dashed, black, no markers, domain = 0.64:0.76] {-0.25};
	                        			
	                        			\addplot[name path=A4, very thin, dashed, black, no markers, domain = 0.84:1.3] {0.25};
	                        			\addplot[name path=B4, very thin, dashed, black, no markers, domain = 0.84:1.3] {-0.25};
	                        			
	                        			\addplot[name path=C, very thin, dashed, white, no markers, domain = 0.16:0.24] {2.2};
	                        			\addplot[name path=CD, very thin, dashed, white, no markers, domain = 0.16:0.24] {0};
	                        			\addplot[name path=D, very thin, dashed, white, no markers, domain = 0.16:0.24] {-2.2};
	                        			
	                        			\addplot[name path=E, very thin, dashed, white, no markers, domain = 0.56:0.64] {2.2};
	                        			\addplot[name path=EF, very thin, dashed, white, no markers, domain = 0.56:0.64] {0};
	                        			\addplot[name path=F, very thin, dashed, white, no markers, domain = 0.56:0.64] {-2.2};
	                        			
	                        			\addplot[name path=G, very thin, dashed, white, no markers, domain = 0.76:0.84] {2.2};
	                        			\addplot[name path=GH, very thin, dashed, white, no markers, domain = 0.76:0.84] {0};
	                        			\addplot[name path=H, very thin, dashed, white, no markers, domain = 0.76:0.84] {-2.2};
	                        			
	                        			\addplot[very thick, red, no markers, domain=0.16:0.24] {0};
	                        			\addplot[very thick, red, no markers, domain=0.56:0.64] {0};
	                        			\addplot[very thick, red, no markers, domain=0.76:0.84] {0};
	                        			
	                        			\addplot[teal!10] fill between[of=A1 and B1,
	                        				soft clip={domain=0:1},
	                        			];
	                        			
	                        			\addplot[teal!10] fill between[of=A2 and B2,
	                        				soft clip={domain=0:1},
	                        			];
	                        			
	                        			\addplot[teal!10] fill between[of=A3 and B3,
	                        				soft clip={domain=0:1},
	                        			];
	                        			
	                        			\addplot[teal!10] fill between[of=A4 and B4,
	                        				soft clip={domain=0:1.3},
	                        			];
	                        			
	                        			\addplot[white, bottom color=red!10, top color=red!0] fill between[of=C and CD,
	                        				soft clip={domain=0:1},
	                        			];
	                        			
	                        			\addplot[white, bottom color=red!0, top color=red!10] fill between[of=CD and D,
	                        				soft clip={domain=0:1},
	                        			];
	                        			
	                        			\addplot[white, bottom color=red!10, top color=red!0] fill between[of=E and EF,
	                        				soft clip={domain=0:1},
	                        			];
	                        			
	                        			\addplot[white, bottom color=red!0, top color=red!10] fill between[of=EF and F,
	                        				soft clip={domain=0:1},
	                        			];
	                        			
	                        			\addplot[white, bottom color=red!10, top color=red!0] fill between[of=G and GH,
	                        				soft clip={domain=0:1},
	                        			];
	                        			
	                        			\addplot[white, bottom color=red!0, top color=red!10] fill between[of=GH and H,
	                        				soft clip={domain=0:1},
	                        			];
	                        		\end{axis}
				\end{tikzpicture}			
			}{
	                        	\begin{tikzpicture}
	                        		\tikzstyle{every node}=[font=\scriptsize]
	                        		\begin{axis}[
	                        			axis x line = middle, axis y line = middle,
	                        			xmin = -0.05,xmax = 1.05,
	                        			ymin = -2.5, ymax = 2.5, 
	                        			grid = none, 
	                        			height = 9cm, width = 0.9\textwidth, 
	                        			scale only axis, 
	                        			enlarge x limits = false, enlarge y limits = false,
	                        			ytick = {-0.25, 0, 0.25},
	                        			yticklabels = {$-\varepsilon$, $0$, $\varepsilon$}, 
	                        			xtick = {0, 0.16, 0.2, 0.24, 0.56, 0.6, 0.64, 1},
	                        			xticklabels = {$0$, {\tiny$t_1\!-\!\varepsilon$}, \textcolor{red}{$t_1$}, {\tiny $t_1\!+\!\varepsilon$}, {\tiny$t_2\!-\!\varepsilon$}, \textcolor{red}{$t_2$}, {\tiny $t_2\!+\!\varepsilon$}, $1$},
	                        			xticklabel shift = {-0.1cm},
	                        			extra x ticks = {0.76, 0.8, 0.84},
	                        			extra x tick labels = {{\tiny $t_3\!-\!\varepsilon$}, \textcolor{red}{$t_3$}, {\tiny $t_3\!+\!\varepsilon$}},
	                        			extra x tick style = {
	                        				xticklabel style={yshift = -0.05cm, anchor = south}
	                        			},
	                        			legend columns = 2, legend cell align  = left, legend pos = north west, 
	                        			]
	                        			\addplot+[thick, dirac, red] coordinates {(0.2, 1.5) (0.6, 2) (0.8, -1.2)};
	                        			\addplot+[thick, dirac, blue] coordinates {(0.195, 1.1) (0.58, 1) (0.81, -1.3)};
	                        			\addplot+[thick, dirac, blue] coordinates {(0.21, 0.7) (0.605, 1.6) (0.61, 1.2)};
	                        			\addplot+[thick, dirac, blue] coordinates {(0.1, 0.13) (0.35, 0.1) (0.4, -0.13) (0.5, 0.09) (0.7, -0.12) (0.71, 0.19) (0.9, -0.1) (0.95, 0.1)};

	                        			\addplot[name path=A1, very thin, dashed, black, no markers, domain = 0:0.16] {0.25};
	                        			\addplot[name path=B1, very thin, dashed, black, no markers, domain = 0:0.16] {-0.25};
	                        			
	                        			\addplot[name path=A2, very thin, dashed, black, no markers, domain = 0.24:0.56] {0.25};
	                        			\addplot[name path=B2, very thin, dashed, black, no markers, domain = 0.24:0.56] {-0.25};
	                        			
	                        			\addplot[name path=A3, very thin, dashed, black, no markers, domain = 0.64:0.76] {0.25};
	                        			\addplot[name path=B3, very thin, dashed, black, no markers, domain = 0.64:0.76] {-0.25};
	                        			
	                        			\addplot[name path=A4, very thin, dashed, black, no markers, domain = 0.84:1.3] {0.25};
	                        			\addplot[name path=B4, very thin, dashed, black, no markers, domain = 0.84:1.3] {-0.25};
	                        			
	                        			\addplot[name path=C, very thin, dashed, white, no markers, domain = 0.16:0.24] {2.2};
	                        			\addplot[name path=CD, very thin, dashed, white, no markers, domain = 0.16:0.24] {0};
	                        			\addplot[name path=D, very thin, dashed, white, no markers, domain = 0.16:0.24] {-2.2};
	                        			
	                        			\addplot[name path=E, very thin, dashed, white, no markers, domain = 0.56:0.64] {2.2};
	                        			\addplot[name path=EF, very thin, dashed, white, no markers, domain = 0.56:0.64] {0};
	                        			\addplot[name path=F, very thin, dashed, white, no markers, domain = 0.56:0.64] {-2.2};
	                        			
	                        			\addplot[name path=G, very thin, dashed, white, no markers, domain = 0.76:0.84] {2.2};
	                        			\addplot[name path=GH, very thin, dashed, white, no markers, domain = 0.76:0.84] {0};
	                        			\addplot[name path=H, very thin, dashed, white, no markers, domain = 0.76:0.84] {-2.2};
	                        			
	                        			\addplot[very thick, red, no markers, domain=0.16:0.24] {0};
	                        			\addplot[very thick, red, no markers, domain=0.56:0.64] {0};
	                        			\addplot[very thick, red, no markers, domain=0.76:0.84] {0};
	                        			
	                        			\addplot[teal!10] fill between[of=A1 and B1,
	                        				soft clip={domain=0:1},
	                        			];
	                        			
	                        			\addplot[teal!10] fill between[of=A2 and B2,
	                        				soft clip={domain=0:1},
	                        			];
	                        			
	                        			\addplot[teal!10] fill between[of=A3 and B3,
	                        				soft clip={domain=0:1},
	                        			];
	                        			
	                        			\addplot[teal!10] fill between[of=A4 and B4,
	                        				soft clip={domain=0:1.3},
	                        			];
	                        			
	                        			\addplot[path fading=north, red!10] fill between[of=C and CD,
	                        				soft clip={domain=0:1},
	                        			];
	                        			
	                        			\addplot[path fading=south, red!10] fill between[of=CD and D,
	                        				soft clip={domain=0:1},
	                        			];
	                        			
	                        			\addplot[path fading=north, red!10] fill between[of=E and EF,
	                        				soft clip={domain=0:1},
	                        			];
	                        			
	                        			\addplot[path fading=south, red!10] fill between[of=EF and F,
	                        				soft clip={domain=0:1},
	                        			];
	                        			
	                        			\addplot[path fading=north, red!10] fill between[of=G and GH,
	                        				soft clip={domain=0:1},
	                        			];
	                        			
	                        			\addplot[path fading=south, red!10] fill between[of=GH and H,
	                        				soft clip={domain=0:1},
	                        			];
	                        		\end{axis}
	                        	\end{tikzpicture}
			}
		}
		\caption{Illustration of Proposition~\ref{prop: proposition weak-* convergence support} with $L = 3$. The union of intervals depicted in red represents $T_\varepsilon$ on the fundamental interval $[0, 1)$. The spikes in red represent $\meas$. Proposition \ref{prop: proposition weak-* convergence support} guarantees that for $N \geq M_\varepsilon$, each interval $[t_\measSpikeIdx - \varepsilon, t_\measSpikeIdx + \varepsilon]$ contains at least one atom of $\measProb_N$.  The spikes in blue correspond to $\measProb_N$ for $N \geq M_\varepsilon$. The weights of the atoms of $\measProb_N$ outside $T_\varepsilon$ is guaranteed to be smaller than $\varepsilon\normTV{\meas}$.}
	}{
		\ifthenelse{\equal{\pdfFigs}{true}}{
			\includegraphics[width = 0.9\textwidth]{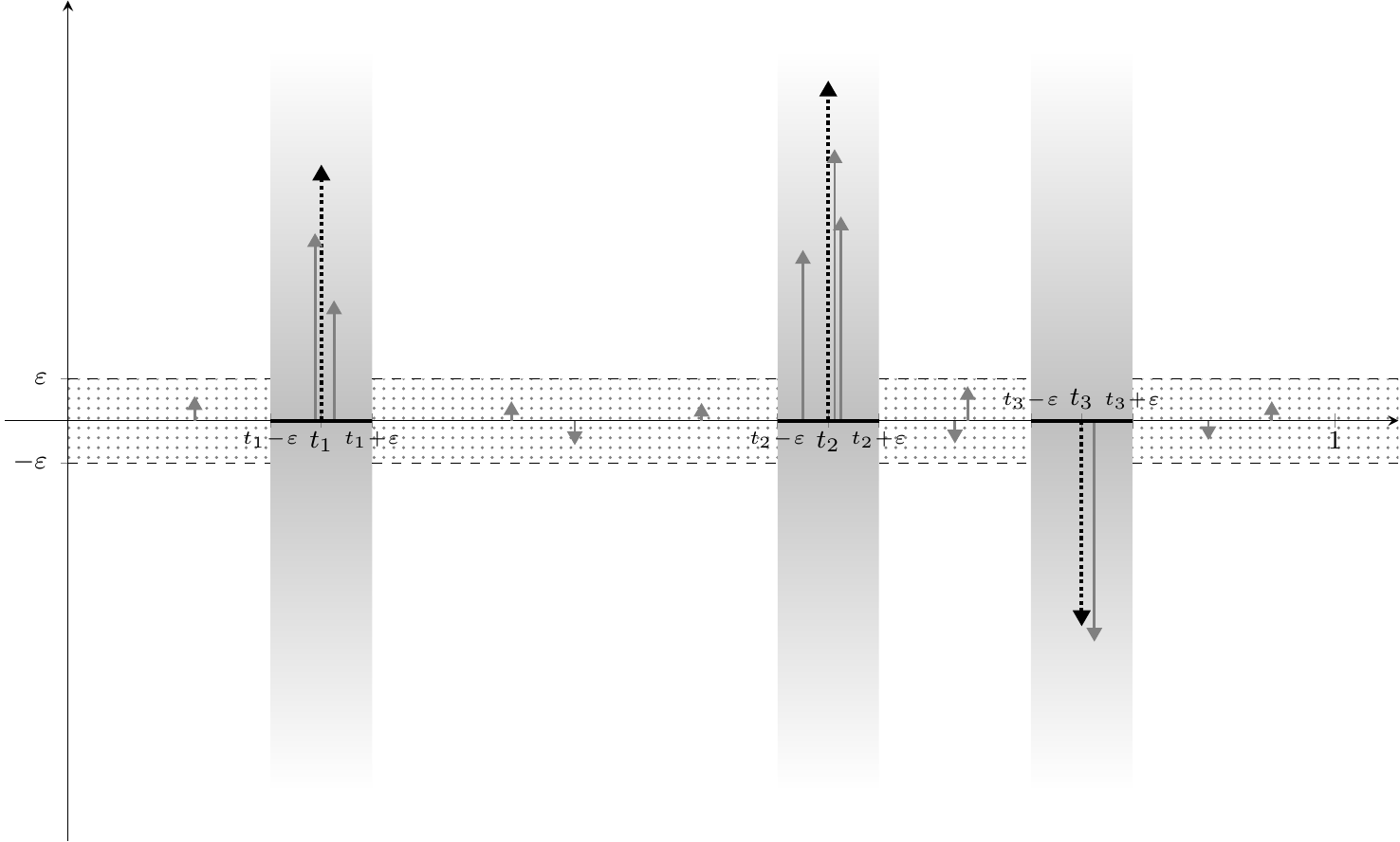}
			\caption{Illustration of Proposition~\ref{prop: proposition weak-* convergence support} with $L = 3$. The union of intervals corresponding to the shaded areas represents $T_\varepsilon$ on the fundamental interval $[0, 1)$. The black dotted spikes represent $\meas$. Proposition \ref{prop: proposition weak-* convergence support} guarantees that for $N \geq M_\varepsilon$, each interval $[t_\measSpikeIdx - \varepsilon, t_\measSpikeIdx + \varepsilon]$ contains at least one atom of $\measProb_N$.  The gray spikes correspond to $\measProb_N$ for $N \geq M_\varepsilon$. The weights of the atoms of $\measProb_N$ outside $T_\varepsilon$ is guaranteed to be smaller than $\varepsilon\normTV{\meas}$.}
		}{
			\ifthenelse{\equal{\generateFigs}{yes}}{
				\tikzsetnextfilename{illustrationPropositionBW}
		                	\begin{tikzpicture}
			                		\tikzstyle{every node}=[font=\scriptsize]
			                		\begin{axis}[
			                			axis x line = middle, axis y line = middle,
			                			xmin = -0.05,xmax = 1.05,
			                			ymin = -2.5, ymax = 2.5, 
			                			grid = none, 
			                			height = 9cm, width = 0.9\textwidth, 
			                			scale only axis, 
			                			enlarge x limits = false, enlarge y limits = false,
			                			ytick = {-0.25, 0, 0.25},
			                			yticklabels = {$-\varepsilon$, $0$, $\varepsilon$}, 
			                			xtick = {0, 0.16, 0.2, 0.24, 0.56, 0.6, 0.64, 1},
			                			xticklabels = {$0$, {\tiny$t_1\!-\!\varepsilon$}, {\textcolor{black}{$\boldmath t_1$}}, {\tiny $t_1\!+\!\varepsilon$}, {\tiny$t_2\!-\!\varepsilon$}, {\textcolor{black}{$t_2$}}, {\tiny $t_2\!+\!\varepsilon$}, $1$},
			                			xticklabel shift = {-0.1cm},
			                			extra x ticks = {0.76, 0.8, 0.84},
			                			extra x tick labels = {{\tiny $t_3\!-\!\varepsilon$}, {\textcolor{black}{$t_3$}}, {\tiny $t_3\!+\!\varepsilon$}},
			                			extra x tick style = {
			                				xticklabel style={yshift = -0.05cm, anchor = south}
			                			},
			                			legend columns = 2, legend cell align  = left, legend pos = north west, 
			                			]
			                			\addplot+[very thick, densely dotted, dirac, black] coordinates {(0.2, 1.5) (0.6, 2) (0.8, -1.2)};
			                			\addplot+[thick, dirac, darkgray] coordinates {(0.195, 1.1) (0.58, 1) (0.81, -1.3)};
			                			\addplot+[thick, dirac, darkgray] coordinates {(0.21, 0.7) (0.605, 1.6) (0.61, 1.2)};
			                			\addplot+[thick, dirac, darkgray] coordinates {(0.1, 0.13) (0.35, 0.1) (0.4, -0.13) (0.5, 0.09) (0.7, -0.12) (0.71, 0.19) (0.9, -0.1) (0.95, 0.1)};

			                			\addplot[name path=A1, very thin, dashed, black, no markers, domain = 0:0.16] {0.25};
			                			\addplot[name path=B1, very thin, dashed, black, no markers, domain = 0:0.16] {-0.25};
			                			
			                			\addplot[name path=A2, very thin, dashed, black, no markers, domain = 0.24:0.56] {0.25};
			                			\addplot[name path=B2, very thin, dashed, black, no markers, domain = 0.24:0.56] {-0.25};
			                			
			                			\addplot[name path=A3, very thin, dashed, black, no markers, domain = 0.64:0.76] {0.25};
			                			\addplot[name path=B3, very thin, dashed, black, no markers, domain = 0.64:0.76] {-0.25};
			                			
			                			\addplot[name path=A4, very thin, dashed, black, no markers, domain = 0.84:1.3] {0.25};
			                			\addplot[name path=B4, very thin, dashed, black, no markers, domain = 0.84:1.3] {-0.25};
			                			
			                			\addplot[name path=C, very thin, dashed, white, no markers, domain = 0.16:0.24] {2.2};
			                			\addplot[name path=CD, very thin, dashed, white, no markers, domain = 0.16:0.24] {0};
			                			\addplot[name path=D, very thin, dashed, white, no markers, domain = 0.16:0.24] {-2.2};
			                			
			                			\addplot[name path=E, very thin, dashed, white, no markers, domain = 0.56:0.64] {2.2};
			                			\addplot[name path=EF, very thin, dashed, white, no markers, domain = 0.56:0.64] {0};
			                			\addplot[name path=F, very thin, dashed, white, no markers, domain = 0.56:0.64] {-2.2};
			                			
			                			\addplot[name path=G, very thin, dashed, white, no markers, domain = 0.76:0.84] {2.2};
			                			\addplot[name path=GH, very thin, dashed, white, no markers, domain = 0.76:0.84] {0};
			                			\addplot[name path=H, very thin, dashed, white, no markers, domain = 0.76:0.84] {-2.2};
			                			
			                			\addplot[very thick, black, no markers, domain=0.16:0.24] {0};
			                			\addplot[very thick, black, no markers, domain=0.56:0.64] {0};
			                			\addplot[very thick, black, no markers, domain=0.76:0.84] {0};
			                			
			                			\addplot[pattern=dots, pattern color=lightgray] fill between[of=A1 and B1,
			                				soft clip={domain=0:1},
			                			];
			                			
			                			\addplot[pattern=dots, pattern color=lightgray] fill between[of=A2 and B2,
			                				soft clip={domain=0:1},
			                			];
			                			
			                			\addplot[pattern=dots, pattern color=lightgraypattern=dots, pattern color=lightgray] fill between[of=A3 and B3,
			                				soft clip={domain=0:1},
			                			];
			                			
			                			\addplot[pattern=dots, pattern color=lightgray] fill between[of=A4 and B4,
			                				soft clip={domain=0:1.3},
			                			];
			                			
			                			\addplot[white, bottom color=lightgray!50, top color=lightgray!0] fill between[of=C and CD,
			                				soft clip={domain=0:1},
			                			];
			                			
			                			\addplot[white, bottom color=lightgray!0, top color=lightgray!50] fill between[of=CD and D,
			                				soft clip={domain=0:1},
			                			];
			                			
			                			\addplot[white, bottom color=lightgray!50, top color=lightgray!0] fill between[of=E and EF,
			                				soft clip={domain=0:1},
			                			];
			                			
			                			\addplot[white, bottom color=lightgray!0, top color=lightgray!50] fill between[of=EF and F,
			                				soft clip={domain=0:1},
			                			];
			                			
			                			\addplot[white, bottom color=lightgray!50, top color=lightgray!0] fill between[of=G and GH,
			                				soft clip={domain=0:1},
			                			];
			                			
			                			\addplot[white, bottom color=lightgray!0, top color=lightgray!50] fill between[of=GH and H,
			                				soft clip={domain=0:1},
			                			];		
			                		\end{axis}
			                	\end{tikzpicture}
					\caption{Illustration of Proposition~\ref{prop: proposition weak-* convergence support} with $L = 3$. The union of intervals corresponding to the shaded areas represents $T_\varepsilon$ on the fundamental interval $[0, 1)$. The black dotted spikes represent $\meas$. Proposition \ref{prop: proposition weak-* convergence support} guarantees that for $N \geq M_\varepsilon$, each interval $[t_\measSpikeIdx - \varepsilon, t_\measSpikeIdx + \varepsilon]$ contains at least one atom of $\measProb_N$.  The gray spikes correspond to $\measProb_N$ for $N \geq M_\varepsilon$. The weights of the atoms of $\measProb_N$ outside $T_\varepsilon$ is guaranteed to be smaller than $\varepsilon\normTV{\meas}$.}
			}{
		                	\begin{tikzpicture}
			                		\tikzstyle{every node}=[font=\scriptsize]
			                		\begin{axis}[
			                			axis x line = middle, axis y line = middle,
			                			xmin = -0.05,xmax = 1.05,
			                			ymin = -2.5, ymax = 2.5, 
			                			grid = none, 
			                			height = 9cm, width = 0.9\textwidth, 
			                			scale only axis, 
			                			enlarge x limits = false, enlarge y limits = false,
			                			ytick = {-0.25, 0, 0.25},
			                			yticklabels = {$-\varepsilon$, $0$, $\varepsilon$}, 
			                			xtick = {0, 0.16, 0.2, 0.24, 0.56, 0.6, 0.64, 1},
			                			xticklabels = {$0$, {\tiny$t_1\!-\!\varepsilon$}, {\textcolor{black}{$\boldmath t_1$}}, {\tiny $t_1\!+\!\varepsilon$}, {\tiny$t_2\!-\!\varepsilon$}, {\textcolor{black}{$t_2$}}, {\tiny $t_2\!+\!\varepsilon$}, $1$},
			                			xticklabel shift = {-0.1cm},
			                			extra x ticks = {0.76, 0.8, 0.84},
			                			extra x tick labels = {{\tiny $t_3\!-\!\varepsilon$}, {\textcolor{black}{$t_3$}}, {\tiny $t_3\!+\!\varepsilon$}},
			                			extra x tick style = {
			                				xticklabel style={yshift = -0.05cm, anchor = south}
			                			},
			                			legend columns = 2, legend cell align  = left, legend pos = north west, 
			                			]
			                			\addplot+[very thick, densely dotted, dirac, black] coordinates {(0.2, 1.5) (0.6, 2) (0.8, -1.2)};
			                			\addplot+[thick, dirac, darkgray] coordinates {(0.195, 1.1) (0.58, 1) (0.81, -1.3)};
			                			\addplot+[thick, dirac, darkgray] coordinates {(0.21, 0.7) (0.605, 1.6) (0.61, 1.2)};
			                			\addplot+[thick, dirac, darkgray] coordinates {(0.1, 0.13) (0.35, 0.1) (0.4, -0.13) (0.5, 0.09) (0.7, -0.12) (0.71, 0.19) (0.9, -0.1) (0.95, 0.1)};

			                			\addplot[name path=A1, very thin, dashed, black, no markers, domain = 0:0.16] {0.25};
			                			\addplot[name path=B1, very thin, dashed, black, no markers, domain = 0:0.16] {-0.25};
			                			
			                			\addplot[name path=A2, very thin, dashed, black, no markers, domain = 0.24:0.56] {0.25};
			                			\addplot[name path=B2, very thin, dashed, black, no markers, domain = 0.24:0.56] {-0.25};
			                			
			                			\addplot[name path=A3, very thin, dashed, black, no markers, domain = 0.64:0.76] {0.25};
			                			\addplot[name path=B3, very thin, dashed, black, no markers, domain = 0.64:0.76] {-0.25};
			                			
			                			\addplot[name path=A4, very thin, dashed, black, no markers, domain = 0.84:1.3] {0.25};
			                			\addplot[name path=B4, very thin, dashed, black, no markers, domain = 0.84:1.3] {-0.25};
			                			
			                			\addplot[name path=C, very thin, dashed, white, no markers, domain = 0.16:0.24] {2.2};
			                			\addplot[name path=CD, very thin, dashed, white, no markers, domain = 0.16:0.24] {0};
			                			\addplot[name path=D, very thin, dashed, white, no markers, domain = 0.16:0.24] {-2.2};
			                			
			                			\addplot[name path=E, very thin, dashed, white, no markers, domain = 0.56:0.64] {2.2};
			                			\addplot[name path=EF, very thin, dashed, white, no markers, domain = 0.56:0.64] {0};
			                			\addplot[name path=F, very thin, dashed, white, no markers, domain = 0.56:0.64] {-2.2};
			                			
			                			\addplot[name path=G, very thin, dashed, white, no markers, domain = 0.76:0.84] {2.2};
			                			\addplot[name path=GH, very thin, dashed, white, no markers, domain = 0.76:0.84] {0};
			                			\addplot[name path=H, very thin, dashed, white, no markers, domain = 0.76:0.84] {-2.2};
			                			
			                			\addplot[very thick, black, no markers, domain=0.16:0.24] {0};
			                			\addplot[very thick, black, no markers, domain=0.56:0.64] {0};
			                			\addplot[very thick, black, no markers, domain=0.76:0.84] {0};
			                			
			                			\addplot[pattern=dots, pattern color=lightgray] fill between[of=A1 and B1,
			                				soft clip={domain=0:1},
			                			];
			                			
			                			\addplot[pattern=dots, pattern color=lightgray] fill between[of=A2 and B2,
			                				soft clip={domain=0:1},
			                			];
			                			
			                			\addplot[pattern=dots, pattern color=lightgraypattern=dots, pattern color=lightgray] fill between[of=A3 and B3,
			                				soft clip={domain=0:1},
			                			];
			                			
			                			\addplot[pattern=dots, pattern color=lightgray] fill between[of=A4 and B4,
			                				soft clip={domain=0:1.3},
			                			];	
						         \addplot[path fading=north, pattern=north west lines, pattern color=lightgray] fill between[of=C and CD,
			                				soft clip={domain=0:1},
			                			];
			                			
			                			\addplot[path fading=south, pattern=north west lines, pattern color=lightgray] fill between[of=CD and D,
			                				soft clip={domain=0:1},
			                			];
			                			
			                			\addplot[path fading=north, pattern=north west lines, pattern color=lightgray] fill between[of=E and EF,
			                				soft clip={domain=0:1},
			                			];
			                			
			                			\addplot[path fading=south, pattern=north west lines, pattern color=lightgray] fill between[of=EF and F,
			                				soft clip={domain=0:1},
			                			];
			                			
			                			\addplot[path fading=north, pattern=north west lines, pattern color=lightgray] fill between[of=G and GH,
			                				soft clip={domain=0:1},
			                			];
			                			
			                			\addplot[path fading=south, pattern=north west lines, pattern color=lightgray] fill between[of=GH and H,
			                				soft clip={domain=0:1},
			                			];
			                		\end{axis}
			                	\end{tikzpicture}
					\caption{Illustration of Proposition~\ref{prop: proposition weak-* convergence support} with $L = 3$. The union of intervals corresponding to the hashed areas represents $T_\varepsilon$ on the fundamental interval $[0, 1)$. The black dotted spikes represent $\meas$. Proposition \ref{prop: proposition weak-* convergence support} guarantees that for $N \geq M_\varepsilon$, each interval $[t_\measSpikeIdx - \varepsilon, t_\measSpikeIdx + \varepsilon]$ contains at least one atom of $\measProb_N$.  The gray spikes correspond to $\measProb_N$ for $N \geq M_\varepsilon$. The weights of the atoms of $\measProb_N$ outside $T_\varepsilon$ is guaranteed to be smaller than $\varepsilon\normTV{\meas}$.}
			}
		}
	}
	\label{fig: proof last proposition}
\end{figure}

\begin{proof}
	We start by establishing \eqref{eq: support not empty near tell}.
	To this end, let $\varepsilon \in (0, \Delta/4]$, $\measSpikeIdx \in \{1, 2, \ldots, L\}$, and define $\psi_{\varepsilon, \measSpikeIdx} \in \C(\torus)$ as 
	\begin{equation}
		\forall t \in \torus, \qquad \psi_{\varepsilon, \measSpikeIdx}(t) = \frac{\overline{\measAmplitude{\measSpikeIdx}}}{\abs{\measAmplitude{\measSpikeIdx}}}\sum_{n \in \Z} \psi\!\left(\frac{t - t_\measSpikeIdx -n}{\varepsilon}\right)\!,
		\label{eq: def bump function}
	\end{equation}
	where $\psi \colon \R \rightarrow \R$ is given by
	\begin{equation*}
		\forall t \in \R, \qquad \psi(t) \triangleq \begin{cases} 1 - \abs{t}, & \abs{t} \leq 1 \\ 0, & \text{otherwise.}\end{cases}
	\end{equation*}
	As the minimum wrap-around distance between the $t_\measSpikeIdx$ is $\Delta$, by assumption, and $\varepsilon \leq \Delta/4$, we have
	\begin{equation*}
		\innerProd{\psi_{\varepsilon, \measSpikeIdx}}{\meas} = \sum_{\measSpikeIdxBis = 1}^L \measAmplitude{\measSpikeIdxBis} \psi_{\varepsilon, \measSpikeIdx}(t_\measSpikeIdxBis) = \measAmplitude{\measSpikeIdx}\psi_{\varepsilon, \measSpikeIdx}(t_\measSpikeIdx) = \abs{\measAmplitude{\measSpikeIdx}}.
	\end{equation*}
	As $\{\measProb_N\}_{N \in \N}$ is, by assumption, a sequence of solutions of $\gaborProbN$ and $\meas$ is the unique solution of $\gaborProb$, it follows from Proposition~\ref{prop: convergence solution} that $\{\measProb_N\}_{N \in \N}$ converges in the weak-* sense to $\meas$, i.e., for every $\varphi \in \C(\torus)$, it holds that 
	\begin{equation}
		\forall \eta > 0, \ \exists M_{\varphi,\eta} > 0, \ \forall N \geq M_{\varphi,\eta}, \qquad  \abs{\innerProd{\varphi}{\measProb_N} - \innerProd{\varphi}{\meas}} \leq \eta.
		\label{eq: limit weak convergence}
	\end{equation} 
	Since $\psi$ is continuous, by construction, $\psi_{\varepsilon, \ell} \in C(\torus)$, and setting $\eta = \abs{\measAmplitude{\measSpikeIdx}}/2 > 0$ and $\varphi = \psi_{\varepsilon, \measSpikeIdx}$ in \eqref{eq: limit weak convergence} implies that there exists an $M_{\varepsilon, \measSpikeIdx} > 0$ such that for all $N \geq M_{\varepsilon, \measSpikeIdx}$, we have 
	\begin{equation}
		\abs{\innerProd{\psi_{\varepsilon, \measSpikeIdx}}{\measProb_N} - \innerProd{\psi_{\varepsilon, \measSpikeIdx}}{\meas}} = \abs{\innerProd{\psi_{\varepsilon, \measSpikeIdx}}{\measProb_N} - \abs{\measAmplitude{\measSpikeIdx}}} \leq \abs{\measAmplitude{\measSpikeIdx}}/2.
		\label{eq: way of contradiction}
	\end{equation}
	Now, set 
	\begin{equation*}
		M_\varepsilon \triangleq \max_{1 \leq \measSpikeIdx \leq L} M_{\varepsilon, \measSpikeIdx},
	\end{equation*}
	let $N \geq M_\varepsilon$, and assume, by way of contradiction, that there exists an $\measSpikeIdx' \in \{1, 2, \ldots, L\}$ such that $\supp(\measProb_N) \cap [t_{\measSpikeIdx'} - \varepsilon, t_{\measSpikeIdx'} + \varepsilon] = \emptyset$. Then, as $\psi_{\varepsilon, \measSpikeIdx'} = 0$ on $[0,1) \!\setminus\! [t_{\measSpikeIdx'} - \varepsilon, t_{\measSpikeIdx'} + \varepsilon]$, we have $\innerProd{\psi_{\varepsilon, \measSpikeIdx'}}{\measProb_N} = 0$, which by \eqref{eq: way of contradiction} would imply $\abs{\measAmplitude{\measSpikeIdx'}} \leq \abs{\measAmplitude{\measSpikeIdx'}}/2$ and thereby contradict our assumption $\measAmplitude{\measSpikeIdx'} \neq 0$. 
	We can therefore conclude that for $N \geq M_\varepsilon$, for all $\measSpikeIdx \in \{1, 2, \ldots, L\}$, we have $\supp(\measProb_N) \cap [t_\measSpikeIdx - \varepsilon, t_\measSpikeIdx + \varepsilon] \neq \emptyset$. This completes the proof of \eqref{eq: support not empty near tell}.
			
	We proceed to establishing \eqref{eq: concentration inequality measure}. For $\varepsilon \in (0, \Delta/4]$, define $\psi_\varepsilon \in C(\torus)$ as
	\begin{equation*}
		\forall t \in \torus, \qquad \psi_\varepsilon(t) \triangleq \sum_{\measSpikeIdx = 1}^L \psi_{\varepsilon, \measSpikeIdx}(t),
	\end{equation*}
	where $\psi_{\varepsilon,\measSpikeIdx} \colon \R \rightarrow \R$, $\measSpikeIdx \in \{1, 2, \ldots, L\}$, is as in \eqref{eq: def bump function}. As for all $\measSpikeIdx \in \{1, 2, \ldots, L\}$, $\psi_{\varepsilon, \measSpikeIdx}$ is supported on $\bigcup_{n \in \Z} [t_\measSpikeIdx - \varepsilon + n, t_\measSpikeIdx + \varepsilon + n]$, $\psi_\varepsilon$ is supported on $T_\varepsilon$. Moreover, since $\varepsilon \leq \Delta/4$, by asssumption, the $\psi_{\varepsilon, \measSpikeIdx}$ have disjoint supports, and hence
	\begin{equation}
		\forall t \in \torus, \qquad \abs{\psi_\varepsilon(t)} = \sum_{\measSpikeIdx = 1}^L \abs{\psi_{\varepsilon, \measSpikeIdx}(t)} \leq 1.
		\label{eq: inequality varphi epsilon smaller than one}
	\end{equation}
	We also have
	\begin{equation*}
		\innerProd{\psi_\varepsilon}{\meas} = \sum_{\measSpikeIdx = 1}^L \innerProd{\psi_{\varepsilon, \measSpikeIdx}}{\meas} = \sum_{\measSpikeIdx = 1}^L \abs{\measAmplitude{\measSpikeIdx}} = \normTV{\meas}.
	\end{equation*}
	Applying \eqref{eq: limit weak convergence} with $\varphi = \psi_\varepsilon$ and $\eta = \varepsilon\!\normTV{\meas} > 0$, we can conclude that there exists an $M_\varepsilon > 0$ such that for all $N \geq M_\varepsilon$, we have
	\begin{align}
		\varepsilon\normTV{\meas} &\geq \abs{\innerProd{\psi_\varepsilon}{\measProb_N} - \innerProd{\psi_\varepsilon}{\meas}} \notag \\
			&=  \abs{\innerProd{\psi_\varepsilon}{\measProb_{N, T_\varepsilon}} - \normTV{\meas}} \label{eq: inequality proof convergence support} \\
			&\geq \normTV{\meas} - \innerProd{\psi_\varepsilon}{\measProb_{N, T_\varepsilon}}, \notag
	\end{align}
	where \eqref{eq: inequality proof convergence support} is a consequence of $\psi_\varepsilon$ being supported on $T_\varepsilon$.
	It therefore follows that for all $N \geq M_\varepsilon$, we have
	\begin{equation}
		(1 - \varepsilon)\normTV{\meas} \leq \innerProd{\psi_\varepsilon}{\measProb_{N, T_\varepsilon}} \leq \int_0^1 \abs{\psi_\varepsilon(t)}\dmeas{\!\abs{\measProb_{N,T_\varepsilon}}\!}{t} \leq \int_0^1 \dmeas{\!\abs{\measProb_{N,T_\varepsilon}}\!}{t} = \normTV{\measProb_{N, T_\varepsilon}},
		\label{eq: previous equation limit}
	\end{equation}
	where the last inequality follows from \eqref{eq: inequality varphi epsilon smaller than one}.
	Furthermore, we have for all $N \geq M_\varepsilon$,
	\begin{align}
		\normTV{\measProb_{N, T_\varepsilon^c}} &= \normTV{\measProb_N} - \normTV{\measProb_{N, T_\varepsilon}} \label{eq: disjoint set measure limit} \\
			&\leq \normTV{\measProb_N} - (1-\varepsilon)\normTV{\meas} \label{eq: consequence of previous equation} \\
			&\leq \normTV{\meas} - (1 - \varepsilon)\normTV{\meas} = \varepsilon\normTV{\meas}, \label{eq: last equation limit}
	\end{align}
	where \eqref{eq: disjoint set measure limit} is a consequence of the supports of $\measProb_{N, T_\varepsilon}$ and $\measProb_{N, T_\varepsilon^c}$ being disjoint, \eqref{eq: consequence of previous equation} follows from \eqref{eq: previous equation limit}, and \eqref{eq: last equation limit} is by $\normTV{\measProb_N} \leq \normTV{\meas}$, for all $N \in \N$, as established in \eqref{eq: proof convergence problem simulations 2}.
\end{proof}

%% file: simulations.tex

\section{Numerical results}
\label{section: simulations}

For our numerical results we consider the case $\group = \torus$, i.e., recovery of the discrete complex measure $\meas = \sum_{\measSpikeIdx =1}^L \measAmplitude{\measSpikeIdx} \delta_{t_\measSpikeIdx} \in \measSpace{\torus}$ from the measurements $\measurements(\tau, k) = (\mathcal{A}_{\window_{N}}\measProb_{N})(\tau, k)$. We solve the predual problem $\dualGaborProbN$ with $N = 25$ and $N = 50$ by applying the convex solver \texttt{cvx} to the formulation of $\dualGaborProbN$ as given in \eqref{eq: equivalent problem predual problem}.

\definecolor{fourierColor}{RGB}{153, 8, 247}
\definecolor{stft10Color}{RGB}{8, 208, 247}
\definecolor{stft5Color}{RGB}{9, 140, 247}
\definecolor{stft05Color}{RGB}{0, 0, 250}

\definecolor{fourierGray}{gray}{0}
\definecolor{stft10Gray}{gray}{0.3}
\definecolor{stft5Gray}{gray}{0.4}
\definecolor{stft05Gray}{gray}{0.5}

\begin{figure}
	\ifthenelse{\equal{\versionColor}{true}}{
		\ifthenelse{\equal{\pdfFigs}{true}}{
			\includegraphics[width = 0.9\textwidth]{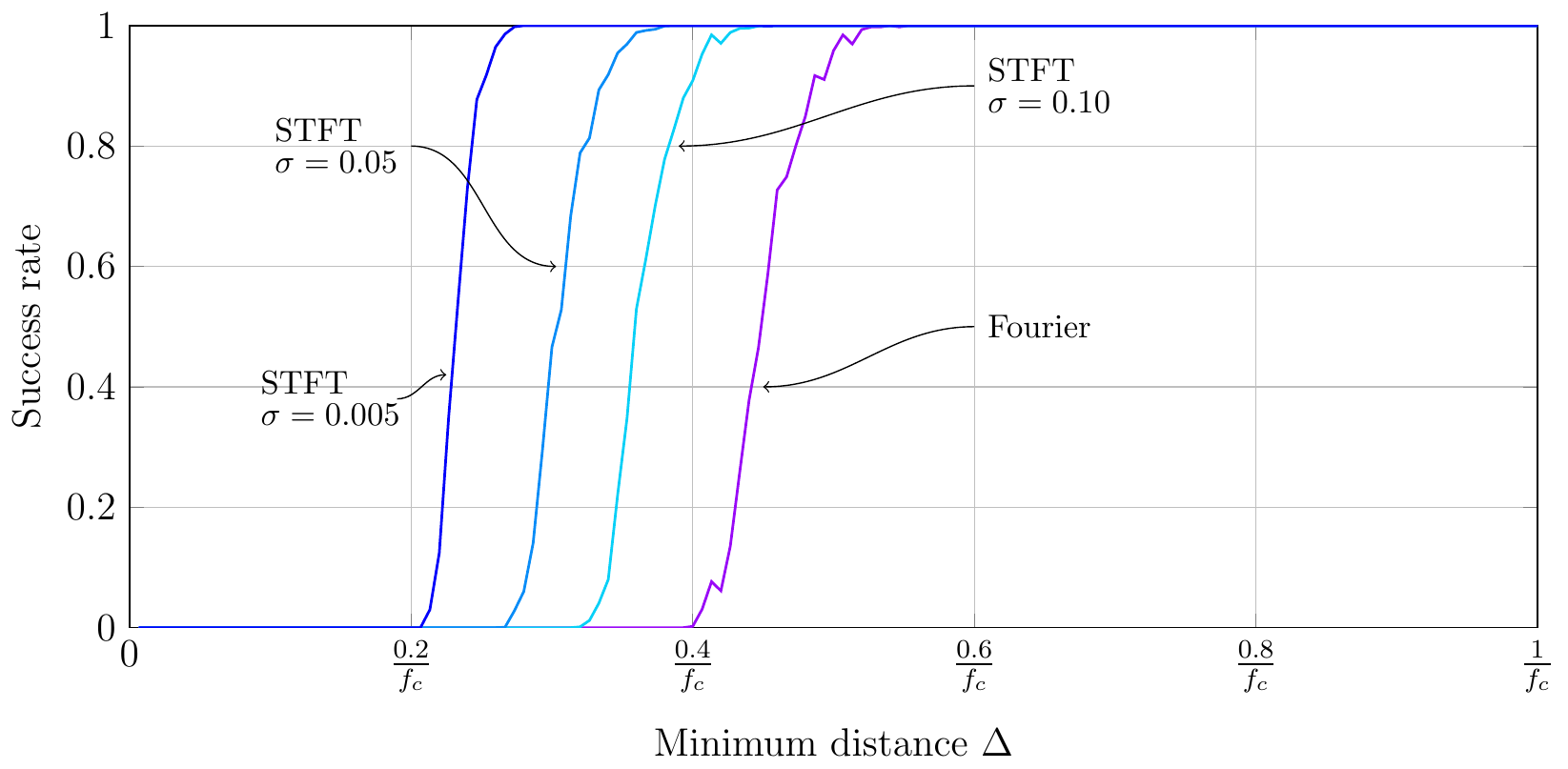}
		}{
            		\ifthenelse{\equal{\generateFigs}{yes}}{\tikzsetnextfilename{numericalResultsColor}}{}
                        	\begin{tikzpicture}
                        		\begin{axis}
                        			[no markers,
                        			width = \columnwidth, height = 8cm,
                        			xlabel = {\begin{minipage}{4cm}\vspace{0.1cm} \centering Minimum distance $\Delta$\end{minipage}},
                        			ylabel = Success rate,
                        			grid = both,
                        			enlarge x limits = false, enlarge y limits = false,
                        			xmin = 0, xmax = 1,
                        			xtick = {0, 0.2, 0.4, 0.6, 0.8, 1}, 
                        			xticklabels = {$0$, $\frac{0.2}{f_c}$,$\frac{0.4}{f_c}$,$\frac{0.6}{f_c}$,$\frac{0.8}{f_c}$,$\frac{1}{f_c}$},
                        			ymin = 0, ymax = 1]
                        			
                        			\addplot[thick, fourierColor] file {data/numericalResults/Fourier_cvx_2000trials_support.dat};
                        			\addplot[thick, stft10Color] file {data/numericalResults/STFT_cvx_window10_2000trials_support.dat};
                        			\addplot[thick, stft5Color] file {data/numericalResults/STFT_cvx_window5_2000trials_support.dat};
                        			\addplot[thick, stft05Color] file {data/numericalResults/STFT_cvx_window05_2000trials_support.dat};
                        			
                        			\node[anchor=west] (Fourier) at (60,50) {{\footnotesize Fourier}};
                        			\draw (Fourier) edge[out=180,in=0,->] (45,40);
                        			
                        			\node[anchor=west, align = left] (STFT10) at (60,90) {{\footnotesize STFT}\\[-5pt] {\footnotesize $\sigma = 0.10$}};
                        			\draw (STFT10) edge[out=180,in=0,->] (39,80);
                        			
                        			\node[anchor=east, align = left] (STFT5) at (20,80) {{\footnotesize STFT}\\[-5pt] {\footnotesize $\sigma = 0.05$}};
                        			\draw (STFT5) edge[out=0,in=180,->] (30.3,60);
                        			
                        			\node[anchor=east, align = left] (STFT1) at (19, 38) {{\footnotesize STFT}\\[-5pt] {\footnotesize $\sigma = 0.005$\!\!\!}};
                        			\draw (STFT1) edge[out=0,in=180,->] (22.5,42);
                        		\end{axis}
                        	\end{tikzpicture}
		}
	}{
		\ifthenelse{\equal{\pdfFigs}{true}}{
			\includegraphics[width = 0.9\textwidth]{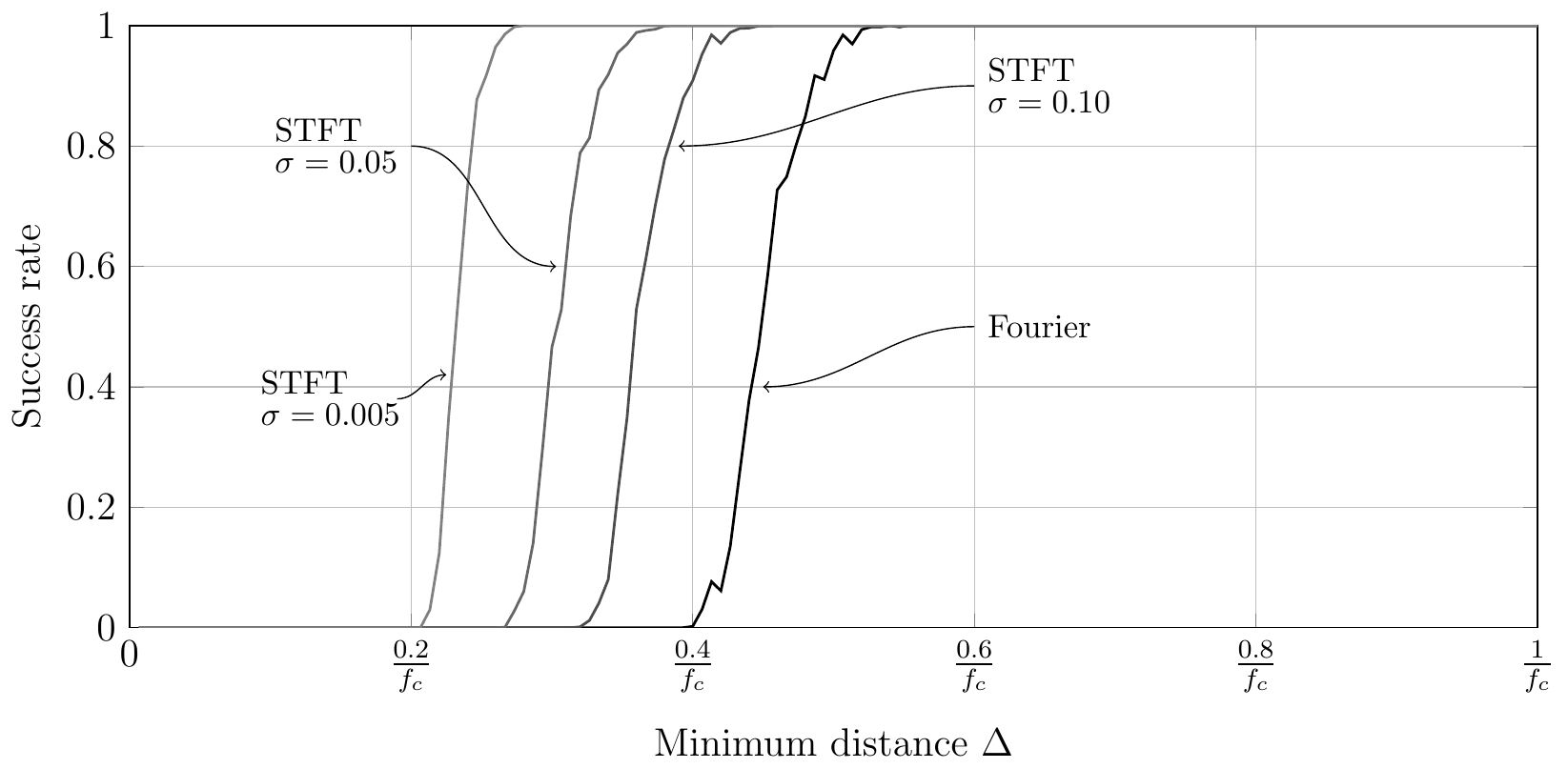}
		}{
        		\ifthenelse{\equal{\generateFigs}{yes}}{\tikzsetnextfilename{numericalResultsBW}}{}
                    	\begin{tikzpicture}
                    		\begin{axis}
                    			[no markers,
                    			width = \columnwidth, height = 8cm,
                    			xlabel = {\begin{minipage}{4cm}\vspace{0.1cm} \centering Minimum distance $\Delta$\end{minipage}},
                    			ylabel = Success rate,
                    			grid = both,
                    			enlarge x limits = false, enlarge y limits = false,
                    			xmin = 0, xmax = 1,
                    			xtick = {0, 0.2, 0.4, 0.6, 0.8, 1}, 
                    			xticklabels = {$0$, $\frac{0.2}{f_c}$,$\frac{0.4}{f_c}$,$\frac{0.6}{f_c}$,$\frac{0.8}{f_c}$,$\frac{1}{f_c}$},
                    			ymin = 0, ymax = 1]
                    			
                    			\addplot[thick, fourierGray] file {data/numericalResults/Fourier_cvx_2000trials_support.dat};
                    			\addplot[thick, stft10Gray] file {data/numericalResults/STFT_cvx_window10_2000trials_support.dat};
                    			\addplot[thick, stft5Gray] file {data/numericalResults/STFT_cvx_window5_2000trials_support.dat};
                    			\addplot[thick, stft05Gray] file {data/numericalResults/STFT_cvx_window05_2000trials_support.dat};
                    			
                    			\node[anchor=west] (Fourier) at (60,50) {{\footnotesize Fourier}};
                    			\draw (Fourier) edge[out=180,in=0,->] (45,40);
                    			
                    			\node[anchor=west, align = left] (STFT10) at (60,90) {{\footnotesize STFT}\\[-5pt] {\footnotesize $\sigma = 0.10$}};
                    			\draw (STFT10) edge[out=180,in=0,->] (39,80);
                    			
                    			\node[anchor=east, align = left] (STFT5) at (20,80) {{\footnotesize STFT}\\[-5pt] {\footnotesize $\sigma = 0.05$}};
                    			\draw (STFT5) edge[out=0,in=180,->] (30.3,60);
                    			
                    			\node[anchor=east, align = left] (STFT1) at (19, 38) {{\footnotesize STFT}\\[-5pt] {\footnotesize $\sigma = 0.005$\!\!\!}};
                    			\draw (STFT1) edge[out=0,in=180,->] (22.5,42);
                    		\end{axis}
                    	\end{tikzpicture}
		}
	}	
	\caption{Success rate for support recovery from pure Fourier measurements and from STFT measurements with $\freq_c = 50$ and $N = 50$.}
	\label{fig: monte carlo simulations for the support}
\end{figure}

To assess recovery performance, we run $1500$ trials as follows. For each $\Delta$, we construct a discrete complex measure $\meas$ supported on the set $\measSupport = \{t_\measSpikeIdx\}_{\measSpikeIdx = 0}^{S}$ with $S = \lfloor 1/(2\Delta)\rfloor$ and $t_\measSpikeIdx = 2\measSpikeIdx\Delta + r_\measSpikeIdx$, where $r_\measSpikeIdx$ is chosen uniformly at random in $[0, \Delta]$. The minimum wrap-around distance between the points in $\measSupport$ is therefore guaranteed to be greater than or equal to $\Delta$. The complex weights $a_\ell$ are obtained by choosing their real and imaginary parts independently and uniformly at random in $[0, 1000]$. We declare success if the reconstructed measure $\hat{\meas}$ has support $\hat{T} = \{\hat{t}_\measSpikeIdx\}_{\measSpikeIdx \in \measSupportIdx}$ satisfying $\|\hat{T} - T\|_{\ell^2}/\|T\|_{\ell^2} \leq 10^{-3}$. The corresponding results are depicted in Fig.~\ref{fig: monte carlo simulations for the support}. 
We observe that both recovery from STFT measurements and from pure Fourier measurements actually work beyond their respective thresholds $\Delta > 1/\freq_c$ and $\Delta > 2/\freq_c$, thus suggesting that neither of the thresholds is sharp. We also observe a factor-of-two improvement in the case of recovery from STFT measurements relative to recovery from pure Fourier measurements, suggesting that the improvement in the recovery threshold $\Delta > 1/\freq_c$ for STFT measurements relative to $\Delta > 2/\freq_c$ for pure Fourier measurements is due to the recovery problem itself. 
Specifically, in the STFT case, we perform windowing and the STFT measurements are highly redundant. To see this note that in the pure Fourier case, the measurements consist of the vector $\mathbf{y} = \{y_k\}_{k = -K}^K \in \C^{2K+1}$ as defined in \eqref{eq: measurements in the case of Fourier measurements}, whereas the STFT measurements here are characterized by the $(2N+1)\times (2K+1)$ entries of the matrix $\mathbf{Y} \in \C^{(2K+1) \times (2N+1)}$ containing the $2N+1$ Fourier series coefficients of the $2K+1$ functions $\tau \mapsto \measurements(\tau, k)$, $k \in \{-K, \ldots, K\}$. The increased number of measurements in the STFT case leads to larger optimization problem sizes and hence entails increased computational complexity relative to the pure Fourier case.

In Fig.~\ref{fig: dual polynomial}, we compare the dual polynomial for pure Fourier and for STFT measurements in a situation where recovery in the former case fails and where it succeeds in the latter.\footnote{In the case of pure Fourier measurements, we follow the recovery procedure as described in \cite{Candes2012} and solve $\dualFourierProb$ defined in Section~\ref{section: previous work} using the convex solver \texttt{cvx}.} We can see that in both cases, the magnitude of the dual polynomial equals $1$ for all $t_\measSpikeIdx$, $\measSpikeIdx \in \measSupportIdx$. However, in the case of pure Fourier measurements, the magnitude of the dual polynomial also takes on the value $1$ at eight additional locations not belonging to the support set $\measSupport$. For example, we can see in the top plot in Fig.~\ref{fig: dual polynomial} that spikes are detected for $t = 0.1791$ and $t = 0.2172$, while these two points do not belong to the support set of the original measure. In contrast, in the case of STFT measurements, the locations where the magnitude of the polynomial takes on the value $1$ approximate the points of the support set of the original measure (represented by circles in Fig.~\ref{fig: dual polynomial}) well, namely with a relative error of $\|\hat{T} - T\|_{\ell^2}/\|T\|_{\ell^2} = 2\cdot 10^{-7}$.

\begin{figure}
	\ifthenelse{\equal{\versionColor}{true}}{
            	\ifthenelse{\equal{\pdfFigs}{true}}{
            			\includegraphics[width = 0.95\textwidth]{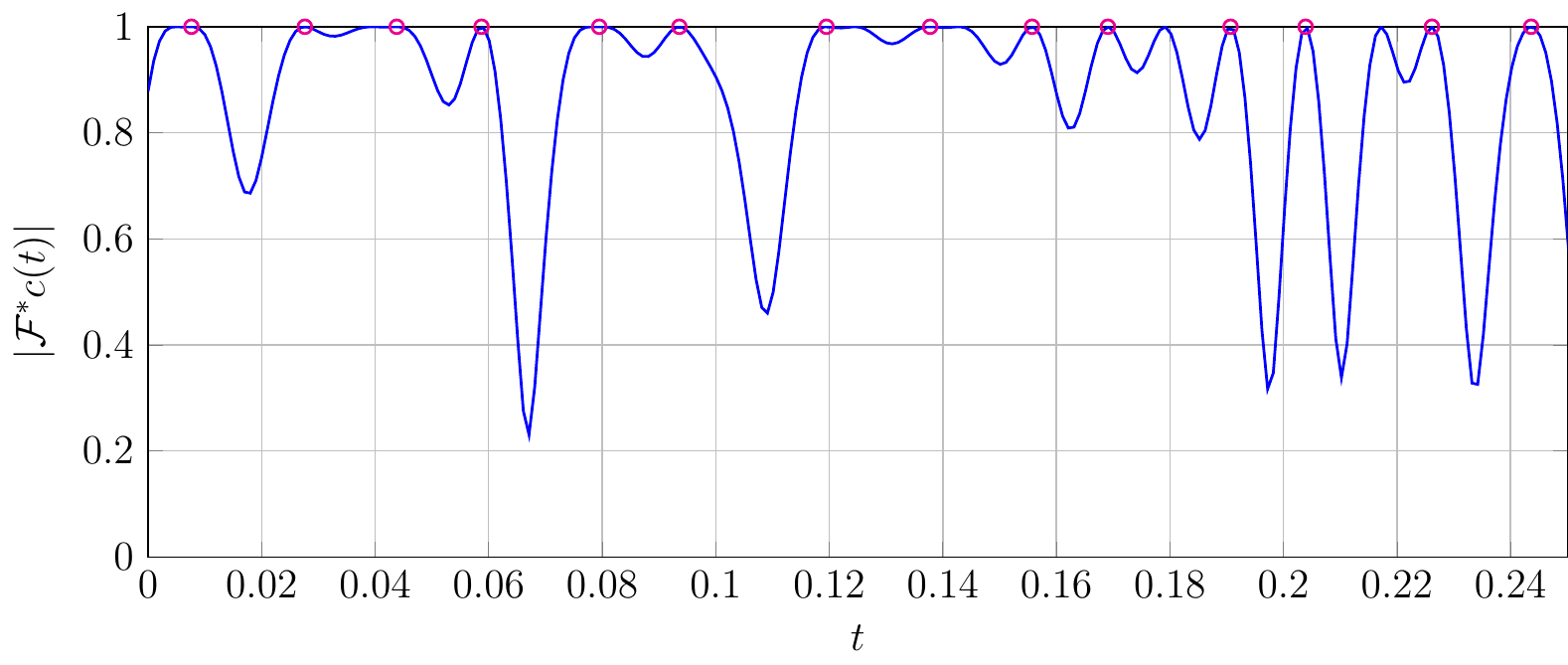}
            	}{
                        		\ifthenelse{\equal{\generateFigs}{yes}}{\tikzsetnextfilename{interpolationDualPolynomial1Color}}{}
                                    	\begin{tikzpicture}
                                    		\begin{axis}
                                    			[width = 0.97\columnwidth, height = 7cm,
                                    			xlabel = $t$,
                                    			ylabel = $\abs{\mathcal{F}^*c(t)}$,
                                    			grid = both,
                                    			enlarge x limits = false, enlarge y limits = false,
                                    			xmin = 0, xmax = 0.25,
                                    			xticklabel style = /pgf/number format/fixed,
                                    			ymin = 0, ymax = 1]
                                    			\addplot[thick, magenta, only marks, mark = o] file {data/exampleInterpolation/measureSupport.dat};
                                    			\addplot[thick, blue, no markers] file {data/exampleInterpolation/dualPolynomialFourier.dat};
                                    		\end{axis}
                                    	\end{tikzpicture}
            	}
            	\ifthenelse{\equal{\pdfFigs}{true}}{
            		\includegraphics[width = 0.95\textwidth]{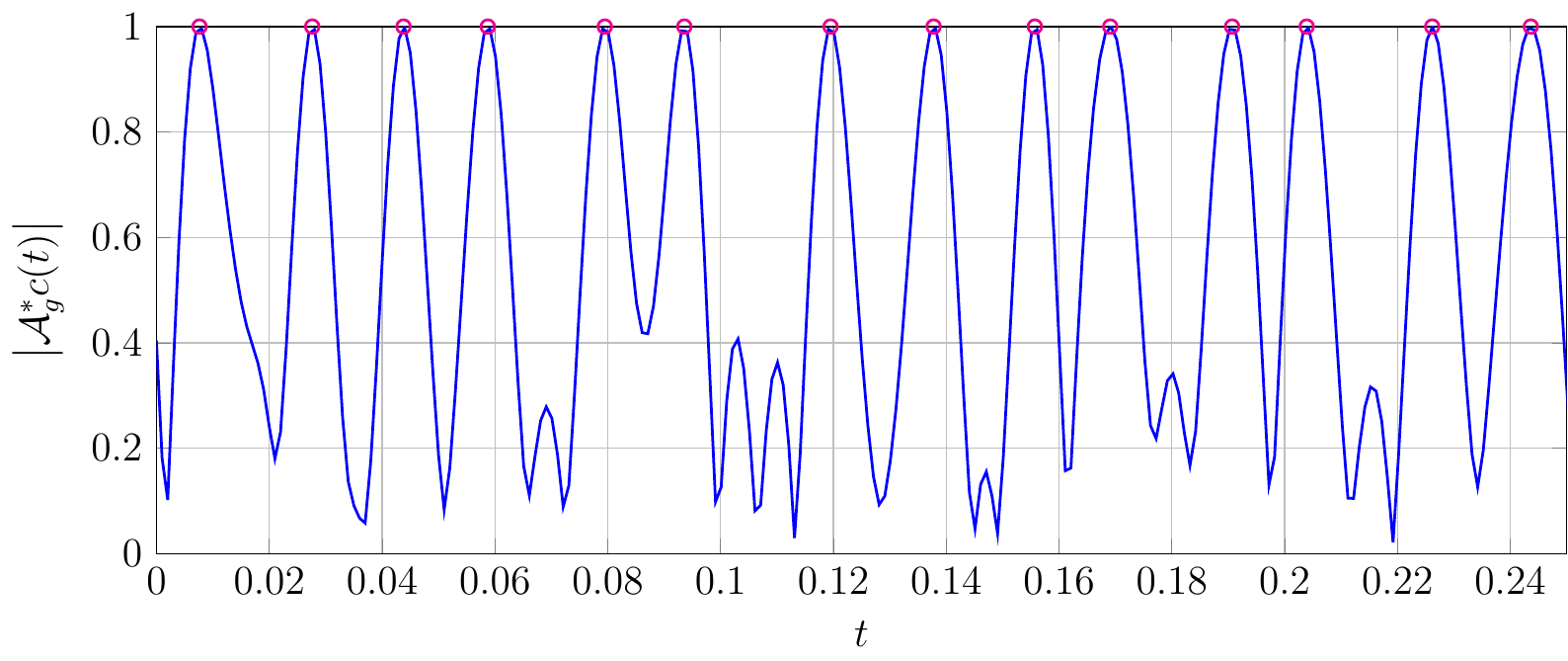}
            	}{
            		\ifthenelse{\equal{\generateFigs}{yes}}{\tikzsetnextfilename{interpolationDualPolynomial2Color}}{}
                        	\begin{tikzpicture}
                        		\begin{axis}
                        			[width = 0.97\columnwidth, height = 7cm,
                        			xlabel = $t$,
                        			ylabel = $\abs{\gaborOp^*c(t)}$,
                        			grid = both,
                        			enlarge x limits = false, enlarge y limits = false,
                        			xmin = 0, xmax = 0.25,
                        			xticklabel style = /pgf/number format/fixed,
                        			ymin = 0, ymax = 1]
                        			\addplot[thick, magenta, only marks, mark = o] file {data/exampleInterpolation/measureSupport.dat};
                        			\addplot[thick, blue, no markers] file {data/exampleInterpolation/dualPolynomialSTFT.dat};
                        		\end{axis}
                        	\end{tikzpicture}
            	}
	}{
            	\ifthenelse{\equal{\pdfFigs}{true}}{
            		\includegraphics[width = 0.95\textwidth]{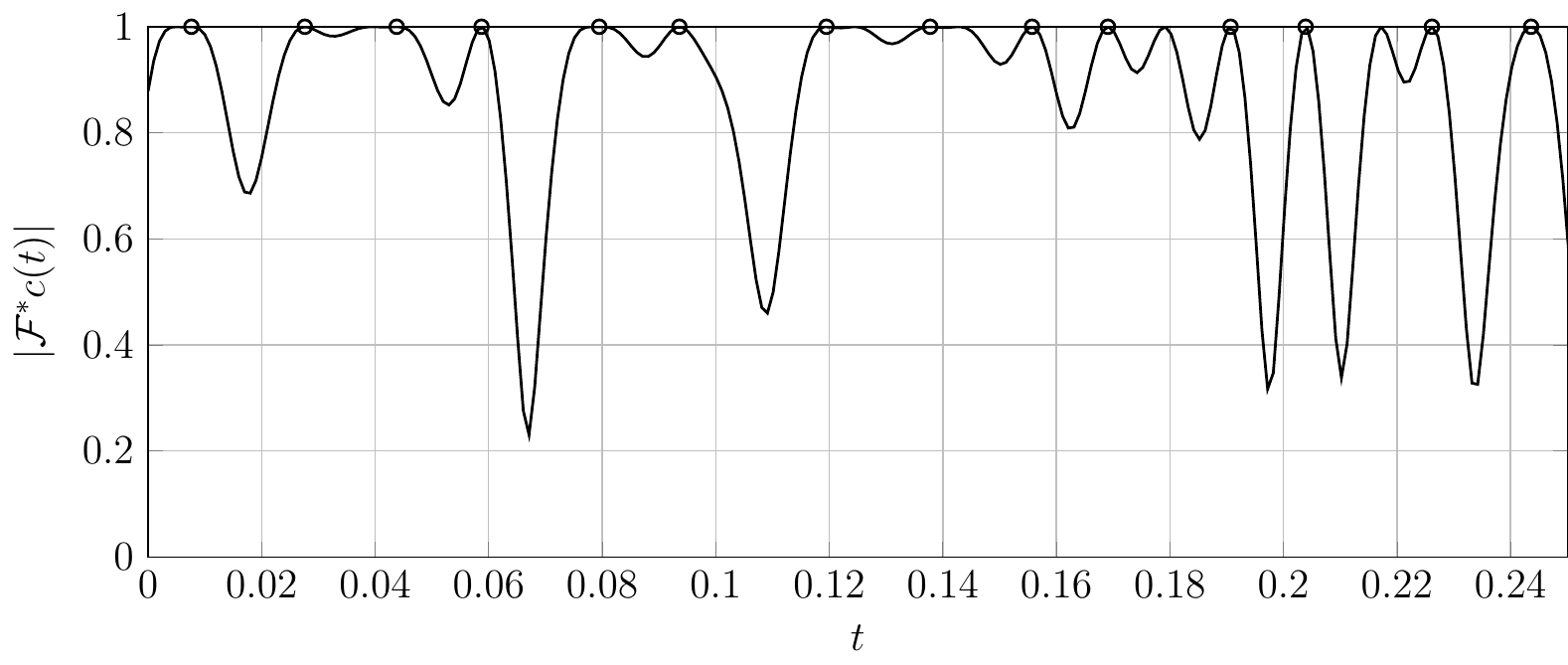}
            	}{
            		\ifthenelse{\equal{\generateFigs}{yes}}{\tikzsetnextfilename{interpolationDualPolynomial1BW}}{}
                        	\begin{tikzpicture}
                        		\begin{axis}
                        			[width = 0.97\columnwidth, height = 7cm,
                        			xlabel = $t$,
                        			ylabel = $\abs{\mathcal{F}^*c(t)}$,
                        			grid = both,
                        			enlarge x limits = false, enlarge y limits = false,
                        			xmin = 0, xmax = 0.25,
                        			xticklabel style = /pgf/number format/fixed,
                        			ymin = 0, ymax = 1]
                        			\addplot[thick, black, only marks, mark = o] file {data/exampleInterpolation/measureSupport.dat};
                        			\addplot[thick, black, no markers] file {data/exampleInterpolation/dualPolynomialFourier.dat};
                        		\end{axis}
                        	\end{tikzpicture}
            	}
            	\ifthenelse{\equal{\pdfFigs}{true}}{
            		\includegraphics[width = 0.95\textwidth]{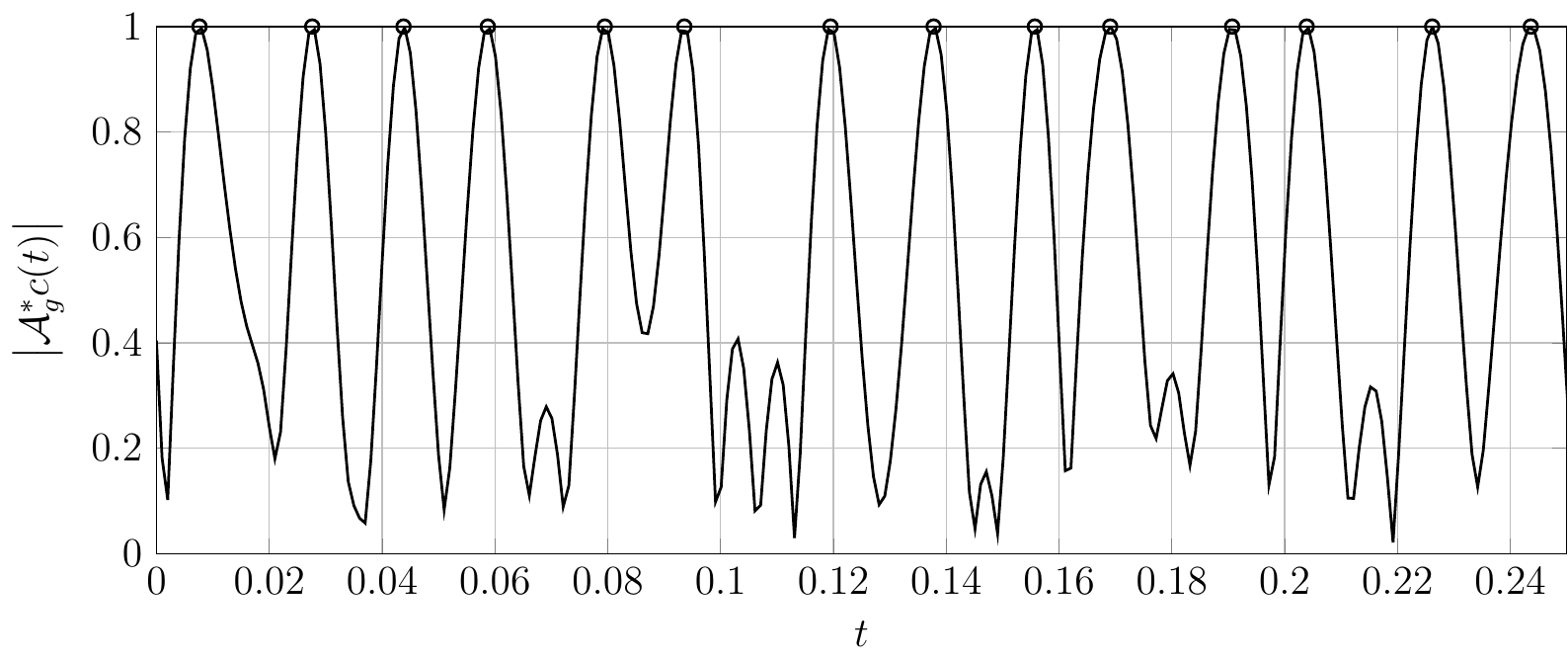}
            	}{
            		\ifthenelse{\equal{\generateFigs}{yes}}{\tikzsetnextfilename{interpolationDualPolynomial2BW}}{}
                        	\begin{tikzpicture}
                        		\begin{axis}
                        			[width = 0.97\columnwidth, height = 7cm,
                        			xlabel = $t$,
                        			ylabel = $\abs{\gaborOp^*c(t)}$,
                        			grid = both,
                        			enlarge x limits = false, enlarge y limits = false,
                        			xmin = 0, xmax = 0.25,
                        			xticklabel style = /pgf/number format/fixed,
                        			ymin = 0, ymax = 1]
                        			\addplot[thick, black, only marks, mark = o] file {data/exampleInterpolation/measureSupport.dat};
                        			\addplot[thick, black, no markers] file {data/exampleInterpolation/dualPolynomialSTFT.dat};
                        		\end{axis}
                        	\end{tikzpicture}
		}
	}
	
	\caption{Example of reconstruction of a discrete measure $\meas$ with $\abs{\supp(\meas)} = 54$ and $\Delta = 0.011880$ from both pure Fourier and STFT measurements with $\freq_c = 50$ and $N = 25$. The curves represent the dual polynomials and the circles represent the elements in the original measure's support set.}
	\label{fig: dual polynomial}
\end{figure}

%% file: proofMainResultReal.tex

\section{Proof of Theorem~\ref{thm: exact recovery}}
\label{section: proof main result}

Let $\varepsilon = \{\varepsilon_\measSpikeIdx\}_{\measSpikeIdx \in \measSupportIdx}$ be a sequence of complex unit-magnitude numbers. The goal is to construct a function $c_0 \in L^\infty(\R^2)$ such that $(\gaborOp^*c_0)(t_\measSpikeIdx) = \varepsilon_\measSpikeIdx$, for all $\measSpikeIdx \in \measSupportIdx$, and $\abs{(\gaborOp^*c_0)(t)} < 1$ for all $t \in \R \setminus \measSupport$, where $\measSupport = \{t_\measSpikeIdx\}_{\measSpikeIdx \in \measSupportIdx}$ is the support set of $\meas$. Inspired by~\cite{Candes2012}, we take $c_0$ to be of the form
\begin{equation}
	\forall (\tau, \freq) \in \R^2, \quad c_0(\tau, \freq) \triangleq \frac{1}{2\freq_c}\sum_{\measSpikeIdx \in \measSupportIdx} \left(\alpha_\measSpikeIdx\, \window(t_\measSpikeIdx - \tau)e^{-2\pi i\freq t_\measSpikeIdx} + \beta_\measSpikeIdx\, \window'(t_\measSpikeIdx - \tau)e^{-2\pi i\freq t_\measSpikeIdx}\right),
	\label{eq: choice of c0}
\end{equation}
where $\alpha_\measSpikeIdx, \beta_\measSpikeIdx \in \C$, for all $\measSpikeIdx \in \measSupportIdx$, and we proceed as follows:

\begin{enumerate}

	\item We first verify that $\alpha \triangleq \{\alpha_\measSpikeIdx\}_{\measSpikeIdx \in \measSupportIdx}$ and $\beta \triangleq \{\beta_\measSpikeIdx\}_{\measSpikeIdx \in \measSupportIdx}$ both in $ \ell^\infty(\measSupportIdx)$ implies $c_0 \in L^\infty(\R^2)$.

	\item Next, we show that one can find $\alpha, \beta \in \ell^\infty(\measSupportIdx)$ such that the interpolation conditions $(\gaborOp^*c_0)(t_\measSpikeIdx) = \varepsilon_\measSpikeIdx$, for all $\measSpikeIdx \in \measSupportIdx$, are satisfied and $\abs{\gaborOp^*c_0}$ has a local extremum at every $t_\measSpikeIdx$, $\measSpikeIdx \in \measSupportIdx$.
	
	\item Then, we verify, with $\alpha, \beta \in \ell^\infty(\measSupportIdx)$ chosen as in 2., that the magnitude of $\gaborOp^*c_0$ is indeed strictly smaller than $1$ outside the support set $\measSupport = \{t_\measSpikeIdx\}_{\measSpikeIdx \in \measSupportIdx}$ of $\meas$. This will be accomplished in two stages. First, we show that $\abs{\gaborOp^*c_0}$ is strictly smaller than $1$ ``away'' from each point $t_\measSpikeIdx$, $\measSpikeIdx \in \measSupportIdx$, specifically, on $\R \setminus \bigcup_{\measSpikeIdx \in \measSupportIdx} [t_\measSpikeIdx - \frac{1}{7\freq_c}, t_\measSpikeIdx + \frac{1}{7\freq_c}]$. We then complete the proof by establishing that $\abs{\gaborOp^*c_0}$ is strictly concave on each set $[t_\measSpikeIdx - \frac{1}{7\freq_c}, t_\measSpikeIdx + \frac{1}{7\freq_c}]$, $\measSpikeIdx \in \measSupportIdx$, which, combined with the fact that $\abs{(\gaborOp^*c_0)(t_\measSpikeIdx)} = 1$, for every $\measSpikeIdx \in \measSupportIdx$, implies that $\abs{\gaborOp^*c_0}$ is also strictly smaller than $1$ on each of these sets.

\end{enumerate}

The main conceptual components in our proof are due to Cand\`es and Fernandez-Granda \cite{Candes2012}. 
Although \cite{Candes2012} considers recovery of measures on $\torus$ only and from pure Fourier measurements, we can still borrow technical ingredients from the proof of \cite[Thm.~1.2]{Candes2012}.
However, the different nature of the measurements and, in particular, the case $\group = \R$, pose additional technical challenges relative to the proof of \cite[Thm.~1.2]{Candes2012}. Specifically, the sum corresponding to \eqref{eq: choice of c0} in \cite{Candes2012} is always finite, whereas here it can be infinite, which presents us with delicate convergence issues that need to be addressed properly. Further fundamental differences between the proof in \cite{Candes2012} for the pure Fourier case and our proof stem from the choice of the interpolation kernel, which here is given by $t \mapsto \autocorrelation(t)\sinc(2\pi\freq_c t)$. Specifically, we do not have to impose a bandwidth constraint on the interpolation kernel. 
For pure Fourier measurements, on the other hand, the interpolation kernel has to be band-limited to $[-\freq_c, \freq_c]$ (Cand\`es and Fernandez-Granda~\cite{Candes2012} use the square of the Fej\'er kernel which offers a good trade-off between localization in time and frequency).  
As already mentioned in the main body, this leads to a factor-of-two improvement in the minimum spacing condition for STFT measurements over pure Fourier measurements. Note, however, that STFT measurements, owing to their redundancy, provide more information than pure Fourier measurements. We finally note that our proof also borrows a number of technical results from~\cite{Kahane2011}.

\subsection{$\alpha, \beta \in \ell^\infty(\measSupportIdx)$ implies $c_0 \in L^\infty(\R^2)$}
\label{subsection: c0 is in Linfinty}
Let $\alpha \triangleq \{\alpha_\measSpikeIdx\}_{\measSpikeIdx \in \measSupportIdx} \in \ell^\infty(\measSupportIdx)$ and $\beta \triangleq \{\beta_\measSpikeIdx\}_{\measSpikeIdx \in \measSupportIdx} \in \ell^\infty(\measSupportIdx)$. Take $(\tau, \freq) \in \R^2$ and define $\measSpikeIdx_0, \measSpikeIdx_1 \in \measSupportIdx$ to be the indices of the points in $\measSupport$, that are closest and second closest, respectively, to $\tau$, that is, 
\begin{equation*}
	\measSpikeIdx_0 \triangleq \argmin_{\measSpikeIdx \in \measSupportIdx} \abs{t_\measSpikeIdx - \tau} \qquad \text{and} \qquad  \measSpikeIdx_1 \triangleq \argmin_{\measSpikeIdx \in \measSupportIdx \setminus \{\measSpikeIdx_0\}} \abs{t_\measSpikeIdx - \tau}.
\end{equation*}
For brevity of exposition, we detail the case $t_{\measSpikeIdx_0} \leq \tau \leq t_{\measSpikeIdx_1}$ only, the cases $t_{\measSpikeIdx_1} \leq \tau \leq t_{\measSpikeIdx_0}$, $t_{\measSpikeIdx_1} \leq t_{\measSpikeIdx_0} \leq \tau$, and $\tau \leq t_{\measSpikeIdx_0} \leq t_{\measSpikeIdx_1}$ are all dealt with similarly. For all $\measSpikeIdx \in \measSupportIdx_\tau^{-} \triangleq \{\measSpikeIdxBis \in \measSupportIdx \colon t_\measSpikeIdxBis < t_{\measSpikeIdx_0}\}$, it then holds that 
\begin{equation}
	\tau - t_\measSpikeIdx \geq  t_{\measSpikeIdx_0} - t_\measSpikeIdx \geq \Delta, 
	\label{eq: schema 1}
\end{equation}
and for all $\measSpikeIdx \in \measSupportIdx_\tau^{+} \triangleq \{\measSpikeIdxBis \in \measSupportIdx \colon t_\measSpikeIdxBis > t_{\measSpikeIdx_1}\}$, we have 
\begin{equation}
	t_\measSpikeIdx - \tau \geq t_\measSpikeIdx - t_{\measSpikeIdx_1} \geq \Delta. 
	\label{eq: schema 2}
\end{equation}
Hence, we get the following:
\begin{align}
	\abs{c_0(\tau, \freq)} &\leq \frac{1}{2\freq_c}\left(\normlinfty{\alpha}\sum_{\measSpikeIdx \in \measSupportIdx}\abs{\window(t_\measSpikeIdx - \tau)} + \normlinfty{\beta}\sum_{\measSpikeIdx \in \measSupportIdx}\abs{\window'(t_\measSpikeIdx - \tau)}\right) \notag \\
		&= \frac{\normlinfty{\alpha}}{2\freq_c}\abs{\window(t_{\measSpikeIdx_0} - \tau)} + \frac{\normlinfty{\alpha}}{2\freq_c}\abs{\window(t_{\measSpikeIdx_1} - \tau)} + \frac{\normlinfty{\beta}}{2\freq_c}\abs{\window'(t_{\measSpikeIdx_0} - \tau)}  \label{eq: sum 0 proof c_0 in Linfty} \\
		&\kern-1.2cm + \frac{\normlinfty{\beta}}{2\freq_c}\abs{\window'(t_{\measSpikeIdx_1} - \tau)} + \frac{\normlinfty{\alpha}}{2\freq_c}\sum_{\measSpikeIdx \in \measSupportIdx \setminus \{\measSpikeIdx_0, \measSpikeIdx_1\}} \abs{\window(t_{\measSpikeIdx} - \tau)}  + \frac{\normlinfty{\beta}}{2\freq_c}\sum_{\measSpikeIdx \in \measSupportIdx \setminus \{\measSpikeIdx_0, \measSpikeIdx_1\}} \abs{\window'(t_\measSpikeIdx - \tau)} \notag \\
		&\leq \frac{\normlinfty{\alpha}}{\freq_c}\normInfty{\window} + \frac{\normlinfty{\beta}}{\freq_c}\normInfty{\window'} + \frac{\normlinfty{\alpha}}{2\freq_c}\sum_{\measSpikeIdx \in \measSupportIdx_\tau^-} \abs{\window(t_{\measSpikeIdx_0} - t_\measSpikeIdx)} \label{eq: sum 1 proof c_0 in Linfty} \\
		&\kern-1.2cm +  \frac{\normlinfty{\alpha}}{2\freq_c}\sum_{\measSpikeIdx \in \measSupportIdx_\tau^+} \abs{\window(t_{\measSpikeIdx} - t_{\measSpikeIdx_1})} + \frac{\normlinfty{\beta}}{2\freq_c}\sum_{\measSpikeIdx \in \measSupportIdx_\tau^-} \abs{\window'(t_{\measSpikeIdx_0} - t_\measSpikeIdx)}+ \frac{\normlinfty{\beta}}{2\freq_c}\sum_{\measSpikeIdx \in \measSupportIdx_\tau^+} \abs{\window'(t_\measSpikeIdx - t_{\measSpikeIdx_1})}, \label{eq: sum 2 proof c_0 in Linfty}
\end{align}
where the step from \eqref{eq: sum 0 proof c_0 in Linfty} to \eqref{eq: sum 1 proof c_0 in Linfty}-\eqref{eq: sum 2 proof c_0 in Linfty} follows from $\window, \window' \in C_b(\R)$, \eqref{eq: schema 1}, \eqref{eq: schema 2}, and the fact that $\abs{\window}$ and $\abs{\window'}$ are both symmetric and non-increasing on $[\Delta, \infty)$, which is by the assumption $\Delta > 4\sigma$. Note that we eliminated the dependence of the upper bound in \eqref{eq: sum 1 proof c_0 in Linfty}-\eqref{eq: sum 2 proof c_0 in Linfty} on $(\tau, \freq)$. It remains to establish that every sum in the upper bound \eqref{eq: sum 1 proof c_0 in Linfty}-\eqref{eq: sum 2 proof c_0 in Linfty} is finite.
The minimum separation between pairs of points of $\measSupport = \{t_\measSpikeIdx\}_{\measSpikeIdx \in \measSupportIdx}$ is $\Delta$, by assumption. Consequently, since $\abs{\window}$ and $\abs{\window'}$ are both symmetric and non-increasing on $[\Delta, \infty)$, the sums in \eqref{eq: sum 1 proof c_0 in Linfty} and \eqref{eq: sum 2 proof c_0 in Linfty} take on their maxima when the points $t_\measSpikeIdx$, $\measSpikeIdx \in \measSupportIdx$,  are equi-spaced on $\R$ with spacing $\Delta$, i.e., when
\begin{align*}
	\big\{t_{\measSpikeIdx_0} - t_\measSpikeIdx \colon \measSpikeIdx \in \measSupportIdx_\tau^-\big\} &\subseteq \big\{n\Delta \colon n \in \N \setminus \{0\}\big\} \\
	\big\{t_\measSpikeIdx - t_{\measSpikeIdx_1} \colon \measSpikeIdx \in \measSupportIdx_\tau^+\big\} &\subseteq \big\{n\Delta \colon n \in \N \setminus \{0\}\big\}.
\end{align*}
It therefore follows that
\begin{equation}
	\abs{c_0(\tau, \freq)} \leq  \frac{\normlinfty{\alpha}}{\freq_c}\normInfty{\window} + \frac{\normlinfty{\beta}}{\freq_c}\normInfty{\window'} + \frac{\normlinfty{\alpha}}{\freq_c}\sum_{n = 1}^\infty \abs{\window(n\Delta)} + \frac{\normlinfty{\beta}}{\freq_c}\sum_{n = 1}^\infty \abs{\window'(n\Delta)}, \label{eq: finite sum *}
\end{equation}
where
\begin{align*}
	\sum_{n = 1}^\infty \abs{\window(n\Delta)} &= \frac{1}{\sqrt{\sigma}}\sum_{n = 1}^\infty\exp\!\left(-\frac{\pi n^2\Delta^2}{2\sigma^2}\right) < \infty \\
	\sum_{n = 1}^\infty \abs{\window'(n\Delta)} &=  \frac{\pi}{\sigma^2\sqrt{\sigma}} \sum_{n = 1}^\infty n\Delta\exp\!\left(-\frac{\pi n^2\Delta^2}{2\sigma^2}\right) < \infty,
\end{align*}
which establishes that $c_0 \in L^\infty(\R^2)$.

\subsection{Existence of $\alpha, \beta \in \ell^\infty(\measSupportIdx)$ such that $(\gaborOp^*c_0)(t_\measSpikeIdx) = \varepsilon_\measSpikeIdx$ and $\abs{\gaborOp^*c_0}$ has a local extremum at every $t_\measSpikeIdx$, $\measSpikeIdx \in \measSupportIdx$}

Using~\eqref{eq: choice of c0} in \eqref{eq: adjoint operator expression}, we get
\begin{align}
	\forall t &\in \R, \quad (\gaborOp^* c_0)(t) = \int_{-\freq_c}^{\freq_c}\int_\R c_0(\tau, \freq)\window(t - \tau)e^{2\pi i\freq t}\mathrm{d}\tau\mathrm{d}\freq \notag \\
		&= \int_{-\freq_c}^{\freq_c} \int_\R \frac{1}{2\freq_c}\sum_{\measSpikeIdx \in \measSupportIdx} \Big(\alpha_\measSpikeIdx \window(t_\measSpikeIdx - \tau)e^{-2\pi i\freq t_\measSpikeIdx} + \beta_\measSpikeIdx \window'(t_\measSpikeIdx - \tau)e^{-2\pi i\freq t_\measSpikeIdx}\Big)\window(t - \tau) e^{2\pi i\freq t} \mathrm{d}\tau\mathrm{d}\freq \label{eq: integration term by term dual 1} \\
		&= \sum_{\measSpikeIdx \in \measSupportIdx} \left[\alpha_\measSpikeIdx \left(\int_\R \window(t_\measSpikeIdx - \tau)\window(t - \tau)\mathrm{d}\tau\right) \left(\frac{1}{2\freq_c}\int_{-\freq_c}^{\freq_c} e^{2\pi i\freq(t - t_\measSpikeIdx)}\mathrm{d}\freq\right) \right. \label{eq: integration term by term dual 2}  \\
		&\qquad\qquad\qquad\qquad \left. + \beta_\measSpikeIdx \left(\int_\R \window'(t_\measSpikeIdx - \tau)\window(t - \tau)\mathrm{d}\tau\right) \left(\frac{1}{2\freq_c}\int_{-\freq_c}^{\freq_c} e^{2\pi i\freq(t - t_\measSpikeIdx)}\mathrm{d}\freq \right)\right]  \label{eq: integration term by term dual}  \\
		&= \sum_{\measSpikeIdx \in \measSupportIdx} \Big(\alpha_\measSpikeIdx \underbrace{\autocorrelation(t - t_\measSpikeIdx)\sinc(2\pi\freq_c(t - t_\measSpikeIdx))}_{\triangleq u(t - t_\measSpikeIdx)} + \beta_\measSpikeIdx \underbrace{\autocorrelation'(t_\measSpikeIdx - t)\sinc(2\pi\freq_c(t - t_\measSpikeIdx))}_{\triangleq v(t - t_\measSpikeIdx)}\Big), \label{eq: function dual certificate proof}
\end{align}
where we set
\begin{align}
	\forall t \in \R, \quad u(t) &\triangleq \autocorrelation(t)\sinc(2\pi\freq_c t) \label{eq: definition function u} \\
	\forall t \in \R, \quad v(t) &\triangleq \autocorrelation'(-t)\sinc(2\pi\freq_c t) = \frac{\pi t}{2\sigma^2}\autocorrelation(t)\sinc(2\pi\freq_c t) \label{eq: definition function v}
\end{align}
with
\begin{equation*}
	\forall t \in \R, \quad \autocorrelation(t) = \exp\!\left(-\frac{\pi t^2}{4\sigma^2}\right)
\end{equation*}
as defined in~\eqref{eq: autocorrelation function}. The conditions for Fubini's Theorem, applied in the step from \eqref{eq: integration term by term dual 1}-\eqref{eq: integration term by term dual 2} to \eqref{eq: integration term by term dual}, can be verified as follows:
\begin{align}
	\sum_{\measSpikeIdx \in \measSupportIdx} \int_{-\freq_c}^{\freq_c} &\int_\R \abs{\alpha_\measSpikeIdx \window(t_\measSpikeIdx - \tau)\window(t - \tau)e^{2\pi i\freq(t - t_\measSpikeIdx)}}\mathrm{d}\tau\mathrm{d}\freq \notag\\
		&\leq \sum_{\measSpikeIdx \in \measSupportIdx} \abs{\alpha_\measSpikeIdx} \left(\int_\R \abs{\window(t_\measSpikeIdx - \tau)\window(t - \tau)}\mathrm{d}\tau\right)\left(\frac{1}{2\freq_c}\int_{-\freq_c}^{\freq_c} \abs{e^{2\pi i\freq t}}\mathrm{d}\freq\right) \notag \\
		&=\sum_{\measSpikeIdx \in \measSupportIdx} \abs{\alpha_\measSpikeIdx}\autocorrelation(t - t_\measSpikeIdx) \leq 2\normlinfty{\alpha}\left(\normInfty{\autocorrelation} + \sum_{n = 1}^\infty \autocorrelation(n\Delta)\right) \label{eq: finite sum 1}
\end{align}
and
\begin{align}
	\sum_{\measSpikeIdx \in \measSupportIdx} \int_{-\freq_c}^{\freq_c} &\int_\R \abs{\beta_\measSpikeIdx \window'(t_\measSpikeIdx - \tau)\window(t - \tau)e^{2\pi i\freq(t - t_\measSpikeIdx)}}\mathrm{d}\tau\mathrm{d}\freq \notag\\
		&\leq \sum_{\measSpikeIdx \in \measSupportIdx} \abs{\beta_\measSpikeIdx} \left(\int_\R \abs{\window'(t_\measSpikeIdx - \tau)\window(t - \tau)}\mathrm{d}\tau\right)\left(\frac{1}{2\freq_c}\int_{-\freq_c}^{\freq_c} \abs{e^{2\pi i\freq t}}\mathrm{d}\freq\right) \notag \\
		&\leq \sum_{\measSpikeIdx \in \measSupportIdx} \abs{\beta_\measSpikeIdx}\widetilde{\autocorrelation}(t_\measSpikeIdx-t) \leq 2\normlinfty{\beta}\left(\|\widetilde{\autocorrelation}\|_\infty + \sum_{n = 1}^\infty \widetilde{\autocorrelation}(n\Delta)\right), \label{eq: finite sum 2}
\end{align}
where we used 
\begin{align*}
	\forall t \in \R, \quad \int_\R \window(\tau)\abs{\window'(t + \tau)}\mathrm{d}\tau &= \frac{1}{\sigma}\exp\!\left(-\frac{\pi t^2}{2\sigma^2}\right) - \autocorrelation'(t)\erf\left(\frac{\sqrt{\pi}}{2\sigma}t\right) \\
		&\leq \frac{1}{\sigma}\exp\!\left(-\frac{\pi t^2}{2\sigma^2}\right) + \abs{\autocorrelation'(t)} \reversetriangleq \widetilde{\autocorrelation}(t) 
\end{align*}
and the fact that $\widetilde{\autocorrelation}$ is bounded, symmetric, and non-decreasing on $[\Delta, \infty)$ as a consequence of $\Delta > 4\sigma$. The upper bounds in~\eqref{eq: finite sum 1} and \eqref{eq: finite sum 2} are both finite as the series $\sum_{n \geq 1} \autocorrelation(n\Delta)$ and $\sum_{n \geq 1} \widetilde{\autocorrelation}(n\Delta)$ converge.

We have shown in Lemma~\ref{lem: lemma for adjoint operator} that  for $c_0 \in L^\infty(\R^2)$, the function $\gaborOp^*c_0$ is in $C_b(\R)$. With $c_0$ taken as in \eqref{eq: choice of c0}, $\gaborOp^*c_0$ is not only in $C_b(\R)$, but also differentiable, as we show next.
We start by noting that the functions $u$ and $v$ defined in \eqref{eq: definition function u} and \eqref{eq: definition function v} are differentiable on $\R$, and their derivatives are given by
\begin{align*}
	\forall t \in \R, \quad u'(t) &= \autocorrelation'(t)\sinc(2\pi\freq_c t) + 2\pi\freq_c\,\autocorrelation(t)\sinc'(2\pi\freq_c t) \\
	\forall t \in \R, \quad v'(t) &= -\autocorrelation''(-t)\sinc(2\pi\freq_c t) + 2\pi\freq_c\,\autocorrelation'(-t)\sinc'(2\pi\freq_c t) \\
		&=\left[\left(\frac{\pi}{2\sigma^2} - \frac{\pi^2t^2}{4\sigma^4}\right)\sinc(2\pi\freq_c t) + \frac{\pi^2\freq_c t}{\sigma^2}\sinc'(2\pi\freq_c t)\right]\autocorrelation(t).
\end{align*}
Then, using 
\begin{equation}
	\forall t \in \R\setminus \{0\}, \quad \abs{\sinc(t)}  \leq \frac{1}{\abs{t}} \qquad \text{and} \qquad \abs{\sinc'(t)} = \abs{\frac{\cos(t)}{t} - \frac{\sin(t)}{t^2}} \leq \frac{1}{\abs{t}} + \frac{1}{\abs{t}^2}, \label{eq: inequalities sinc and derivative}
\end{equation}
we obtain the following upper bounds on $u'$ and  $v'$:
\begin{align}
	\forall t \in \R\setminus\{0\}, \quad \abs{u'(t)} &\leq \left(\frac{1}{4\sigma^2\freq_c} + \frac{1}{\abs{t}} + \frac{1}{2\pi\freq_c\abs{t}^2}\right)\autocorrelation(t) \reversetriangleq U(t)
	\label{eq: upper bound u prime} \\
	\forall t \in \R\setminus\{0\}, \quad \abs{v'(t)} &\leq \left(\frac{1}{2\sigma^2\freq_c\abs{t}} + \frac{\pi}{2\sigma^2} + \frac{\pi\abs{t}}{8\sigma^4\freq_c}\right)\autocorrelation(t) \reversetriangleq V(t).
	\label{eq: upper bound v}
\end{align}
Next, we establish that $\sum_{\measSpikeIdx \in \measSupportIdx} \left(\alpha_\measSpikeIdx u'(t - t_\measSpikeIdx) + \beta_\measSpikeIdx v'(t - t_\measSpikeIdx)\right)$ converges uniformly on every compact set $[-r, r]$, $r > 0$, so that we can apply~\cite[Thm.~V.2.14]{Colmez2009} to show that the series in \eqref{eq: function dual certificate proof} can be differentiated term by term.  
For $r > 0$, we have 
\begin{align*}
	\sum_{\measSpikeIdx \in \measSupportIdx} &\sup_{t \in [-r, r]} \abs{\alpha_\measSpikeIdx u'(t - t_\measSpikeIdx) + \beta_\measSpikeIdx v'(t - t_\measSpikeIdx)} \leq \sum_{\measSpikeIdx \in \measSupportIdx} \left(\normlinfty{\alpha}\hspace{-0.2cm}\sup_{t \in [-r, r]} \!\!\abs{u'(t - t_\measSpikeIdx)} + \normlinfty{\beta}\hspace{-0.2cm}\sup_{t \in [-r, r]}\!\!\abs{v'(t - t_\measSpikeIdx)} \right) \\
		&= \normlinfty{\alpha}\left(\sum_{\measSpikeIdx \in \measSupportIdx_r} \sup_{t \in [-r, r]} \abs{u'(t - t_\measSpikeIdx)} + \sum_{\measSpikeIdx \in \measSupportIdx_r^+} \sup_{t \in [-r, r]} \abs{u'(t - t_\measSpikeIdx)}  + \sum_{\measSpikeIdx \in \measSupportIdx_r^-} \sup_{t \in [-r, r]} \abs{u'(t - t_\measSpikeIdx)} \right) \\
		&\qquad + \normlinfty{\beta}\left(\sum_{\measSpikeIdx \in \measSupportIdx_r} \sup_{t \in [-r, r]} \abs{v'(t - t_\measSpikeIdx)} + \sum_{\measSpikeIdx \in \measSupportIdx_r^+} \sup_{t \in [-r, r]} \abs{v'(t - t_\measSpikeIdx)}  + \sum_{\measSpikeIdx \in \measSupportIdx_r^-} \sup_{t \in [-r, r]} \abs{v'(t - t_\measSpikeIdx)} \right),
\end{align*}
where we defined the sets $\measSupportIdx_r \triangleq \left\{\measSpikeIdx \in \measSupportIdx \colon t_\measSpikeIdx \in [-r, r]\right\}$, $\measSupportIdx_r^+ \triangleq \left\{\measSpikeIdx \in \measSupportIdx \colon t_\measSpikeIdx > r\right\}$, and $\measSupportIdx_r^- \triangleq \left\{\measSpikeIdx \in \measSupportIdx \colon t_\measSpikeIdx < -r\right\}$.
The functions $U$ and $V$ are both positive and symmetric, $U$ is non-increasing on $(0, \infty)$, and $V$ is non-increasing on $(0, \infty)$ as
\begin{align*}
	\forall t \in (0, \infty), \quad V'(t) &= \left(-\frac{1}{2\sigma^2\freq_c t^2} + \frac{\pi}{8\sigma^4\freq_c}\right)\autocorrelation(t) + \left(\frac{1}{2\sigma^2\freq_c t} + \frac{\pi}{2\sigma^2} + \frac{\pi t}{8\sigma^4\freq_c}\right)\autocorrelation'(t) \\
		&= -\left(\frac{1}{2\sigma^2\freq_c t^2} + \frac{\pi}{8\sigma^4\freq_c} + \left(\frac{\pi}{2\sigma^2}\right)^2t + \frac{\pi^2 t^2}{16\sigma^6\freq_c}\right)\autocorrelation(t) \leq 0.
\end{align*}
It therefore follows that
\begin{equation*}
	\sup_{t \in [-r, r]} \abs{u'(t- t_\measSpikeIdx)} \leq \begin{cases} \normInfty{u'}, & t_\measSpikeIdx \in [-r, r] \\ 
		U(r - t_\measSpikeIdx), & t_\measSpikeIdx > r \\
		U(-r - t_\measSpikeIdx), & t_\measSpikeIdx < -r \end{cases}
\end{equation*}
and
\begin{equation*}
	\sup_{t \in [-r, r]} \abs{v'(t-t_\measSpikeIdx)} \leq \begin{cases} \normInfty{v'}, & t_\measSpikeIdx \in [-r, r] \\ 
		V(r - t_\measSpikeIdx), & t_\measSpikeIdx > r \\
		V(-r - t_\measSpikeIdx) & t_\measSpikeIdx < -r. \end{cases}
\end{equation*}
Since the support set $\measSupport = \{t_\measSpikeIdx\}_{\measSpikeIdx \in \measSupportIdx}$ is closed and uniformly discrete, by assumption, and $[-r, r]$ is compact, the set $\measSupport \cap [-r, r]$, and thereby the index set $\measSupportIdx_r$, contains a finite number of elements, say $L_r \in \N$. We thus have the following
\begin{align*}
	\sum_{\measSpikeIdx \in \measSupportIdx} &\sup_{t \in [-r, r]} \abs{\alpha_\measSpikeIdx u'(t - t_\measSpikeIdx) + \beta_\measSpikeIdx v'(t - t_\measSpikeIdx)} \leq \normlinfty{\alpha}\!\!\left(\!L_r\normInfty{u'} + \sum_{\measSpikeIdx \in \measSupportIdx_r^+} U(r - t_\measSpikeIdx) + \sum_{\measSpikeIdx \in \measSupportIdx_r^-} U(-r - t_\measSpikeIdx) \!\!\right) \\
	&\qquad\qquad + \normlinfty{\beta}\!\!\left(L_r\normInfty{v'} + \sum_{\measSpikeIdx \in \measSupportIdx_r^+} V(r - t_\measSpikeIdx) + \sum_{\measSpikeIdx \in \measSupportIdx_r^-} V(-r - t_\measSpikeIdx) \right) \\
	&\hspace{-0.6cm}\leq \normlinfty{\alpha}\!\!\left(\!L_r\normInfty{u'} + 2\normInfty{U} + 2\sum_{n = 1}^\infty U(n\Delta)\!\!\right) + \normlinfty{\beta}\!\!\left(\!L_r\normInfty{v'} + 2\normInfty{V} + 2\sum_{n = 1}^\infty V(n\Delta)\!\!\right)\!,
\end{align*}
where we isolated the points in $\measSupport$ that are closest to $r$ and $-r$ as in~\eqref{eq: sum 1 proof c_0 in Linfty}-\eqref{eq: sum 2 proof c_0 in Linfty} and we used the fact that a regular spacing of the $t_\measSpikeIdx$, $\measSpikeIdx \in \measSupportIdx$, maximizes the sum as in \eqref{eq: finite sum *}.
Since $\sum_{n \geq 1} U(n\Delta) < \infty$ and $\sum_{n \geq 1} V(n\Delta) < \infty$, the Weierstrass M-test tells us that $\sum_{\measSpikeIdx \in \measSupportIdx} \left(\alpha_\measSpikeIdx u'(t - t_\measSpikeIdx) + \beta_\measSpikeIdx v'(t - t_\measSpikeIdx)\right)$ converges uniformly on every compact set $[-r, r]$, $r > 0$. Thanks to~\cite[Thm.~V.2.14]{Colmez2009} this implies that the function $\gaborOp^*c_0$ is differentiable on $\R$, and that its derivative equals
\begin{equation*}
	\forall t \in \R, \quad (\gaborOp^*c_0)'(t) = \sum_{\measSpikeIdx \in \measSupportIdx} \big(\alpha_\measSpikeIdx u'(t - t_\measSpikeIdx) + \beta_\measSpikeIdx v'(t - t_\measSpikeIdx)\big)
\end{equation*}
for $\alpha, \beta \in \ell^\infty(\measSupportIdx)$.
We next show that there exist $\alpha, \beta \in \ell^\infty(\measSupportIdx)$ such that $(\gaborOp^*c_0)(t_\measSpikeIdx) = \varepsilon_\measSpikeIdx$, for all $\measSpikeIdx \in \measSupportIdx$, and $\abs{(\gaborOp^*c_0)(t)} < 1$ for all $t \in \R \setminus \measSupport$. To this end, we seek $\alpha, \beta \in \ell^\infty(\measSupportIdx)$ such that
\begin{equation}
	\begin{array}{l}
		(\gaborOp^*c_0)(t_\measSpikeIdx) = \varepsilon_\measSpikeIdx \\
		(\gaborOp^*c_0)'(t_\measSpikeIdx) = 0,
	\end{array}
	\label{eq: inversion operator condition}
\end{equation}
for all $\measSpikeIdx \in \measSupportIdx$.
In developing an approach to solving the equation system~\eqref{eq: inversion operator condition}, it will turn out convenient to define the operators
\begin{equation*}
	\begin{array}{llll}
		\mathcal{U}_p \colon & \ell^\infty(\measSupportIdx) & \longrightarrow & \ell^\infty(\measSupportIdx) \\
						 & \alpha = \{\alpha_\measSpikeIdx\}_{\measSpikeIdx \in \measSupportIdx} & \longmapsto & \left\{\sum\limits_{\measSpikeIdxBis \in \measSupportIdx} \alpha_\measSpikeIdxBis u^{(p)}(t_\measSpikeIdx - t_\measSpikeIdxBis)\right\}_{\measSpikeIdx \in \measSupportIdx}
	\end{array}
\end{equation*}
and
\begin{equation*}
	\begin{array}{llll}
		\mathcal{V}_p \colon & \ell^\infty(\measSupportIdx) & \longrightarrow & \ell^\infty(\measSupportIdx) \\
						 & \beta = \{\beta_\measSpikeIdx\}_{\measSpikeIdx \in \measSupportIdx} & \longmapsto & \left\{\sum\limits_{\measSpikeIdxBis \in \measSupportIdx} \beta_\measSpikeIdxBis v^{(p)}(t_\measSpikeIdx - t_\measSpikeIdxBis)\right\}_{\measSpikeIdx \in \measSupportIdx},
	\end{array}
\end{equation*}
where $p \in \{0, 1\}$. We defer the proof of $\mathcal{U}_p$ and $\mathcal{V}_p$, $p \in \{0, 1\}$, mapping $\ell^\infty(\measSupportIdx)$ into $\ell^\infty(\measSupportIdx)$ to later.
The equation system~\eqref{eq: inversion operator condition} can now be expressed as
\begin{equation}
	\left\{\begin{array}{l}
		\mathcal{U}_0\alpha + \mathcal{V}_0\beta = \varepsilon \\
		\mathcal{U}_1\alpha + \mathcal{V}_1\beta = 0.
	\end{array}\right.
	\label{eq: inversion operator condition 2}
\end{equation}
If both $\mathcal{V}_1$ and $\mathcal{U}_0 - \mathcal{V}_0\mathcal{V}_1^{-1}\mathcal{U}_1$ are invertible, then, as in \cite{Candes2012}, one can choose $\alpha = (\mathcal{U}_0 - \mathcal{V}_0\mathcal{V}_1^{-1}\mathcal{U}_1)^{-1}\varepsilon$ and  $\beta = -\mathcal{V}_1^{-1}\mathcal{U}_1\alpha$ to satisfy \eqref{eq: inversion operator condition 2}. The Neumann expansion theorem~\cite[Thm.~1.3, p.~5]{Kubrusly2012} now says that $\normOp{\mathcal{I} - (v'(0))^{-1}\mathcal{V}_1} < 1$ and $\normOp{\mathcal{I} - (\mathcal{U}_0 - \mathcal{V}_0\mathcal{V}_1^{-1}\mathcal{U}_1)} < 1$ are sufficient conditions for $\mathcal{V}_1$ and $\mathcal{U}_0 - \mathcal{V}_0\mathcal{V}_1^{-1}\mathcal{U}_1$ to be invertible. 
We next verify these conditions.

\subsubsection{$\mathcal{V}_1$ is invertible}
\label{sec: V1 inversible}
Fix a sequence $\beta \in \ell^\infty(\measSupportIdx)$, define $\zeta = (\mathcal{I} - (v'(0))^{-1}\mathcal{V}_1)\beta$, and let $\measSpikeIdx \in \measSupportIdx$. We then have
\begin{align*}
	\zeta_\measSpikeIdx &= \beta_\measSpikeIdx - (v'(0))^{-1}\sum_{\measSpikeIdxBis \in \measSupportIdx} \beta_\measSpikeIdxBis v'(t_\measSpikeIdx - t_\measSpikeIdxBis) \\
		&= - (v'(0))^{-1}\sum_{\measSpikeIdxBis \in \measSupportIdx \setminus \{\measSpikeIdx\}} \beta_\measSpikeIdxBis v'(t_\measSpikeIdx - t_\measSpikeIdxBis) \\
		&= (\autocorrelation''(0))^{-1}\sum_{\measSpikeIdxBis \in \measSupportIdx \setminus \{\measSpikeIdx\}} \beta_\measSpikeIdxBis v'(t_\measSpikeIdx - t_\measSpikeIdxBis),
\end{align*}
where we used $v'(0) = -\autocorrelation''(0) > 0$. With \eqref{eq: upper bound v} and $\autocorrelation''(0) = -\frac{\pi}{2\sigma^2}$, we obtain
\begin{align}
	\abs{\zeta_\measSpikeIdx} &\leq \sum_{\measSpikeIdxBis \in \measSupportIdx \setminus \{\measSpikeIdx\}} \frac{2\sigma^2}{\pi}\abs{\beta_\measSpikeIdxBis} \abs{v'(t_\measSpikeIdx - t_\measSpikeIdxBis)} \notag \\
		&\leq \frac{2\sigma^2}{\pi}\normlinfty{\beta} \sum_{\measSpikeIdxBis \in \measSupportIdx \setminus \{\measSpikeIdx\}} V(t_\measSpikeIdx - t_\measSpikeIdxBis) \\
		&= \normlinfty{\beta} \sum_{\measSpikeIdxBis \in \measSupportIdx \setminus \{\measSpikeIdx\}} \left(\frac{1}{\pi\freq_c \abs{t_\measSpikeIdx - t_\measSpikeIdxBis}} + \frac{\abs{t_\measSpikeIdx - t_\measSpikeIdxBis}}{4\sigma^2\freq_c} + 1\right) \autocorrelation(t_\measSpikeIdx - t_\measSpikeIdxBis). \label{eq: proof V1 invertible bound for sum}
\end{align}
We further upper-bound \eqref{eq: proof V1 invertible bound for sum} using the same line of reasoning that led to~\eqref{eq: finite sum *}. Specifically, we make use of the fact that $V$ is non-increasing on $(0, \infty)$ and that the minimum distance between points in $\measSupport$ is $\Delta$. This implies that $\sum_{\measSpikeIdxBis \in \measSupportIdx \setminus \{\measSpikeIdx\}} V(t_\measSpikeIdx - t_\measSpikeIdxBis)$ is maximized for
\begin{equation*}
	\big\{t_\measSpikeIdx - t_\measSpikeIdxBis \colon \measSpikeIdxBis \in \measSupportIdx \setminus \{\measSpikeIdx\}\big\} \subseteq \big\{n\Delta \colon n \in \Z \setminus \{0\}\big\}.
\end{equation*}
With \eqref{eq: upper bound v} this gives 
\begin{equation}
	\abs{\zeta_\measSpikeIdx} \leq 2\normlinfty{\beta}\left(\sum_{n = 1}^\infty \frac{\autocorrelation(n\Delta)}{\pi\freq_c n\Delta} + \sum_{n = 1}^\infty \frac{n\Delta}{4\sigma^2\freq_c}\autocorrelation(n\Delta) + \sum_{n = 1}^\infty \autocorrelation(n\Delta)\right). \label{eq: previous equation}
\end{equation}
As $n \geq 1$, we have $\autocorrelation(n\Delta) = \exp\!\left(-\frac{\pi n^2 \Delta^2}{4\sigma^2}\right) \leq \exp\!\left(-\frac{\pi n \Delta^2}{4\sigma^2}\right)$, which when used in~\eqref{eq: previous equation} leads to a further upper bound in terms of the following power series
\begin{equation*}
	\forall x \in (-1, 1), \quad \ln(1 - x) = -\sum_{n = 1}^\infty \frac{x^n}{n}, \qquad \frac{x}{(1 - x)^2} = \sum_{n = 1}^\infty nx^n, \qquad \frac{x}{1 - x} = \sum_{n = 1}^{\infty} x^n,
\end{equation*}
all evaluated at $x = \exp(-\frac{\pi\Delta^2}{4\sigma^2}) < 1$. Putting things together, we obtain
\begin{align*}
	\normOp{\mathcal{I} - (v'(0))^{-1}\mathcal{V}_1} &= \sup_{\beta \neq 0} \frac{\normlinfty{\zeta}}{\normlinfty{\beta}} \leq -\frac{2}{\pi\freq_c\Delta}\ln\left(1 - \exp\!\left(-\frac{\pi \Delta^2}{4\sigma^2}\right)\right)  \\
	&+ \frac{\Delta \exp\!\left(-\frac{\pi \Delta^2}{4\sigma^2}\right)}{2\sigma^2\freq_c\left(1 - \exp\!\left(-\frac{\pi \Delta^2}{4\sigma^2}\right)\right)^2} + \frac{2\exp\!\left(-\frac{\pi \Delta^2}{4\sigma^2}\right)}{1 - \exp\!\left(-\frac{\pi \Delta^2}{4\sigma^2}\right)}.
\end{align*}
Defining the functions
\begin{align*}
	\forall x > 0, \quad \varphi(x) &\triangleq -\ln\left(1 - \exp\!\left(-\pi x^2\right)\right) \\
	\forall x > 0, \quad \psi(x) &\triangleq \frac{x^2 \exp\!\left(-\pi x^2\right)}{\left(1 - \exp\!\left(-\pi x^2\right)\right)^2} \\
	\forall x > 0, \quad \xi(x) &= \frac{\exp\!\left(-\pi x^2\right)}{1 - \exp\!\left(-\pi x^2\right)},
\end{align*}
we can then write
\begin{equation*}
	\normOp{\mathcal{I} - (v'(0))^{-1}\mathcal{V}_1} \leq \frac{\frac{2}{\pi}\varphi\!\left(\frac{\Delta}{2\sigma}\right) + 2\psi\!\left(\frac{\Delta}{2\sigma}\right)}{\freq_c\Delta} + 2\xi\!\left(\frac{\Delta}{2\sigma}\right).
\end{equation*}
The functions $\varphi$ and $\xi$ are non-increasing on $(0, \infty)$, as their derivatives satisfy
\begin{equation*}
	\forall x > 0, \qquad \varphi'(x) = \frac{-2\pi x \exp\!\left(-\pi x^2\right)}{1 - \exp\!\left(-\pi x^2\right)} \leq 0
\end{equation*}
\begin{equation*}
	\forall x > 0, \quad \xi'(x) = -\frac{2\pi x\exp(-\pi x^2)}{1 - \exp\!\left(-\pi x^2\right)} - \frac{2\pi x\exp\!\left(-2\pi x^2\right)}{\left(1 - \exp\!\left(-\pi x^2\right)\right)^2} \leq 0.
\end{equation*}
As for $\psi$, we first write
\begin{equation*}
	\forall x > 0, \qquad \psi(x) = \left(\frac{x\exp(-\pi x^2/2)}{1 - \exp(-\pi x^2)}\right)^2 = \left(\frac{x}{2\sinh(\pi x^2/2)}\right)^2,
\end{equation*}
and then show that the function 
\begin{equation*}
	\forall x > 0, \quad \eta(x) \triangleq \frac{x}{2\sinh(\pi x^2/2)}
\end{equation*}
is non-increasing on $(0,\infty)$
by computing its first derivative:
\begin{align*}
	\forall x > 0, \quad \eta'(x) &= \frac{2\sinh(\pi x^2/2) - 2\pi x^2\cosh(\pi x^2/2)}{\left(2\sinh(\pi x^2/2)\right)^2} \\
		&= \frac{\cosh(\pi x^2/2)\left(\tanh(\pi x^2/2) - \pi x^2\right)}{2\left(\sinh(\pi x^2/2)\right)^2} \leq 0,
\end{align*}
where the inequality is thanks to $\tanh(\pi x^2/2) \leq \pi x^2/2 \leq \pi x^2$, for all $x > 0$.
Therefore, the function $\psi$ is also non-increasing on $(0, \infty)$. Since by assumption $\Delta > 4\sigma$ and $\Delta > 1/\freq_c$, we get
\begin{align*}
	\normOp{\mathcal{I} - (v'(0))^{-1}\mathcal{V}_1} &\leq \frac{\frac{2}{\pi}\varphi(2) + 2\psi(2)}{\freq_c\Delta} + 2\xi(2) \\
		&< \frac{2\varphi(2)}{\pi} + 2\psi(2) + 2\xi(2) \leq 3.71 \cdot 10^{-5} < 1.
\end{align*}
It therefore follows that $\mathcal{V}_1$ is invertible. Furthermore, according to the Neumann expansion theorem, the operator norm of $\mathcal{V}_1^{-1}$ satisfies
\begin{equation}
	\normOp{\mathcal{V}_1^{-1}} \leq \frac{\abs{v'(0)}^{-1}}{1 - \normOp{\mathcal{I} - (v'(0))^{-1}\mathcal{V}_1}} \leq \frac{2\sigma^2}{\pi - 2(\varphi(2) + \pi\psi(2) + \pi\xi(2))}.
	\label{eq: bound operator norm inverse}
\end{equation}

\subsubsection{$\mathcal{U}_0 - \mathcal{V}_0\mathcal{V}_1^{-1}\mathcal{U}_1$ is invertible}
We start by noting that thanks to the triangle inequality,
\begin{align}
	\normOp{\mathcal{I} - (\mathcal{U}_0 - \mathcal{V}_0\mathcal{V}_1^{-1}\mathcal{U}_1)} &\leq \normOp{\mathcal{I} - \mathcal{U}_0} + \normOp{\mathcal{V}_0}\normOp{\mathcal{V}_1^{-1}}\normOp{\mathcal{U}_1} \notag \\
		&\leq  \normOp{\mathcal{I} - \mathcal{U}_0}  + \frac{\abs{v'(0)}^{-1} \normOp{\mathcal{V}_0}\normOp{\mathcal{U}_1}}{1 - \normOp{\mathcal{I} - (v'(0))^{-1}\mathcal{V}_1}}. \label{eq: norm operator triangle inequality}
\end{align}
An upper bound on $\normOp{\mathcal{I} - \mathcal{U}_0}$ can easily be derived using arguments similar to those employed in Section \ref{sec: V1 inversible} to get an upper bound on $\normOp{\mathcal{I} - (v'(0))^{-1}\mathcal{V}_1}$.
Specifically, for $\alpha \in \ell^\infty(\measSupportIdx)$, the sequence $\zeta = (\mathcal{I} - \mathcal{U}_0)\alpha$ obeys
\begin{align*}
	\abs{\zeta_\measSpikeIdx} &\leq \sum_{\measSpikeIdxBis \in \measSupportIdx \setminus \{\measSpikeIdx\}} \!\!\abs{\alpha_\measSpikeIdxBis} \abs{u(t_\measSpikeIdx - t_\measSpikeIdxBis)} \leq \normlinfty{\alpha}\!\!\!\sum_{\measSpikeIdxBis \in \measSupportIdx \setminus \{\measSpikeIdx\}} \frac{1}{2\pi\freq_c\abs{t_\measSpikeIdx - t_\measSpikeIdxBis}}\autocorrelation(t_\measSpikeIdx - t_\measSpikeIdxBis) \leq 2\normlinfty{\alpha} \sum_{n = 1}^\infty \frac{\autocorrelation(n\Delta)}{2\pi\freq_c n\Delta} \\
		&\leq 2\normlinfty{\alpha} \sum_{n = 1}^\infty \frac{1}{2\pi\freq_c n\Delta}\exp\!\left(-\frac{\pi n\Delta^2}{4\sigma^2}\right) = -\frac{\normlinfty{\alpha}}{\pi\freq_c\Delta} \ln\left(1 - \exp\!\left(-\frac{\pi\Delta^2}{4\sigma^2}\right)\right),
\end{align*}
for all $\measSpikeIdx \in \measSupportIdx$, where we used
\begin{equation*}
	\forall t \in \R \!\setminus\! \{0\}, \quad \abs{u(t)} \leq \frac{\autocorrelation(t)}{2\pi\freq_c\abs{t}}.
\end{equation*}
This implies that 
 \begin{equation}
 	\normOp{\mathcal{I} - \mathcal{U}_0} \leq \frac{\varphi\!\left(\frac{\Delta}{2\sigma}\right)}{\pi\freq_c\Delta} \leq \frac{\varphi(2)}{\pi\freq_c\Delta} < \frac{\varphi(2)}{\pi}.
	\label{eq: operator norm I - U0}
 \end{equation}
Next, we compute an upper bound on $\normOp{\mathcal{U}_1}$. To this end, we fix $\alpha \in \ell^\infty(\measSupportIdx)$ and set $\zeta = \mathcal{U}_1\alpha$. Since $\autocorrelation'(0) = 0$ and $\sinc'(0) = 0$, we have $u'(0) = 0$, which, combined with \eqref{eq: upper bound u prime} gives, for all $\measSpikeIdx \in \measSupportIdx$,
\begin{align}
	\abs{\zeta_\measSpikeIdx} &\leq \sum_{\measSpikeIdxBis \in \measSupportIdx} \abs{\alpha_\measSpikeIdxBis} \abs{u'(t_\measSpikeIdx - t_\measSpikeIdxBis)} = \sum_{\measSpikeIdxBis \in \measSupportIdx \setminus \{\measSpikeIdx\}} \abs{\alpha_\measSpikeIdxBis} \abs{u'(t_\measSpikeIdx - t_\measSpikeIdxBis)} \notag \\
		&\leq \normlinfty{\alpha} \sum_{\measSpikeIdxBis \in \measSupportIdx \setminus \{\measSpikeIdx\}} \left(\frac{1}{4\sigma^2\freq_c} + \frac{1}{\abs{t_\measSpikeIdx - t_\measSpikeIdxBis}} + \frac{1}{2\pi\freq_c\abs{t_\measSpikeIdx - t_\measSpikeIdxBis}^2}\right)\autocorrelation(t_\measSpikeIdx - t_\measSpikeIdxBis) \notag \\
		&\leq 2\normlinfty{\alpha} \left(\frac{1}{4\freq_c\sigma^2}\sum_{n = 1}^{\infty}\autocorrelation(n\Delta) + \sum_{n = 1}^{\infty} \frac{\autocorrelation(n\Delta)}{n\Delta}  + \sum_{n = 1}^{\infty} \frac{\autocorrelation(n\Delta)}{2\pi\freq_cn^2\Delta^2} \right) \notag \\
		&\leq \normlinfty{\alpha} \left[\frac{1}{2\freq_c\sigma^2} \sum_{n = 1}^{\infty} \exp\!\left(-\frac{\pi n\Delta^2}{4\sigma^2}\right)+ \left(\frac{2}{\Delta}+  \frac{1}{\pi\freq_c\Delta^2}\right)\sum_{n = 1}^{\infty} \frac{1}{n}\exp\!\left(-\frac{\pi n\Delta^2}{4\sigma^2}\right)\right] \label{eq: upper bound proof U1 invertible} \\
		&=  \normlinfty{\alpha}\left[\frac{\exp\!\left(-\frac{\pi \Delta^2}{4\sigma^2}\right)}{2\freq_c\sigma^2\left(1 - \exp\!\left(-\frac{\pi \Delta^2}{4\sigma^2}\right)\right)} -  \left(\frac{2}{\Delta}+  \frac{1}{\pi\freq_c\Delta^2}\right)\ln\left(1 - \exp\!\left(-\frac{\pi \Delta^2}{4\sigma^2}\right)\right)\right], \notag
\end{align}
where in \eqref{eq: upper bound proof U1 invertible} we used the fact that for $n \geq 1$, $1/n^2 \leq 1/n$ and $\autocorrelation(n\Delta) = \exp\!\left(-\frac{\pi n^2\Delta^2}{4\sigma^2}\right) \leq \exp\!\left(-\frac{\pi n\Delta^2}{4\sigma^2}\right)$.
Based on the upper bound \eqref{eq: upper bound proof U1 invertible} we can now conclude that
\begin{equation*}
	\normOp{\mathcal{U}_1} \leq \frac{\xi\!\left(\frac{\Delta}{2\sigma}\right)}{2\freq_c\sigma^2} + \varphi\!\left(\frac{\Delta}{2\sigma^2}\right)\left(\frac{2}{\Delta}+  \frac{1}{\pi\freq_c\Delta^2}\right).
\end{equation*}
Setting
\begin{equation}
	\forall x > 0, \quad \rho(x) \triangleq x^2\xi(x) = \frac{x^2\exp\!\left(-\pi x^2\right)}{1 - \exp\!\left(-\pi x^2\right)},
	\label{eq: equation rho}
\end{equation}
we can rewrite \eqref{eq: equation rho} as
\begin{equation*}
	\normOp{\mathcal{U}_1} \leq \frac{2\rho\left(\frac{\Delta}{2\sigma}\right)}{\freq_c\Delta^2} + \varphi\!\left(\frac{\Delta}{2\sigma^2}\right)\left(\frac{2}{\Delta}+  \frac{1}{\pi\freq_c\Delta^2}\right).
\end{equation*}
We can verify that $\rho$ is non-increasing on $(0, \infty)$, which finally yields
\begin{equation}
	\normOp{\mathcal{U}_1} \leq \frac{2\rho\left(2\right)}{\freq_c\Delta^2} + \varphi\!\left(2\right)\left(\frac{2}{\Delta}+  \frac{1}{\pi\freq_c\Delta^2}\right) \leq \frac{2\rho(2) + (2 + 1/\pi)\varphi(2)}{4\sigma}.
	\label{eq: operator norm U1}
\end{equation}
It remains to upper-bound $\normOp{\mathcal{V}_0}$. To this end, we fix $\beta \in \ell^\infty(\measSupportIdx)$ and define $\zeta = \mathcal{V}_0\beta$. As $v(0) = 0$ and
\begin{equation*}
	\forall t \in \R \!\setminus \!\{0\}, \quad \abs{v(t)} \leq \frac{\autocorrelation(t)}{4\sigma^2\freq_c},
\end{equation*}
we get
\begin{align*}
	\abs{\zeta_\measSpikeIdx} &\leq \sum_{\measSpikeIdxBis \in \measSupportIdx} \abs{\beta_\measSpikeIdxBis} \abs{v(t_\measSpikeIdx - t_\measSpikeIdxBis)} =  \sum_{\measSpikeIdxBis \in \measSupportIdx \setminus \{\measSpikeIdx\}} \abs{\beta_\measSpikeIdxBis} \abs{v(t_\measSpikeIdx - t_\measSpikeIdxBis)} \\
		&\leq \normlinfty{\beta}\sum_{\measSpikeIdxBis \in \measSupportIdx \setminus \{\measSpikeIdx\}} \frac{\autocorrelation(t_\measSpikeIdx - t_\measSpikeIdxBis)}{4\sigma^2\freq_c} \leq \frac{\normlinfty{\beta}}{2\sigma^2\freq_c}\sum_{n = 1}^\infty \autocorrelation(n\Delta) \\
		&\leq \frac{\normlinfty{\beta}}{2\freq_c\sigma^2}\sum_{n = 1}^\infty \exp\!\left(-\frac{\pi n\Delta^2}{4\sigma^2}\right) =  \frac{\normlinfty{\beta}}{2\freq_c\sigma^2} \frac{\exp\!\left(-\frac{\pi\Delta^2}{4\sigma^2}\right)}{1 - \exp\!\left(-\frac{\pi\Delta^2}{4\sigma^2}\right)},
\end{align*}
for all $\measSpikeIdx \in \measSupportIdx$.
This yields
\begin{equation}
	\normOp{\mathcal{V}_0} \leq \frac{\xi\!\left(\frac{\Delta}{2\sigma}\right)}{2\freq_c\sigma^2} = \frac{\rho\!\left(\frac{\Delta}{2\sigma}\right)}{2\sigma} \leq \frac{\rho(2)}{2\sigma},
	\label{eq: operator norm V0}
\end{equation}
where the last inequality follows from $\Delta > 1/\freq_c$, $\Delta > 4\sigma$, and the fact that $\rho$ is non-increasing on $(0, \infty)$.
Finally, using \eqref{eq: bound operator norm inverse}, \eqref{eq: operator norm I - U0}, $\eqref{eq: operator norm U1}$, and \eqref{eq: operator norm V0} in \eqref{eq: norm operator triangle inequality}, we obtain
\begin{equation*}
	\normOp{\mathcal{I} - \mathcal{W}} \leq \frac{\varphi(2)}{\pi} + \frac{\rho(2)\left[\rho(2) + (1 + 1/(2\pi))\varphi(2)\right]}{2\pi - 4(\varphi(2) + \pi\psi(2) + \pi\xi(2))} \leq 1.12 \cdot 10^{-6} < 1,
\end{equation*}
where $\mathcal{W} \triangleq \mathcal{U}_0 - \mathcal{V}_0\mathcal{V}_1^{-1}\mathcal{U}_1$.
Again, applying the Neumann expansion theorem, we can conclude that the operator $\mathcal{W} = \mathcal{U}_0 - \mathcal{V}_0\mathcal{V}_1^{-1}\mathcal{U}_1$ is invertible and that its inverse satisfies
\begin{equation}
	\normOp{\mathcal{W}^{-1}} \leq \frac{1}{1 - \normOp{\mathcal{I} - \mathcal{W}}} \leq \left(1 - \frac{\varphi(2)}{\pi} - \frac{\rho(2)\left[\rho(2) + (1 + 1/(2\pi))\varphi(2)\right]}{2\pi - 4(\varphi(2) + \pi\psi(2) + \pi\xi(2))}\right)^{-1}.
	\label{eq: bound inverse W}
\end{equation}
For later use, we record that for the choices $\alpha = (\mathcal{U}_0 - \mathcal{V}_0\mathcal{V}_1^{-1}\mathcal{U}_1)^{-1}\varepsilon$ and $\beta = -\mathcal{V}_1^{-1}\mathcal{U}_1\alpha$, we have
\begin{align}
	\normlinfty{\alpha} &\leq \normOp{\mathcal{W}^{-1}}\normlinfty{\varepsilon}  = \normOp{\mathcal{W}^{-1}} \notag \\
		&\leq \left(1 - \frac{\varphi(2)}{\pi} - \frac{\rho(2)\left[\rho(2) + (1 + 1/(2\pi))\varphi(2)\right]}{2\pi - 4(\varphi(2) + \pi\psi(2) + \pi\xi(2))}\right)^{-1} \leq 1.01 \label{eq: upper bound norm alpha}
\end{align}
and
\begin{align}
	\normlinfty{\beta} &\leq \normOp{\mathcal{V}_1^{-1}}\normOp{\mathcal{U}_1}\normlinfty{\alpha} \leq \frac{2\sigma^2}{\pi - 2(\varphi(2) + \pi\psi(2) + \pi\xi(2))} \cdot \frac{2\rho(2) + (2 + 1/\pi)\varphi(2)}{4\sigma} \notag \\
		&\hspace{1cm} \cdot \left(1 - \frac{\varphi(2)}{\pi} - \frac{\rho(2)\left[\rho(2) + (1 + 1/(2\pi))\varphi(2)\right]}{2\pi - 4(\varphi(2) + \pi\psi(2) + \pi\xi(2))}\right)^{-1} \leq 5.73 \cdot 10^{-6}\sigma.  \label{eq: upper bound norm beta} 
\end{align}
In the remainder of the proof, we exclusively consider $c_0$ with $\alpha = (\mathcal{U}_0 - \mathcal{V}_0\mathcal{V}_1^{-1}\mathcal{U}_1)^{-1}\varepsilon$ and $\beta = -\mathcal{V}_1^{-1}\mathcal{U}_1\alpha$.

\subsection{$\abs{(\gaborOp^*c_0)(t)} < 1$ for all $t \in \R \!\setminus\! \measSupport$}
\label{subsection: show that the magnitude is strictly smaller than one outside the support}

\subsubsection{$\abs{(\gaborOp^*c_0)(t)} < 1$ for all $t \in \R \!\setminus\! \bigcup_{\measSpikeIdx \in \measSupport} \left[t_\measSpikeIdx - \frac{1}{7\freq_c}, t_\measSpikeIdx + \frac{1}{7\freq_c}\right]$}
Take $\measSpikeIdx_0 \in \measSupportIdx$ and let $\measSpikeIdx_1 \in \measSupportIdx$ be the index of the point in $\measSupport$ that is closest to $t_{\measSpikeIdx_0}$ and satisfies $t_{\measSpikeIdx_1} > t_{\measSpikeIdx_0}$. Take $t \in \left(t_{\measSpikeIdx_0} + \frac{1}{7\freq_c}, t_{\measSpikeIdx_1} - \frac{1}{7\freq_c}\right)$ and note that the interval $\left(t_{\measSpikeIdx_0} + \frac{1}{7\freq_c}, t_{\measSpikeIdx_1} - \frac{1}{7\freq_c}\right)$ is non-empty because $\abs{t_{\measSpikeIdx_0} - t_{\measSpikeIdx_1}} \geq \Delta > \frac{1}{\freq_c} > \frac{2}{7\freq_c}$. Without loss of generality, we assume that $\abs{t - t_{\measSpikeIdx_0}} \leq \abs{t - t_{\measSpikeIdx_1}}$, which implies $\abs{t - t_{\measSpikeIdx_1}} \geq \abs{t_{\measSpikeIdx_1} - t_{\measSpikeIdx_0}}/2 \geq \Delta/2$.
We set $h \triangleq \abs{t - t_{\measSpikeIdx_0}}$ and note that $h \geq \frac{1}{7\freq_c}$. The following holds
\begin{align}
	\abs{(\gaborOp^*c_0)(t)} &\leq \sum_{\measSpikeIdx \in \measSupportIdx} \Big(\normlinfty{\alpha} \abs{u(t - t_\measSpikeIdx)} + \normlinfty{\beta} \abs{v(t - t_\measSpikeIdx)}\Big) \notag \\
		&\leq \sum_{\measSpikeIdx \in \measSupportIdx} \left(\normlinfty{\alpha} \frac{\autocorrelation(t - t_\measSpikeIdx)}{2\pi \freq_c \abs{t - t_\measSpikeIdx}} + \normlinfty{\beta} \frac{\autocorrelation(t - t_\measSpikeIdx)}{4\sigma^2 \freq_c} \right)\notag \\
		&\leq \normlinfty{\alpha} \frac{\autocorrelation(h)}{2\pi\freq_c h} + \normlinfty{\beta} \frac{\autocorrelation(h)}{4\sigma^2\freq_c} + \normlinfty{\alpha} \frac{\autocorrelation\left(\frac{\Delta}{2}\right)}{\pi\freq_c\Delta} + \normlinfty{\beta}\frac{\autocorrelation\left(\frac{\Delta}{2}\right)}{4\sigma^2\freq_c} \notag \\
		&\qquad\quad + \sum_{\measSpikeIdx \in \measSupportIdx \setminus \{\measSpikeIdx_0, \measSpikeIdx_1\}} \left(\normlinfty{\alpha} \frac{\autocorrelation(t - t_\measSpikeIdx)}{2\pi \freq_c \abs{t - t_\measSpikeIdx}} + \normlinfty{\beta} \frac{\autocorrelation(t - t_\measSpikeIdx)}{4\sigma^2 \freq_c} \right) \notag \\
		&\leq \normlinfty{\alpha} \frac{\autocorrelation(h)}{2\pi\freq_c h} + \normlinfty{\beta} \frac{\autocorrelation(h)}{4\sigma^2\freq_c} + \normlinfty{\alpha} \frac{\autocorrelation\left(\frac{\Delta}{2}\right)}{\pi\freq_c\Delta} + \normlinfty{\beta}\frac{\autocorrelation\left(\frac{\Delta}{2}\right)}{4\sigma^2\freq_c} \notag \\
		&\qquad\quad + 2\sum_{n =1}^\infty \left(\normlinfty{\alpha} \frac{\autocorrelation(n\Delta)}{2\pi \freq_c n\Delta} + \normlinfty{\beta} \frac{\autocorrelation(n\Delta)}{4\sigma^2 \freq_c}\right) \notag \\
		&\leq \normlinfty{\alpha}\frac{\autocorrelation\left(\frac{1}{7\freq_c}\right)}{2\pi/7} + \normlinfty{\beta}\frac{\autocorrelation\left(\frac{1}{7\freq_c}\right)}{4\sigma^2\freq_c} + \normlinfty{\alpha} \frac{\autocorrelation\left(\frac{\Delta}{2}\right)}{\pi\freq_c\Delta} + \normlinfty{\beta}\frac{\autocorrelation\left(\frac{\Delta}{2}\right)}{4\sigma^2\freq_c}  \label{eq: proof dual smaller than 1 far from the support 0} \\
		&\qquad\quad + \frac{\normlinfty{\alpha}}{\pi\freq_c\Delta}\varphi\!\left(\frac{\Delta}{2\sigma}\right) + \frac{\normlinfty{\beta}}{2\sigma^2\freq_c}\xi\!\left(\frac{\Delta}{2\sigma}\right) \label{eq: proof dual smaller than 1 far from the support 1} \\
		\hspace{1.9cm}&\leq \normlinfty{\alpha} \frac{7 \exp\!\left(-\frac{4\pi}{49}\right) }{2\pi}+ \normlinfty{\beta}\frac{\exp\!\left(-\frac{4\pi}{49}\right)}{\sigma} + \normlinfty{\alpha}\frac{\exp(-\pi)}{\pi} \label{eq: proof dual smaller than 1 far from the support 2}\\
		&\qquad\quad+ \normlinfty{\beta} \frac{\exp(-\pi)}{\sigma} + \normlinfty{\alpha}\frac{\varphi(2)}{\pi} + \normlinfty{\beta}\frac{2\xi(2)}{\sigma} \leq 0.876 < 1,  \label{eq: proof dual smaller than 1 far from the support 3}
\end{align}
where \eqref{eq: proof dual smaller than 1 far from the support 0} and \eqref{eq: proof dual smaller than 1 far from the support 1} follow from $h \geq \frac{1}{7\freq_c}$, and \eqref{eq: proof dual smaller than 1 far from the support 2} and \eqref{eq: proof dual smaller than 1 far from the support 3} can be derived invoking the assumptions $\Delta > 1/\freq_c$ and $\sigma = \frac{1}{4\freq_c}$.

\subsection{$\abs{\gaborOp^*c_0}$ is concave on $\bigcup_{\measSpikeIdx \in \measSupportIdx} \left[t_\measSpikeIdx - \frac{1}{7\freq_c}, t_\measSpikeIdx + \frac{1}{7\freq_c}\right]$}
Let $\measSpikeIdx \in \measSupportIdx$. We show that $t \mapsto A(t) \triangleq \abs{(\gaborOp^*c_0)(t)}$ is strictly concave on $\left[t_\measSpikeIdx - \frac{1}{7\freq_c}, t_\measSpikeIdx + \frac{1}{7\freq_c}\right]$. Since $\abs{\autocorrelation}$, $\abs{\sinc}$, $\abs{\autocorrelation'}$, and $\abs{\sinc'}$ are all symmetric, $A$ is symmetric as well, and therefore, it suffices to show that $A''(t) < 0$ for $t \in \left[t_\measSpikeIdx, t_\measSpikeIdx + \frac{1}{7\freq_c}\right]$. Since $\abs{\varepsilon_\measSpikeIdx} = 1$, we can write
\begin{equation*}
	\forall t \in \R, \quad A(t) = \abs{(\gaborOp^*c_0)(t)} = \abs{\frac{(\gaborOp^*c_0)(t)}{\varepsilon_\measSpikeIdx}} = \sqrt{(A_R(t))^2 + (A_I(t))^2}
\end{equation*}
where $A_R(t) \triangleq \Re{\frac{(\gaborOp^*c_0)(t)}{\varepsilon_\measSpikeIdx}}$, $A_I(t) \triangleq \Im{\frac{(\gaborOp^*c_0)(t)}{\varepsilon_\measSpikeIdx}}$, for $t \in \R$.
With $\Lambda \triangleq \{t \in \R \colon $\linebreak[4]$A(t) \neq 0\}$ we have
\begin{equation*}
	\forall t \in \Lambda, \quad A''(t) = \frac{A_R''(t)A_R(t) + A_I''(t)A_I(t) + \abs{(\gaborOp^*c_0)'(t)}^2}{A(t)} - \frac{(A_R'(t)A_R(t) + A_I'(t)A_I(t))^2}{(A(t))^3}.
\end{equation*}
For $A$ to be concave on $\left[t_\measSpikeIdx, t_\measSpikeIdx + \frac{1}{7\freq_c}\right]$, it therefore suffices to show that
\begin{equation*}
	\forall t \in \left[t_\measSpikeIdx, t_\measSpikeIdx + \frac{1}{7\freq_c}\right], \quad A_R''(t)A_R(t) + A_I''(t)A_I(t) + \abs{(\gaborOp^*c_0)'(t)}^2 < 0.
\end{equation*} 
Let $t \in \left[t_\measSpikeIdx, t_\measSpikeIdx + \frac{1}{7\freq_c}\right]$. We have the following
\begin{align*}
	A_R(t) &= \sum_{\measSpikeIdxBis \in \measSupportIdx} \left(\Re{\frac{\alpha_\measSpikeIdxBis}{\varepsilon_\measSpikeIdx}} u(t - t_\measSpikeIdxBis) +  \Re{\frac{\beta_\measSpikeIdxBis}{\varepsilon_\measSpikeIdx}}v(t - t_\measSpikeIdxBis)\right) \\
		&= \Re{\frac{\alpha_\measSpikeIdx}{\varepsilon_\measSpikeIdx}}u(t - t_\measSpikeIdx) + \Re{\frac{\beta_\measSpikeIdx}{\varepsilon_\measSpikeIdx}}v(t - t_\measSpikeIdx) + \sum_{\measSpikeIdxBis \in \measSupportIdx \setminus \{\measSpikeIdx\}} \left(\Re{\frac{\alpha_\measSpikeIdxBis}{\varepsilon_\measSpikeIdx}} u(t - t_\measSpikeIdxBis) \right. \\
		&\hspace{3.5cm} \left. +  \Re{\frac{\beta_\measSpikeIdxBis}{\varepsilon_\measSpikeIdx}}v(t - t_\measSpikeIdxBis) \right) \\
		\hspace{0.2cm}&\geq  \Re{\frac{\alpha_\measSpikeIdx}{\varepsilon_\measSpikeIdx}}u(t - t_\measSpikeIdx) - \normlinfty{\beta}\abs{v(t - t_\measSpikeIdx)} - \sum_{\measSpikeIdxBis \in \measSupportIdx \setminus \{\measSpikeIdx\}} \Big(\normlinfty{\alpha} \abs{u(t - t_\measSpikeIdxBis)} \\
		&\hspace{3.5cm}+  \normlinfty{\beta} \abs{v(t - t_\measSpikeIdxBis)} \Big).
\end{align*}
With $\alpha = \mathcal{W}^{-1}\varepsilon = \varepsilon - (\mathcal{I} - \mathcal{W}^{-1})\varepsilon$, it follows that
\begin{align}
	\Re{\frac{\alpha_\measSpikeIdx}{\varepsilon_\measSpikeIdx}} &= 1 - \Re{\frac{\left[(\mathcal{I} - \mathcal{W}^{-1})\varepsilon\right]_\measSpikeIdx}{\varepsilon_\measSpikeIdx}} \geq 1 - \abs{\frac{\left[(\mathcal{I} - \mathcal{W}^{-1})\varepsilon\right]_\measSpikeIdx}{\varepsilon_\measSpikeIdx}} \notag \\
		&\geq 1 - \normOp{\mathcal{I} - \mathcal{W}^{-1}} \notag \\
		&= 1 -  \normOp{\mathcal{W}^{-1}(\mathcal{I} - \mathcal{W})} \geq 1 - \normOp{\mathcal{W}^{-1}}\normOp{\mathcal{I} - \mathcal{W}} \geq 0.999998, \label{eq: bound Re 2}
\end{align}
where \eqref{eq: bound Re 2} is due to \eqref{eq: bound inverse W}.
Next, it follows from $t - t_\measSpikeIdx \leq \frac{1}{7\freq_c}$ and $\sigma = \frac{1}{4\freq_c}$ that
\begin{align}
	u(t - t_\measSpikeIdx) = \autocorrelation(t - t_\measSpikeIdx)\sinc(2\pi\freq_c(t - t_\measSpikeIdx)) &= \exp\!\left(-\frac{\pi(t - t_\measSpikeIdx)^2}{4\sigma^2}\right)\sinc(2\pi\freq_c(t - t_\measSpikeIdx)) \notag \\
		&\geq \exp\!\left(-\frac{\pi}{49\freq_c^2 \cdot 4\sigma^2}\right)\sinc\!\left(\frac{2\pi}{7}\right) \notag \\
		&= \exp\!\left(-\frac{4\pi}{49}\right)\sinc\!\left(\frac{2\pi}{7}\right). \label{eq: lower bound u proof concave}
\end{align}
Since $\abs{\sinc} \leq 1$, we have
\begin{equation*}
	\normlinfty{\beta}\abs{v(t - t_\measSpikeIdx)} \leq \normlinfty{\beta}|\autocorrelation'(t_\measSpikeIdx - t)|\abs{\sinc(2\pi\freq_c(t - t_\measSpikeIdx))} \leq \normlinfty{\beta}\abs{\autocorrelation'(t - t_\measSpikeIdx)}.
\end{equation*}
As $\abs{\autocorrelation'}$ has its maxima at the points $\pm \sigma\sqrt{\frac{2}{\pi}}$ with corresponding maximum values \linebreak[4]$\frac{1}{\sigma}\exp\!\left(-\frac{1}{2}\right)\sqrt{\frac{\pi}{2}}$, we get
\begin{equation}
	\normlinfty{\beta}\abs{v(t - t_\measSpikeIdx)} \leq \frac{\normlinfty{\beta}}{\sigma}\exp\!\left(-\frac{1}{2}\right)\sqrt{\frac{\pi}{2}}. \label{eq: upper bound beta v proof concave}
\end{equation}
As for every $\measSpikeIdxBis \in \measSupportIdx \setminus \{\measSpikeIdx\}$, we have $\abs{t - t_\measSpikeIdxBis} \geq \abs{t_\measSpikeIdx - t_\measSpikeIdxBis} - \abs{t - t_\measSpikeIdx} \geq \Delta - \frac{1}{7\freq_c} > \frac{6}{7\freq_c}$, it holds that
\begin{align}
	\sum_{\measSpikeIdxBis \in \measSupportIdx \setminus \{\measSpikeIdx\}} &\Big(\normlinfty{\alpha}\abs{u(t - t_\measSpikeIdxBis)} + \normlinfty{\beta}\abs{v(t - t_\measSpikeIdxBis)}\Big) \leq \normlinfty{\alpha}\frac{7\autocorrelation\!\left(\frac{6}{7\freq_c}\right)}{12\pi} + \normlinfty{\beta}\frac{\autocorrelation\!\left(\frac{6}{7\freq_c}\right)}{4\sigma^2\freq_c} \notag \\
		&\qquad\quad + 2\sum_{n = 1}^\infty \left(\normlinfty{\alpha} \frac{\autocorrelation(n\Delta)}{2\pi\freq_c n\Delta} + \normlinfty{\beta} \frac{\autocorrelation(n\Delta)}{4\sigma^2\freq_c}\right) \notag \\
		&\hspace{7cm}\leq \normlinfty{\alpha} \frac{7\exp\!\left(-\pi\left(\frac{12}{7}\right)^2\right)}{12\pi} \notag\\
			&\quad\quad\ + \frac{\normlinfty{\beta}}{\sigma}\exp\!\left(-\pi\left(\frac{12}{7}\right)^2\right) + \normlinfty{\alpha}\frac{\varphi(2)}{\pi} + \frac{\normlinfty{\beta}}{\sigma} 2\xi(2). \label{eq: upper bound summation proof concave}
\end{align}
Combining \eqref{eq: upper bound norm alpha}, \eqref{eq: upper bound norm beta}, \eqref{eq: bound Re 2}, \eqref{eq: lower bound u proof concave}, \eqref{eq: upper bound beta v proof concave}, and \eqref{eq: upper bound summation proof concave} yields
\begin{equation}
	A_R(t) \geq 0.673.
	\label{eq: upper bound AR}
\end{equation}
Next, we derive an upper bound on $A_R''(t)$:
\begin{align*}
	A_R''(t) &= \sum_{\measSpikeIdxBis \in \measSupportIdx} \left(\Re{\frac{\alpha_\measSpikeIdxBis}{\varepsilon_\measSpikeIdxBis}}u''(t - t_\measSpikeIdxBis) + \Re{\frac{\beta_\measSpikeIdxBis}{\varepsilon_\measSpikeIdxBis}}v''(t - t_\measSpikeIdxBis)\right) \\
		&\leq \Re{\frac{\alpha_\measSpikeIdx}{\varepsilon_\measSpikeIdx}}u''(t - t_\measSpikeIdx) + \normlinfty{\beta}\abs{v''(t - t_\measSpikeIdx)} \\
		&\hspace{3.5cm}+ \sum_{\measSpikeIdxBis \in \measSupportIdx \setminus \{\measSpikeIdx\}} \Big(\normlinfty{\alpha}\abs{u''(t - t_\measSpikeIdxBis)} + \normlinfty{\beta} \abs{v''(t - t_\measSpikeIdxBis)}\Big).
\end{align*}
For all $t \in \R$, we have
\begin{equation}
	u''(t) = \autocorrelation''(t)\sinc(2\pi\freq_c t) + 4\pi\freq_c\autocorrelation'(t)\sinc'(2\pi\freq_c t) + (2\pi\freq_c)^2\autocorrelation(t)\sinc''(2\pi\freq_c t). 
	\label{eq: equality u''}
\end{equation}
The function $t \mapsto \autocorrelation''(t)\sinc(2\pi\freq_c t)$ is non-decreasing on $\big[0, \frac{1}{7\freq_c}\big]$, since, on this interval, $\autocorrelation''$ is negative and non-decreasing and $t \mapsto \sinc(2\pi\freq_c t)$ is positive and non-increasing. The function $t \mapsto 4\pi\freq_c\autocorrelation'(t)\sinc'(2\pi\freq_c t)$ is non-decreasing on $\big[0, \frac{1}{7\freq_c}\big]$, as both $\autocorrelation'$ and $t \mapsto \sinc'(2\pi\freq_c t)$ are negative and non-increasing on this interval. The function $t \mapsto (2\pi\freq_c)^2\autocorrelation(t)\sinc''(2\pi\freq_c t)$ is non-decreasing on $\big[0, \frac{1}{7\freq_c}\big]$, as, on this interval, $\autocorrelation$ is positive and non-increasing, and $t \mapsto \sinc''(2\pi\freq_c t)$ is negative and non-decreasing. Taken together, it follows that $u''$ is non-decreasing on $\big[0, \frac{1}{7\freq_c}\big]$.
Since  $t - t_\measSpikeIdx \in \big[0, \frac{1}{7\freq_c}\big]$, we then have
\begin{align*}
	u''(t - t_\measSpikeIdx) &\leq u''\!\left(\frac{1}{7\freq_c}\right) = \freq_c^2\left[8\pi\left(\frac{8\pi}{49} - 1\right)\sinc\left(\frac{2\pi}{7}\right) - \frac{32\pi^2}{7}\sinc'\left(\frac{2\pi}{7}\right) \right.\\
		&\hspace{2cm} \left. + 4\pi^2\sinc''\left(\frac{2\pi}{7}\right)\right]\exp\!\left(-\frac{4\pi}{49}\right),
\end{align*}
where we used $\sigma = \frac{1}{4\freq_c}$.
Combined with \eqref{eq: bound Re 2}, this yields
\begin{equation*}
	\Re{\frac{\alpha_\measSpikeIdx}{\varepsilon_\measSpikeIdx}}u''(t - t_\measSpikeIdx) \leq -6.46 \freq_c^2.
\end{equation*}
Since 
\begin{equation*}
	\forall x \in \R \setminus \{0\}, \quad \sinc'(x) = \frac{\cos(x)}{x} - \frac{\sin(x)}{x^2} \quad \text{ and }\qquad \sinc''(x) = -\frac{\sin(x)}{x} - \frac{2\cos(x)}{x^2} + \frac{2\sin(x)}{x^3},
\end{equation*}
we get the following from \eqref{eq: equality u''}:
\begin{align*}
	\abs{u''(t)} &\leq \left(\frac{\pi}{2\sigma^2} + \frac{\pi^2\abs{t}^2}{4\sigma^4}\right)\frac{1}{2\pi\freq_c \abs{t}}\autocorrelation(t)+ 4\pi\freq_c\frac{\pi\abs{t}}{2\sigma^2}\left(\frac{1}{2\pi\freq_c\abs{t}} + \frac{1}{(2\pi\freq_c)^2\abs{t}^2}\right)\autocorrelation(t) \\
		&\quad + (2\pi\freq_c)^2 \left(\frac{1}{2\pi\freq_c\abs{t}} + \frac{2}{(2\pi\freq_c)^2\abs{t}^2} + \frac{2}{(2\pi\freq_c)^3\abs{t}^3)} \right)\autocorrelation(t) \\
		&= \left[\frac{\pi\abs{t}}{8\freq_c\sigma^4} + \frac{\pi}{\sigma^2} + \left(2\pi\freq_c + \frac{3}{4\sigma^2\freq_c}\right)\frac{1}{\abs{t}} + \frac{2}{\abs{t}^2} + \frac{1}{\pi\freq_c\abs{t}^3}\right]\autocorrelation(t).
\end{align*}
As a result, we have the following chain of inequalities
\begin{align*}
	&\sum_{\measSpikeIdxBis \in \measSupportIdx \setminus \{\measSpikeIdx\}} \abs{u''(t - t_\measSpikeIdxBis)} \leq 2\sum_{n = 1}^\infty \left[\pi\frac{n\Delta}{8\freq_c\sigma^4} + \frac{\pi}{\sigma^2} + \left(2\pi\freq_c + \frac{3}{4\sigma^2\freq_c}\right)\frac{1}{n\Delta} + \frac{2}{n^2\Delta^2} + \frac{1}{\pi\freq_c n^3\Delta^3}\right]\autocorrelation(n\Delta) \\
		&\!\!\leq \frac{\pi\Delta}{4\sigma^4\freq_c}\sum_{n = 1}^\infty n\exp\!\left(-\frac{\pi n\Delta^2}{4\sigma^2}\right) + \frac{2\pi}{\sigma^2}\sum_{n = 1}^\infty \exp\!\left(-\frac{\pi n\Delta^2}{4\sigma^2}\right) \\
		&\hspace{3cm} + \left(\frac{4\pi\freq_c}{\Delta} + \frac{3}{2\sigma^2\freq_c\Delta}\right)\sum_{n = 1}^\infty \frac{1}{n}\exp\!\left(-\frac{\pi n\Delta^2}{4\sigma^2}\right) + \frac{4}{\Delta^2}\sum_{n = 1}^\infty \frac{1}{n}\exp\!\left(-\frac{\pi n\Delta^2}{4\sigma^2}\right) \\
		&\hspace{3cm}+ \frac{2}{\pi\freq_c\Delta^3}\sum_{n = 1}^\infty \frac{1}{n}\exp\!\left(-\frac{\pi n\Delta^2}{4\sigma^2}\right) \\
		&\!\!\leq \frac{\pi\Delta}{4\sigma^4\freq_c}\psi\!\left(\frac{\Delta}{2\sigma}\!\right) + \frac{2\pi}{\sigma^2}\xi\!\left(\frac{\Delta}{2\sigma}\!\right) + \left(\!\frac{4\pi\freq_c}{\Delta} + \frac{3}{2\sigma^2\freq_c\Delta}\!\right)\varphi\!\left(\frac{\Delta}{2\sigma}\!\right) + \frac{4}{\Delta^2}\varphi\!\left(\frac{\Delta}{2\sigma}\!\right) + \frac{2}{\pi\freq_c\Delta^3}\varphi\!\left(\frac{\Delta}{2\sigma}\!\right). 
\end{align*}
We can now define $\eta(x) \triangleq x^2\psi(x)$, for all $x > 0$, which can be shown to be non-increasing on $(0, \infty)$. This yields
\begin{align*}
		\sum_{\measSpikeIdxBis \in \measSupportIdx \setminus \{\measSpikeIdx\}} \abs{u''(t - t_\measSpikeIdxBis)} &\leq \frac{\pi}{\Delta\sigma^2\freq_c}\eta\!\left(\frac{\Delta}{2\sigma}\!\right) + \frac{2\pi}{\sigma^2}\xi\!\left(\frac{\Delta}{2\sigma}\!\right) + \left(\!\frac{4\pi\freq_c}{\Delta} + \frac{3}{2\sigma^2\freq_c\Delta}\!\right)\varphi\!\left(\frac{\Delta}{2\sigma}\!\right) \\
		&\hspace{2cm} + \frac{4}{\Delta^2}\varphi\!\left(\frac{\Delta}{2\sigma}\!\right) + \frac{2}{\pi\freq_c\Delta^3}\varphi\!\left(\frac{\Delta}{2\sigma}\!\right) \\
		&\leq \left[16\pi\eta(2) + 32\pi\xi(2) + \left(4\pi + 28 + \frac{2}{\pi}\right)\varphi(2)\right]\freq_c^2. \\[-0.4cm]
\end{align*}
Combined with \eqref{eq: upper bound norm alpha}, we get 
\begin{equation*}
	\normlinfty{\alpha} \sum_{\measSpikeIdxBis \in \measSupportIdx \setminus \{\measSpikeIdx\}} \abs{u''(t - t_\measSpikeIdxBis)} \leq 1.196 \cdot 10^{-3}\freq_c^2. \\[0.25cm]
\end{equation*}
We have, for all $t \in \R$, that
\begin{equation*}
	v''(t) = \autocorrelation'''(-t)\sinc(2\pi\freq_c t) - 4\pi\freq_c\autocorrelation''(-t)\sinc'(2\pi\freq_c t) + (2\pi\freq_c)^2\autocorrelation'(-t)\sinc''(2\pi\freq_c t). \\[0.25cm]
\end{equation*}
Therefore, we get
\begin{align*}
	\abs{v''(t)} &\leq \left(\frac{3\pi^2}{4\sigma^4}\abs{t} + \frac{\pi^3}{8\sigma^6}\abs{t}^3\right)\frac{1}{2\pi\freq_c\abs{t}}\autocorrelation\left(t\right)  \\
	&\hspace{2cm} + 4\pi\freq_c\left(\frac{\pi^2\abs{t}^2}{4\sigma^4} + \frac{\pi}{2\sigma^2}\right)\left(\frac{1}{2\pi\freq_c\abs{t}} + \frac{1}{(2\pi\freq_c)^2\abs{t}^2}\right)\autocorrelation(t) \\
	&\hspace{2cm} + (2\pi\freq_c)^2\frac{\pi}{2\sigma^2}\abs{t}\left(\frac{1}{2\pi\freq_c\abs{t}} + \frac{2}{(2\pi\freq_c)^2\abs{t}^2} + \frac{2}{(2\pi\freq_c)^3\abs{t}^3}\right) \\
	&\leq \left(\frac{\pi^2\abs{t}^2}{16\sigma^6\freq_c} + \frac{\pi^2\abs{t}}{2\sigma^4} + \frac{5\pi}{8\freq_c\sigma^4} + \frac{\pi^2\freq_c}{\sigma^2} + \frac{2\pi}{\sigma^2\abs{t}} + \frac{1}{\freq_c\sigma^2\abs{t}^2}\right)\autocorrelation(t), \\[-0.4cm]
\end{align*}
which leads to the following chain of inequalities
\begin{align}
	\sum_{\measSpikeIdxBis \in \measSupportIdx \setminus \{\measSpikeIdx\}} \!\!\!\abs{v''(t_\measSpikeIdx - t_\measSpikeIdxBis)} &\leq 2\sum_{n = 1}^\infty \!\left(\frac{\pi^2n^2\Delta^2}{16\sigma^6\freq_c} + \frac{\pi^2 n\Delta}{2\sigma^4} + \frac{5\pi}{8\freq_c\sigma^4} + \frac{\pi^2\freq_c}{\sigma^2} + \frac{2\pi}{\sigma^2 n\Delta} + \frac{1}{\freq_c\sigma^2 n^2\Delta^2}\right) \!\autocorrelation(n\Delta) \notag \\
	&\leq \frac{\pi^2\Delta^2}{8\sigma^6\freq_c}\sum_{n = 1}^\infty n^2\exp\!\left(-\frac{\pi n\Delta^2}{4\sigma^2}\right) + \frac{\pi^2\Delta}{\sigma^4}\sum_{n = 1}^\infty n\exp\!\left(-\frac{\pi n\Delta^2}{4\sigma^2}\right) \label{eq: last line ref first term} \\
	&\qquad\qquad + \frac{5\pi}{4\freq_c\sigma^4}\sum_{n = 1}^\infty \exp\!\left(-\frac{\pi n\Delta^2}{4\sigma^2}\right) + \frac{2\pi^2\freq_c}{\sigma^2}\sum_{n = 1}^\infty \exp\!\left(-\frac{\pi n\Delta^2}{4\sigma^2}\right) \notag \\
	&\qquad\qquad + \frac{4\pi}{\sigma^2\Delta}\sum_{n = 1}^\infty \frac{1}{n}\exp\!\left(-\frac{\pi n\Delta^2}{4\sigma^2}\right) + \frac{1}{\freq_c\sigma^2\Delta^2}\sum_{n = 1}^\infty \frac{1}{n}\exp\!\left(-\frac{\pi n\Delta^2}{4\sigma^2}\right). \notag
\end{align}
Now, in \eqref{eq: last line ref first term}, we recognize the power series
\begin{equation*}
	\forall x \in (-1, 1), \quad \sum_{n = 1}^\infty n^2 x^n = \frac{x(x+1)}{(1 - x)^3}
\end{equation*}
evaluated at $x = \exp\!\left(-\frac{\pi\Delta^2}{4\sigma^2}\right)$, which leads us to set
\begin{equation*}
	\forall x > 0, \quad \Gamma(x) \triangleq \frac{\exp(-\pi x^2)(\exp(-\pi x^2) + 1)}{(1 - \exp(-\pi x^2))^3}.
\end{equation*} 
This yields 
\begin{align*}
	\sum_{\measSpikeIdxBis \in \measSupportIdx \setminus \{\measSpikeIdx\}} \abs{v''(t_\measSpikeIdx - t_\measSpikeIdxBis)} &\leq \frac{\pi^2\Delta^2}{8\sigma^6\freq_c}\Gamma\!\left(\frac{\Delta}{2\sigma}\right) + \frac{\pi^2\Delta}{\sigma^4}\psi\!\left(\frac{\Delta}{2\sigma}\right) + \left(\frac{5\pi}{4\freq_c\sigma^4} + \frac{2\pi^2\freq_c}{\sigma^2}\right)\xi\!\left(\frac{\Delta}{2\sigma}\right) \\
	&\qquad\qquad + \left(\frac{4\pi}{\sigma^2\Delta} + \frac{1}{\freq_c\sigma^2\Delta^2}\right)\varphi\!\left(\frac{\Delta}{2\sigma}\right) \\
	&\leq \frac{\pi^2}{2\freq_c\sigma^4}\gamma\!\left(\frac{\Delta}{2\sigma}\right) + \frac{2\pi^2}{3\sigma^2\Delta}\eta\!\left(\frac{\Delta}{2\sigma}\right) + \left(\frac{5\pi}{4\freq_c\sigma^4} + \frac{2\pi^2\freq_c}{\sigma^2}\right)\xi\!\left(\frac{\Delta}{2\sigma}\right) \\
	&\qquad\qquad + \left(\frac{4\pi}{\sigma^2\Delta} + \frac{1}{\freq_c\sigma^2\Delta^2}\right)\varphi\!\left(\frac{\Delta}{2\sigma}\right) \\
	&\leq \Big[128\pi^2 \gamma(2) + 64\pi^2\eta(2) + (320\pi + 32\pi^2)\xi(2) + (64\pi+16)\varphi(2)\Big]\freq_c^3,
\end{align*}
where we set $\gamma(x) \triangleq x^2\Gamma(x)$, for $x > 0$, and used the fact that $\gamma$ is non-increasing on $(0, \infty)$. 
Combined with \eqref{eq: upper bound norm beta}, this results in
\begin{equation}
	\normlinfty{\beta}\sum_{\measSpikeIdxBis \in \measSupportIdx \setminus \{\measSpikeIdx\}} \abs{v''(t_\measSpikeIdx - t_\measSpikeIdxBis)} \leq 8.34 \cdot 10^{-8}\freq_c^2.
	\label{eq: multiplying this bound}
\end{equation}
Finally, we have $A_R''(t) \leq -22.1\freq_c^2$. Multiplying \eqref{eq: multiplying this bound} with \eqref{eq: upper bound AR} leads to $A_R(t)A_R''(t) \leq -14.6\freq_c^2$. Exactly the same line of reasoning can be applied to get $A_I(t)A_I''(t) \leq -14.6\freq_c^2$, and therefore,\begin{equation}
	A_R(t)A_R''(t) + A_I(t)A_I''(t) \leq -29.1\freq_c^2.
	\label{eq: intermediary step 1}
\end{equation} 
It remains to find an upper bound on $\abs{(\gaborOp^*c_0)'(t)}^2$. We have 
\begin{align}
	\abs{(\gaborOp^*c_0)'(t)} &\leq \sum_{\measSpikeIdxBis \in \measSupportIdx} \Big(\normlinfty{\alpha} \abs{u'(t - t_\measSpikeIdxBis)} + \normlinfty{\beta} \abs{v'(t - t_\measSpikeIdxBis)}\Big) \notag \\
		&\leq \normlinfty{\alpha} \abs{u'(t - t_\measSpikeIdx)} + \normlinfty{\beta} \abs{v'(t - t_\measSpikeIdx)} \label{eq: upper bound A' 0} \\
		&\hspace{2cm} + \sum_{\measSpikeIdxBis \in \measSupportIdx \setminus \{\measSpikeIdx\}} \Big(\normlinfty{\alpha} \abs{u'(t - t_\measSpikeIdxBis)} + \normlinfty{\beta} \abs{v'(t - t_\measSpikeIdxBis)} \Big). \label{eq: upper bound A'}
\end{align}
We can derive upper bounds for the terms in \eqref{eq: upper bound A' 0} by noting that
\begin{align}
	\abs{u'(t)} \leq u'\!\left(\frac{1}{7\freq_c}\right) &= \autocorrelation'\!\left(\frac{1}{7\freq_c}\right)\sinc\!\left(\frac{2\pi}{7}\right) + 2\pi\freq_c\autocorrelation\left(\frac{1}{7\freq_c}\right)\sinc'\left(\frac{2\pi}{7}\right) \label{eq: upper bound u' first} \\
		&= \exp\!\left(-\frac{4\pi}{49}\right)\sinc\!\left(\frac{2\pi}{7}\right) + 2\pi\freq_c \exp\!\left(-\frac{4\pi}{49}\right)\sinc'\left(\frac{2\pi}{7}\right) \notag
\end{align}
and
\begin{align}
	\abs{v'(t)} &\leq v'(0) = -\autocorrelation''(0) = \frac{\pi}{2\sigma^2}\label{eq: upper bound v' first}
\end{align}
for all $t \in \left[0, \frac{1}{7\freq_c}\right]$. Indeed, we have seen that $u''(t) \leq 0$ for all $t \in \left[0, \frac{1}{7\freq_c}\right]$, which implies that $u'$ is non-increasing on $\left[0, \frac{1}{7\freq_c}\right]$. As $u'(0) = 0$, this means that $u'$ is non-positive on $\left[0, \frac{1}{7\freq_c}\right]$. Therefore, $\abs{u'}$ is non-decreasing on $\left[0, \frac{1}{7\freq_c}\right]$, which results in~\eqref{eq: upper bound u' first}. The inequality in \eqref{eq: upper bound v' first} follows from the fact that $\abs{v'}$ is decreasing on $\left[0, \frac{1}{7\freq_c}\right]$, as we show next. We have 
\begin{align*}
	\forall t \in \R, \quad v'(t) &= -\autocorrelation''(-t)\sinc(2\pi\freq_c t) + 2\pi\freq_c \autocorrelation'(-t)\sinc'(2\pi\freq_c t) \\
		&= -\autocorrelation''(t)\sinc(2\pi\freq_c t) - 2\pi\freq_c \autocorrelation'(t)\sinc'(2\pi\freq_c t).
\end{align*}
As the functions $t \mapsto \autocorrelation''(t)\sinc(2\pi\freq_c t)$ and $t \mapsto 2\pi\freq_c\autocorrelation'(t)\sinc'(2\pi\freq_c t)$ were shown to both be non-decreasing on $\left[0, \frac{1}{7\freq_c}\right]$, we get that $v'$ is non-increasing on $\left[0, \frac{1}{7\freq_c}\right]$. Moreover, we have
\begin{equation*}
	v'\!\left(\frac{1}{7\freq_c}\right) \geq 3.43\freq_c^2 \geq 0.
\end{equation*}
Hence, $v'$ is non-negative on $\left[0, \frac{1}{7\freq_c}\right)$. This allows us to conclude that $\abs{v'}$ is non-increasing on $\left[0, \frac{1}{7\freq_c}\right]$, which establishes \eqref{eq: upper bound v' first}.
It remains to upper-bound the term in \eqref{eq: upper bound A'}, which is done as follows:
\begin{align}
	 \hspace{-0.1cm}\sum_{\measSpikeIdxBis \in \measSupportIdx \setminus \{\measSpikeIdx\}}& \Big(\normlinfty{\alpha} \abs{u'(t - t_\measSpikeIdxBis)} + \normlinfty{\beta} \abs{v'(t - t_\measSpikeIdxBis)} \Big) \notag \\
	 	&\leq  \sum_{\measSpikeIdxBis \in \measSupportIdx \setminus \{\measSpikeIdx\}} \Big(\normlinfty{\alpha} U(t - t_\measSpikeIdxBis) + \normlinfty{\beta} V(t - t_\measSpikeIdxBis) \Big) \notag \\
	 	&\leq \normlinfty{\alpha} U\!\left(\frac{6}{7\freq_C}\right) + \normlinfty{\beta}V\!\left(\frac{6}{7\freq_c}\right) + \normlinfty{\alpha} \left[\frac{2\rho(2) + (2+1/\pi)\varphi(2)}{\Delta}\right] \notag \\
		&\hspace{2cm} \normlinfty{\beta} \left[\frac{\pi}{2\sigma^2} \left(\frac{2}{\pi} \varphi(2) + 2\psi(2) + 2\xi(2)\right)\right] \hspace{4.5cm}\notag \\
		&\leq \normlinfty{\alpha} \left(\frac{1}{4\sigma^2\freq_c} + \frac{7\freq_c}{6} + \frac{49\freq_c}{72\pi}\right)\autocorrelation\!\left(\frac{6}{7\freq_c}\right) + \normlinfty{\beta}\left(\frac{7}{12\sigma^2} + \frac{\pi}{2\sigma^2} + \frac{3\pi}{28\sigma^4\freq_c^2}\right)\autocorrelation\!\left(\frac{6}{7\freq_c}\right) \notag \\
		&\hspace{2cm} + \normlinfty{\alpha} \big[2\rho(2) + (2+1/\pi)\varphi(2)\big]\freq_c + \frac{\normlinfty{\beta}}{\sigma^2} \big(\varphi(2) + \pi\psi(2) + \pi\xi(2)\big) \notag \\
		&\leq \normlinfty{\alpha}\left(4+\frac{7}{6}+\frac{49}{72\pi}\right)\freq_c\exp\!\left(-\frac{576\pi}{49}\right) + \frac{\normlinfty{\beta}}{\sigma}\left(\frac{7}{3} + 2\pi + \frac{48\pi}{7}\right)\freq_c\exp\!\left(-\frac{576\pi}{49}\right) \notag \\
		&\hspace{2cm} + \normlinfty{\alpha} \big[2\rho(2) + (2+1/\pi)\varphi(2)\big]\freq_c + \frac{\normlinfty{\beta}}{\sigma} \big(4\varphi(2) + 4\pi\psi(2) + 4\pi\xi(2)\big)\freq_c \notag \\
		&\leq 4.05\freq_c. \label{eq: intermediary step 2}
\end{align}
Putting \eqref{eq: intermediary step 1} and \eqref{eq: intermediary step 2} together yields
\begin{equation*}
	A_R(t)A_R''(t) + A_I(t)A_I''(t) + \abs{(\gaborOp^*c_0)'(t)}^2 \leq -12.68\freq_c^2 < 0,
\end{equation*}
which completes the proof.

%% file: proofMainResultTorus.tex

\section{Proof of Theorem~\ref{thm: exact recovery torus}}
\label{sec: proof torus theorem exact recovery}

We could prove Theorem~\ref{thm: exact recovery torus} following similar arguments as in the proof of Theorem~\ref{thm: exact recovery}, namely by choosing a function $c_0 \in L^\infty(\torus \times \Z)$ of the form
\begin{equation*}
	\forall \tau \in \torus,\ \forall k \in \{-K, \ldots, K\}, \qquad c_0(\tau, k) = \sum_{\measSpikeIdx = 1}^L \Big(\alpha_\ell\window(t_\measSpikeIdx - \tau)e^{-2\pi i kt_\measSpikeIdx} + \beta_\measSpikeIdx \window'(t_\measSpikeIdx - \tau)e^{-2\pi ikt_\measSpikeIdx}\Big)
\end{equation*}
and determining $\alpha \triangleq \{\alpha_\measSpikeIdx\}_{\measSpikeIdx = 1}^L$ and $\beta \triangleq \{\beta_\measSpikeIdx\}_{\measSpikeIdx = 1}^L$ such that the uniqueness conditions \eqref{eq: interpolation} and \eqref{eq: constraint interpolation} are met.
It turns out, however, that a more direct path is possible, namely by choosing a function $c_0 \in L^\infty(\torus \times \Z)$ of slightly different form and then reducing to a case already treated in the proof of Theorem~\ref{thm: exact recovery}; this approach leads to a substantially shorter proof.
We start by defining this function $c_0 \in L^\infty(\torus \times \Z)$ as
\begin{equation*}
	\forall \tau \in \torus,\ \forall k \in \{-K, \ldots, K\}, \quad c_0(\tau, k) \triangleq \frac{1}{2K+1}\sum_{\ell = 1}^L \Big(\alpha_\ell p(\tau - t_\ell)e^{-2\pi i kt_\ell} + \beta_\ell q(\tau - t_\ell) e^{-2\pi ikt_\ell}\Big),
\end{equation*}
where $p \colon \torus \rightarrow \C$  and $q \colon \torus \rightarrow \C$ are defined (for reasons that will become clear later) as
\begin{equation*}
	p(\tau) \triangleq \sum_{n \in \Z} p_n e^{2\pi in\tau} \qquad \text{and} \qquad  q(\tau) \triangleq \sum_{n \in \Z} q_n e^{2\pi in\tau}, 
\end{equation*}
for $\tau \in \torus$, with 
\begin{align*}
	p_n &\triangleq \sqrt{2\sigma}\exp\left(-2\pi\sigma^2 n^2\right)\int_{-1/2}^{1/2} \exp\left(-4\pi\sigma^2  u^2\right)\exp\left(-8\pi\sigma^2  nu\right)\mathrm{d}u \\
	q_n &\triangleq -2\pi i\sigma\sqrt{2\sigma}\exp\left(-2\pi\sigma^2 n^2\right)\int_{-1/2}^{1/2} (u+n)\exp\left(-4\pi\sigma^2  u^2\right)\exp\left(-8\pi\sigma^2  nu\right)\mathrm{d}u,
\end{align*}
for $n \in \Z$.
We first verify that the resulting function $c_0$ is, indeed, in $L^\infty(\torus \times \Z)$. This is accomplished by showing that the functions $p$ and $q$ are well-defined and are in $L^\infty(\torus)$, that is, by verifying that $\sum_{n \in \Z} \abs{p_n} < \infty$ and $\sum_{n \in \Z} \abs{q_n} < \infty$. Indeed, we have
\begin{align}
	\sum_{n \in \Z} \abs{p_n} &= \sqrt{2\sigma}\sum_{n \in \Z}  \exp\left(-2\pi\sigma^2 n^2\right)\int_{-1/2}^{1/2} \exp\left(-4\pi\sigma^2  u^2\right)\exp\left(-8\pi\sigma^2  nu\right)\mathrm{d}u \notag \\
		&\leq \sqrt{2\sigma} \sum_{n \in \Z} \int_{-1/2}^{1/2} \exp\left(-2\pi\sigma^2(4u + n)n\right) \mathrm{d}u \label{eq: definition function c convergence}\\
		& = C\ + \sqrt{2\sigma}\sum_{\substack{n \in \Z \\ \abs{n} \geq 3}} \int_{-1/2}^{1/2} \exp\left(-2\pi\sigma^2(4u + n)n\right)\mathrm{d}u, \label{eq: definition function c convergence bis}
\end{align}
where \eqref{eq: definition function c convergence} follows from $\exp(-4\pi\sigma^2 u^2) \leq 1$, for all $u \in [-1/2, 1/2]$, and we set
\begin{equation*}
	C \triangleq \sqrt{2\sigma}\sum_{n = -2}^2 \int_{-1/2}^{1/2} \exp\left(-2\pi\sigma^2(4u + n)n\right) \mathrm{d}u < \infty.
\end{equation*}
To see that the sum  in~\eqref{eq: definition function c convergence bis} is finite, first note that for $n \geq 3$ and $u \in [-1/2, 1/2]$, we have
\begin{equation*}
	(4u + n)n = (4u + n)\!\abs{n} \geq (-2 + n)\!\abs{n} \geq \abs{n}.
\end{equation*}
Similarly, for $n \leq -3$, we get
\begin{equation*}
	(4u + n)n = -(4u + n)\!\abs{n} \geq -(2 + n)\!\abs{n} \geq \abs{n}.
\end{equation*}
It therefore follows that
\begin{align*}
	 \sum_{\substack{n \in \Z \\ \abs{n} \geq 3}} \int_{-1/2}^{1/2}\exp\left(-2\pi\sigma^2(4u + n)n\right) \mathrm{d}u &\leq  \sum_{\substack{n \in \Z \\ \abs{n} \geq 3}}  \int_{-1/2}^{1/2}\exp(-2\pi\sigma^2\!\abs{n}) \mathrm{d}u  \notag \\
	 	&=  2\sum_{n =  3}^\infty \exp(-2\pi\sigma^2n) < \infty.
\end{align*}
This concludes the proof of $\sum_{n \in \Z} \abs{p_n} < \infty$. Similar reasoning shows that $\sum_{n \in \Z} \abs{q_n} < \infty$.
For $t \in \torus$, we then have
\begin{align}
	(\gaborOp^*c_0)(t) &= \sum_{k = -K}^K \int_{-1/2}^{1/2} c_0(\tau, k)\window(t - \tau)e^{2\pi i kt}\mathrm{d}\tau \label{eq: dual function Agc0}\\
		&= \frac{1}{2K+1}\sum_{\ell = 1}^L \left[ \alpha_\ell \left(\int_{-1/2}^{1/2} p(\tau - t_\ell)\window(t - \tau)\mathrm{d}\tau\right) D_K(t - t_\measSpikeIdx) \right. \notag \\
			&\hspace{2cm} \left. + \beta_\ell \left(\int_{-1/2}^{1/2} q(\tau - t_\ell)\window(t - \tau)\mathrm{d}\tau\right) D_K(t - t_\measSpikeIdx) \right] \notag \\
		&= \frac{1}{2K+1}\sum_{\ell = 1}^L \Big(\alpha_\ell P(t - t_\ell)D_K(t - t_\measSpikeIdx)+ \beta_\ell Q(t - t_\ell)D_K(t - t_\measSpikeIdx)\Big), \notag
\end{align}
where $D_K$ is the Dirichlet kernel, that is,
\begin{equation*}
	\forall t \in \torus, \quad D_K(t) \triangleq \frac{\sin((2K+1)\pi t)}{\sin(\pi t)},
\end{equation*}
and $P$ and $Q$ designate the cross-correlation between the functions $p$ and $\window$, and $q$ and $\window$, respectively, that is,
\begin{align*}
	\forall t \in \torus, \quad P(t) &\triangleq \int_{-1/2}^{1/2} p(\tau)\window(t - \tau)\mathrm{d}\tau \\
	\forall t \in \torus, \quad Q(t) &\triangleq \int_{-1/2}^{1/2} q(\tau)\window(t - \tau)\mathrm{d}\tau.
\end{align*}
Note that since $\window$  and $\tau \mapsto c_0(\tau, k)$, $k \in \{-K, \ldots, K\}$, are all $1$-periodic, we can integrate over the interval $[-1/2, 1/2]$ in \eqref{eq: dual function Agc0} (instead of $[0, 1]$ as done in \eqref{eq: adjoint operator expression}) and in the remainder of the proof. 
We next derive an alternative expression for the function $P$.
As in \eqref{eq: window torus}, we have
\begin{equation*}
	\forall t \in \torus, \quad \window(t) = \sum_{n \in \Z} \window_n e^{2\pi i nt},
\end{equation*}
where $\window_n \triangleq \sqrt{2\sigma} \exp(-2\pi \sigma^2n^2)$, for all $n \in \Z$. The $n$th Fourier series coefficient of $P$ is then given by $p_n\window_n$, and we show that the Fourier series $\sum_{n \in \Z} p_n\window_n e^{2\pi int}$ converges to $P(t)$ for all $t \in \torus$ using Dirichlet's theorem~\cite[Thm.~2.1]{Follard1992}, whose applicability conditions we verify next.
Since $\sum_{n \in \Z} \abs{p_n} < \infty$, $\sum_{n \in \Z} \abs{\window_n} < \infty$, and $\abs{e^{2\pi int}} = 1$, for all $t \in \torus$, by the Weierstrass M-test, the series $\sum_{n \in \Z} p_ne^{2\pi int}$ and $\sum_{n \in \Z} g_ne^{2\pi int}$ converge absolutely and uniformly. This implies that the functions $p$ and $\window$ are both continuous on $\torus$. Moreover, $\window$ is continuously differentiable on $\R$ as $\sum_{n \in \Z} \abs{n\window_n} < \infty$. As a result, the function $P$ is continuously differentiable on $\R$, and by application of Dirichlet's theorem, it follows that
\begin{equation*}
	\forall t \in \torus, \quad P(t) = \sum_{n \in \Z} \window_n p_n e^{2\pi int}.
\end{equation*}
For $n \in \Z$, we have
\begin{align*}
	\window_n p_n &= \sqrt{2\sigma}\exp\left(-2\pi\sigma^2n^2\right)\sqrt{2\sigma}\exp\left(-2\pi\sigma^2 n^2\right)\int_{-1/2}^{1/2} \exp\left(-4\pi\sigma^2  u^2\right)\exp\left(-8\pi\sigma^2  nu\right)\mathrm{d}u \\
		&= 2\sigma \exp\left(-4\pi\sigma^2 n^2\right)\int_{-1/2}^{1/2} \exp\left(-4\pi\sigma^2 u^2\right)\exp\left(-8\pi\sigma^2 nu\right)\mathrm{d}u \\
		&= 2\sigma \int_{-1/2}^{1/2} \exp\left(-4\pi\sigma^2(u + n)^2\right) \mathrm{d}u.
\end{align*}
Now fix $t \in [0, 1)$. If $t = 0$, we have
\begin{align*}
	P(t) = P(0) &= \sum_{n \in \Z} \window_n p_n = 2\sigma\sum_{n \in \Z} \int_{-1/2}^{1/2} \exp\left(-4\pi\sigma^2(u+n)^2\right) \mathrm{d}u \\
		&= 2\sigma\sum_{n \in \Z} \int_{n-1/2}^{n+1/2} \exp(-4\pi\sigma^2v^2)\mathrm{d}v  = 2\sigma\int_{-\infty}^\infty \exp(-4\pi\sigma^2v^2)\mathrm{d}v = 1,
\end{align*}
and if $t \neq 0$, we get
\begin{align}
	P(t) &= \sum_{n \in \Z} \left(2\sigma\int_{-1/2}^{1/2} \exp\left(-4\pi\sigma^2(u+n)^2\right) \mathrm{d}u\right) e^{2\pi int} \notag \\
		&= \int_{-1/2}^{1/2} 2\sigma\sum_{n \in \Z}  \exp\left(-4\pi\sigma^2(u+n)^2\right) e^{2\pi int} \mathrm{d}u \label{eq: licit integration term by term torus} \\
		&= \int_{-1/2}^{1/2} 2\sigma\sum_{n \in \Z}  \exp\left(-4\pi\sigma^2(u+n)^2\right) e^{2\pi i(u+n)t} e^{-2\pi i tu}\mathrm{d}u \\
		&= \int_{-1/2}^{1/2} \psi(u)\varphi(-u)\mathrm{d}u = \xi(0).
\end{align}
Here, $\varphi$ is the $1$-periodic function defined by
$\varphi(u) \triangleq e^{2\pi it u}$, $u \in [-1/2, 1/2)$, and $\psi$ and $\xi$ are given by
\begin{equation*}
	\begin{array}{lll}
		\forall u \in \torus, \qquad &\psi(u) &\triangleq \displaystyle 2\sigma\sum_{n \in \Z} \exp\left(-4\pi\sigma^2(u+n)^2\right) e^{2\pi it(u+n)} \\
		\forall x \in \torus, \qquad &\xi(x) &\triangleq \displaystyle \int_{-1/2}^{1/2} \varphi(u)\psi(x-u)\mathrm{d}u.
	\end{array}
\end{equation*}
The order of summation and integration in~\eqref{eq: licit integration term by term torus} is interchangeable thanks to
\begin{align*}
	\sum_{n \in \Z} \int_{-1/2}^{1/2} \exp\left(-4\pi\sigma^2(u+n)^2\right) \mathrm{d}u &= \sum_{n \in \Z} \int_{n-1/2}^{n+1/2} \exp\left(-4\pi\sigma^2 v^2\right)\mathrm{d}v \\
		&= \int_{-\infty}^\infty \exp\left(-4\pi\sigma^2 v^2\right) \mathrm{d}v < \infty.
\end{align*}
The function $\xi$ can be expanded into a Fourier series. Specifically, it holds that
\begin{equation*}
	\forall x \in \torus, \quad \xi(x) = \sum_{n \in \Z} \varphi_n\psi_n e^{2\pi inx}, 
\end{equation*}
where $\varphi_n$ and $\psi_n$ denote the $n$th Fourier series coefficients of $\varphi$ and $\psi$, respectively.
We have
\begin{equation*}
	\varphi_n = \int_{-1/2}^{1/2} \varphi(x)e^{-2\pi i nx}\mathrm{d}x = \int_{-1/2}^{1/2} e^{2\pi i(t-n)x}\mathrm{d}x = \frac{(-1)^n\sin(\pi t)}{\pi(t-n)}
\end{equation*}
and
\begin{align*}
	\psi_n = \int_{-1/2}^{1/2} \psi(x)e^{-2\pi inx} \mathrm{d}x &= \int_{-1/2}^{1/2} 2\sigma\sum_{m \in \Z} \exp\left(-4\pi\sigma^2(x+m)^2\right)e^{2\pi it(x+m)} e^{-2\pi inx}\mathrm{d}x \\
		&= 2\sigma \sum_{m \in \Z} \int_{m-1/2}^{m+1/2} \exp\left(-4\pi\sigma^2v^2\right) e^{2\pi itv} e^{-2\pi inv}\mathrm{d}v \\
		&=2\sigma \int_{-\infty}^\infty \exp(-4\pi\sigma^2v^2)e^{2\pi i(t-n)v}\mathrm{d}v \\
		&= \exp\left(-\frac{\pi(t-n)^2}{4\sigma^2}\right),
\end{align*}
for $n \in \Z$. It follows that
\begin{equation*}
	\forall t \in \torus, \quad P(t) = \xi(0) = \sum_{n \in \Z} \varphi_n\psi_n = \sum_{n \in \Z} \frac{(-1)^n\sin(\pi t)}{\pi(t-n)}\exp\left(-\frac{\pi(t-n)^2}{4\sigma^2}\right).
\end{equation*}
We then get
\begin{align*}	
	\forall t \in \torus, \quad P(t)D_K(t) &= \sum_{n \in \Z} \frac{(-1)^n\sin((2K+1)\pi t)}{\pi(t - n)}\exp\left(-\frac{\pi(t-n)^2}{4\sigma^2}\right) \\
		&= (2K+1)\sum_{n \in \Z} \sinc((2K+1)\pi (t-n))\autocorrelation(t-n),
\end{align*}
where $\autocorrelation$ was defined in~\eqref{eq: autocorrelation function}. Similarly, we can show that 
\begin{equation*}	
	\forall t \in \torus, \quad Q(t)D_K(t) = (2K+1)\sum_{n \in \Z} \sinc((2K+1)\pi (t-n))\autocorrelation'(n-t).
\end{equation*}
This finally yields
\begin{align*}
	\forall t \in \torus, \quad (\gaborOp^*c_0)(t) = \sum_{\measSpikeIdx = 1}^L\sum_{n \in \Z} & \Big(\alpha_\measSpikeIdx \sinc((2K+1)\pi(t - t_\measSpikeIdx -n))\autocorrelation(t - t_\measSpikeIdx - n) \\
		&\hspace{1cm}+ \beta_\measSpikeIdx \sinc((2K+1)\pi(t - t_\measSpikeIdx -n))\autocorrelation'(n - t + t_\measSpikeIdx)\Big) \\
		=\sum_{\measSpikeIdx = 1}^L\sum_{n \in \Z} & \Big(\alpha_\measSpikeIdx u(t - t_\measSpikeIdx -n) + \beta_\measSpikeIdx v(t - t_\measSpikeIdx -n)\Big),
\end{align*}
where we set 
\begin{align}
	\forall t \in \R, \quad u(t) &\triangleq \autocorrelation(t)\sinc(2\pi\freq_c' t) \label{eq: definition function u for torus} \\
	\forall t \in \R, \quad v(t) &\triangleq \autocorrelation'(-t)\sinc(2\pi\freq_c' t) = \frac{\pi t}{2\sigma^2}\autocorrelation(t)\sinc(2\pi\freq_c' t) \label{eq: definition function v for torus}
\end{align}
as in \eqref{eq: definition function u} and \eqref{eq: definition function v} with $\freq_c' \triangleq K+1/2$. Analogously to the proof of Theorem~\ref{thm: exact recovery} we can define the operators
\begin{equation*}
	\begin{array}{llll}
		\mathcal{U}_p \colon & \C^L & \longrightarrow & \C^L \\
						 & \alpha = \{\alpha_\measSpikeIdx\}_{\measSpikeIdx =1}^L & \longmapsto & \left\{\sum\limits_{\measSpikeIdxBis = 1}^L\sum\limits_{n \in \Z} \alpha_\measSpikeIdxBis u^{(p)}(t_\measSpikeIdx - t_\measSpikeIdxBis - n)\right\}_{\measSpikeIdx = 1}^L
	\end{array}
\end{equation*}
and
\begin{equation*}
	\begin{array}{llll}
		\mathcal{V}_p \colon & \C^L & \longrightarrow & \C^L \\
						 & \beta = \{\beta_\measSpikeIdx\}_{\measSpikeIdx \in \measSupportIdx} & \longmapsto & \left\{\sum\limits_{\measSpikeIdxBis = 1}^L\sum_{n \in \Z} \beta_\measSpikeIdxBis v^{(p)}(t_\measSpikeIdx - t_\measSpikeIdxBis - n)\right\}_{\measSpikeIdx = 1}^L,
	\end{array}
\end{equation*}
where $p \in \{0, 1\}$. Then, given $\varepsilon = \{\varepsilon_\measSpikeIdx\}_{\measSpikeIdx = 1}^L$ with $\abs{\varepsilon_\measSpikeIdx} = 1$, $\measSpikeIdx \in \{1, 2, \ldots, L\}$, we can solve the equation system
\begin{equation}
	\left\{\begin{array}{l}
		\mathcal{U}_0\alpha + \mathcal{V}_0\beta =  \varepsilon\\
		\mathcal{U}_1\alpha + \mathcal{V}_1\beta = 0
	\end{array}\right.
	\label{eq: inversion operator condition 2 torus}
\end{equation}
to determine $\alpha \in \C^L$ and $\beta \in \C^L$ such that the interpolation conditions $(\gaborOp^*c_0)(t_\measSpikeIdx) = \varepsilon_\measSpikeIdx$, for all $\measSpikeIdx \in \{1, 2, \ldots, L\}$, are satisfied and $\gaborOp^*c_0$ has a local extremum at every $t_\measSpikeIdx$, $\measSpikeIdx \in \{1, 2, \ldots, L\}$.
As in the proof of Theorem~\ref{thm: exact recovery}, if the operators $\mathcal{V}_1$ and $\mathcal{U}_0 - \mathcal{V}_0\mathcal{V}_1^{-1}\mathcal{U}_1$ are invertible, then one can choose 
$\alpha = (\mathcal{U}_0 - \mathcal{V}_0\mathcal{V}_1^{-1}\mathcal{U}_1)^{-1}\varepsilon$ and  $\beta = -\mathcal{V}_1^{-1}\mathcal{U}_1\alpha$ to satisfy \eqref{eq: inversion operator condition 2 torus}. 
Proving the invertibility of $\mathcal{V}_1$ and $\mathcal{U}_0 - \mathcal{V}_0\mathcal{V}_1^{-1}\mathcal{U}_1$ is essentially identical to the corresponding part in the proof of Theorem~\ref{thm: exact recovery} with $\freq_c$ replaced by $\freq_c'$. Verifying that $\abs{(\gaborOp^*c_0)(t)} < 1$ for all $t \in \torus \setminus \measSupport$, where $\measSupport = \{t_\measSpikeIdx\}_{\measSpikeIdx = 1}^L$, is also done in a fashion similar to the proof of Theorem~\ref{thm: exact recovery} (see Section \ref{subsection: show that the magnitude is strictly smaller than one outside the support}).